\documentclass[11pt]{article}
\usepackage{textcomp}
\usepackage[english]{babel}
\usepackage[utf8]{inputenc}
\usepackage{amsmath,amssymb,amsfonts,amsthm,titling}
\usepackage{tikz}
\usepackage{mathtools}
\usepackage{graphicx}
\usepackage{subcaption}
\usepackage{pdfpages}
\usepackage{multicol}
\usepackage[colorinlistoftodos]{todonotes}
\usepackage[margin=1in]{geometry}
\usepackage{setspace}
\graphicspath{{img/}}
\usepackage{float}
\usepackage[margin = 1in]{caption}
\usepackage{fancyhdr}
\usepackage{float}
\usepackage{url}
\usepackage[square, numbers]{natbib}
\usepackage{multirow, siunitx, booktabs,array}
\usepackage[font=scriptsize,labelfont=bf]{caption}
\allowdisplaybreaks
\makeatletter 
\def\@eqnnum{{\normalsize \normalcolor (\theequation)}} 
\makeatother

\begin{document}
\begin{titlepage}

\newcommand{\HRule}{\rule{\linewidth}{0.6mm}}
\newcommand{\HRules}{\rule{\linewidth}{0.3mm}} 
\center 

\textsc{\LARGE Partially Specified Spatial Autoregressive Model with Artificial Neural Network}\\[1.5cm] 

 
{Wenqian Wang}\\[0.3cm]
{Beth Andrews}\\[1cm]

{Department of Statistics}\\[0.3cm]
{Northwestern University}\\[1cm]

\begin{abstract}
The spatial autoregressive model has been widely applied in science, in areas such as economics, public finance, political science, agricultural economics, environmental studies and transportation analyses. The classical spatial autoregressive model is a linear model for describing spatial correlation. In this work, we expand the classical model to include related exogenous variables, possibly non-Gaussian, high volatility errors, and a nonlinear neural network component. The nonlinear neural network component allows for more model flexibility — the ability to learn and model nonlinear and complex relationships.
We use a maximum likelihood approach for model parameter estimation. We establish consistency and asymptotic normality for these estimators under some standard conditions on the spatial model and neural network component. We investigate the quality of the asymptotic approximations for finite samples by means of numerical simulation studies. For illustration, we include a real world application. 

\end{abstract}

\vfill 

\end{titlepage}



\section{Introduction}
One commonly used assumption in regression analysis is that observations are uncorrelated, but this assumption is sometimes impossible to be defended in the analysis of spatial data when one observation may be related to neighboring entities. The nature of the covariance among observations may not be known precisely and researchers have been dedicated for years to building appropriate models to describe such correlation. 
The collection of techniques to investigate properties in the spatial models is considered to have begun in the domain of spatial econometrics first proposed by Paelinck in the early 1970s \cite{paelinck1978spatial}. Later, the books by Cliff and Ord \cite{getis1995cliff}, Anselin \cite{anselin2013spatial} and Cressie \cite{cressie2015statistics} detailed research results related to spatial autocorrelation, purely spatial dependence as well as cross-sectional and/or panel data. 

So why has estimating the spatial correlation drawn so much attention? In some applications estimating the spatial structure of the dependence may be a subject of interest or provide a key insight; in other contexts, it may be regarded as serial correlations. However, in either case, inappropriate treatment of data with spatial dependence can lead to inefficient or biased and inconsistent estimates. These consequences may result in misleading conclusions in the analysis of real world problems.
Therefore, it is important to describe spatial dependence; some standard parametric models are spatial autoregressive models (SAR), spatial error models (SEM) and spatial Durbin models (LeSage, R. Pace,  \cite{lesage2009introduction}). According to the spatial autoregressive model, values of the dependent variable are linearly related to observations  in neighboring regions. The SAR model has been widely discussed in the literature, and researchers have proposed various parameter estimation techniques such as the method of maximum likelihood by Ord \cite{ord1975estimation} and Smirnov and Anselin \cite{smirnov2001fast}, the method of moments by Kelejian and Prucha \cite{kelejian1998generalized,kelejian2010specification,kelejian1999generalized} and the method of quasi-maximum likelihood estimation by Lee \cite{lee2004asymptotic}.  

In a SAR model with covariates, the observations are modeled as a weighted average of neighboring observations with weights determined by the distance between them plus a function of the covariates:
{\small\setlength{\abovedisplayskip}{3pt}
  \setlength{\belowdisplayskip}{\abovedisplayskip}
  \setlength{\abovedisplayshortskip}{0pt}
  \setlength{\belowdisplayshortskip}{3pt}
\begin{equation*}
y_s=\rho\sum_{i=1}^{n}w_{si}y_i + x_s^{\prime}\beta +\varepsilon_s \quad s=1,2,\ldots,n
\end{equation*}} %
\vspace{-20pt}

\noindent
where $y_s$ denotes the observation of interest and $x_s$ denotes the value of a $p$ dimensional independent variable at location $s \in \{1,2,\ldots,n\}$. $w_{ij}$ is the $(i,j)$ entry of a $n \times n$ weight matrix $W_n$; it is a nonnegative weight which measures the degree of interaction between units $i$ and $j$. By convention, we always set $w_{ii}=0$. The random disturbances $\{\varepsilon_s\}_{s=1}^n$ are uncorrelated with zero means and equal variances; often in the literature these are taken to be normally distributed. The model has parameter vector $\theta =(\rho, \beta^{\prime})$. However,  parametric models are vulnerable to the preciseness of model specification: a misspecified model can draw misleading inferences. Whereas a nonparametric model is more robust even though it sacrifices the precision. In this sense, to combine the advantages of these two models, we consider a semi-parametric model in the spatial context. The suggested model, a partially specified spatial autoregressive model (PSAR) \cite{su2010profile}, is defined as follows:  
{\small
  \setlength{\abovedisplayskip}{3pt}
  \setlength{\belowdisplayskip}{\abovedisplayskip}
  \setlength{\abovedisplayshortskip}{0pt}
\begin{equation}\label{psar}
y_s=\rho\sum_{i=1}^{n}w_{si}y_i + x_s^{\prime}\beta +g(z_s)+\varepsilon_s \quad s=1,2,\ldots ,n
\end{equation}} %
\vspace{-20pt}

\noindent
where $g(\cdot)$ is an unknown function and $z_s$ denotes a $q$ dimensional vector of explanatory variables at location $s$. This PSAR model has a more flexible functional form than the ordinary spatial autoregressive model. Methods of parameter estimation for the PSAR model include profile quasi-maximum likelihood estimation by Su and Jin \cite{su2010profile} and a sieve method by Zhang \cite{zhang2015statistical}. In Su and Jin \cite{su2010profile}, they used profile quasi-maximum likelihood estimators for independent and identically distributed errors and gave an asymptotic analysis using local polynomials to describe $g$. This method showed its advantage in dimension reduction when maximizing concentrated likelihood function with respect to one parameter $\rho$ but involved in two-stage maximization if we wanted to obtain other parameter estimators such as $\beta$'s. However, in Zhang \cite{zhang2015statistical}, they were using a sieve method (Ai, Chen \cite{ai2003efficient}) to approximate the nonparametric function. They applied a sequence of known basis functions to approximate $g(\cdot)$ in equation (\ref{psar}), and used the two-stage least squares estimation with some instrumental variables to obtain consistent estimators for the PSAR model. 

Both methods use Gaussian likelihood techniques. But normality is unreasonable in many cases when we observe errors with heavy tails or abnormal patterns. If this is the case, maximum likelihood estimation can be more efficient than Gaussian-based quasi-maximum likelihood estimation. 
Another difference is that we are using neural network models to estimate the nonlinear function $g(\cdot)$ whereas Su and Jin \cite{su2010profile} applied a finite order of local polynomials about some explanatory variables and Zhang \cite{zhang2015statistical} used a linear combination of a sequence of known functions to estimate $g(\cdot)$.

The purpose of this paper is to extend an autoregressive artificial neural network model (Medeiros, Teräsvirta, Rech \cite{medeiros2006building}) developed in the context of time series data to a partially specified spatial autoregressive model and we regard the artificial neural network part as a nonlinear statistical component to approximate the nonparametric function $g(\cdot)$ in the PSAR model (\ref{psar}). The use of an ANN (Artificial Neural Network) model is motivated by mathematical results showing that under mild conditions, a relatively simple ANN model is competent in approximating any Borel-measurable function to any given degree of accuracy (see for example Hornik \emph{et al.} \cite{hornik1990universal}, Gallant and White \cite{gallant1992learning}). Under this theoretical foundation, we would expect our model to perform well when modeling nonparamatric components in spatial contexts.
Another improvement is that, in our model, the random error is independent and identically distributed but does not necessarily follow a normal distribution. We derive parameter estimates by maximizing the corresponding likelihood function and discuss the asymptotic properties of our estimators under conditions that the spatial weight matrix is nonsingular and the log likelihood function has some dominated function with a finite mean. 

In Sections 2 and 3, our model PSAR-ANN is given and a likelihood function for corresponding parameters is derived. In Sections 4 and 5, we discuss model identification and establish consistency and asymptotic normality for MLEs of model parameters. In section 6, we describe numerical simulation studies to investigate how well the behavior of estimators for finite samples matches the limiting theory, i.e., the quality of the normal approximation. In the real data example, we would like to explore spatial dynamics in U.S. presidential elections and a PSAR-ANN model is fit to the proportion of votes cast for 2004 U.S. presidential candidates at the county level.

\theoremstyle{definition}
\newtheorem{ass}{Assumption}
\newtheorem{thm}{Theorem}
\newtheorem{lem}{Lemma}
\newtheorem{defn}{Definition}

\section{Model Specification}
The main focus of this paper is to approximate the nonparametric function $g(\cdot)$ in the partially specified spatial autoregressive model (\ref{psar}) by an artificial neural network model. The model in matrix form is defined as
{\small\setlength{\abovedisplayskip}{3pt}
  \setlength{\belowdisplayskip}{\abovedisplayskip}
  \setlength{\abovedisplayshortskip}{0pt}
  \setlength{\belowdisplayshortskip}{3pt}
\begin{equation}\label{model-m}
Y_n= X_n\beta+\rho W_nY_n+ F(X_n\boldsymbol{\gamma}^{\prime})\lambda+\boldsymbol{\varepsilon}_n
\end{equation}} %
\vspace{-20pt}

\noindent
where $Y_n= \{y_s\}_{s=1}^n$ contains observations of the dependent variable at $n$ locations. The independent variable matrix $X_n=(x_1,x_2,\ldots, x_n)^{\prime} \in \mathbb{R}^{n\times q}$ contains values of exogenous regressors for the $n$ regions, where for each region, $x_s = (x_{s1},\ldots, x_{sq})^{\prime}$, $s=1,2,\ldots,n$, is a $q$ dimensional vector. $\boldsymbol{\varepsilon}_n =  \{\varepsilon_s\}_{s=1}^n$ denotes a vector of $n$ independent identically distributed random noises with density function $f(\cdot)$, mean $0$ and variance $\sigma^2=1$.

Exogenous parameters $\beta = (\beta_1,\ldots,\beta_q)^{\prime} \in \mathbb{R}^{q}$ and scalar $\rho$, the spatial autoregressive parameter, are assumed to be the same over all regions. $W_n=\{w_{ij}\} \in \mathbb{R}^{n \times n}$ denotes a spatial weight matrix which characterizes the connections between neighboring regions. For the ease of illustration, we define some additional notations. Given a function $f\in C^{1}(R^{1})$ continuous on $\mathbb{R}$, we define a new matrix mapping $R^{n}\rightarrow R^{n}$ as $\boldsymbol{f}$ s.t. $\boldsymbol{f}(x_1,\ldots,x_n)=(f(x_1),\ldots,f(x_n))^{\prime}$.
Using this notation, the artificial neural network component (Medeiros {\em et al.} \cite{medeiros2006building}) can be written as $ \boldsymbol{F}(X_n\boldsymbol{\gamma}^{\prime})\lambda$ with
{\small
\setlength{\abovedisplayskip}{3pt}
  \setlength{\belowdisplayskip}{\abovedisplayskip}
  \setlength{\abovedisplayshortskip}{0pt}
  \setlength{\belowdisplayshortskip}{3pt}
\begin{align*}
\boldsymbol{F}(X_n\boldsymbol{\gamma}^{\prime})=\begin{bmatrix}
F(x_1^{\prime}\boldsymbol{\gamma}_1)&
F(x_1^{\prime}\boldsymbol{\gamma}_2)&
\ldots&
F(x_1^{\prime}\boldsymbol{\gamma}_h)\\
 F(x_2^{\prime}\boldsymbol{\gamma}_1)&
F(x_2^{\prime}\boldsymbol{\gamma}_2)&\ldots&
F(x_2^{\prime}\boldsymbol{\gamma}_h)\\
\vdots&\vdots&\ddots&\vdots\\
F(x_n^{\prime}\boldsymbol{\gamma}_1)&
F(x_n^{\prime}\boldsymbol{\gamma}_2)&
\ldots&
F(x_n^{\prime}\boldsymbol{\gamma}_h)
\end{bmatrix}
\in \mathbb{R}^{n\times h}
\end{align*}
} %
\noindent
This matrix represents a single layer neural network with $h$ neurons for every location. The value of $h$ is determined by researchers and can be selected by comparing AIC/BIC. Under this setting, the parameter matrix $\boldsymbol{\gamma}=(\boldsymbol{\gamma}_1,\boldsymbol{\gamma}_2,\ldots, \boldsymbol{\gamma}_h)^{\prime} \in \mathbb{R}^{h \times q}$, $\boldsymbol{\gamma}_i=(\gamma_{i1},\ldots,\gamma_{iq})^{\prime}\in \mathbb{R}^{q}$, $i=1,2,\ldots,h$, contains all the weights in a neural network model. $F(\cdot)$ is called the activation function and we discuss the situation when it is the logistic function with range from 0 to 1 (the logistic activation function is the most common choice in neural network modeling \cite{medeiros2006building}). For given information $x_s$ at region $s$, the corresponding output of $i$th neuron in a single layer neural network is 
{\small
\setlength{\abovedisplayskip}{3pt}
  \setlength{\belowdisplayskip}{\abovedisplayskip}
  \setlength{\abovedisplayshortskip}{0pt}
  \setlength{\belowdisplayshortskip}{3pt}
\begin{equation*}
F(x_s^{\prime}\boldsymbol{\gamma}_i)=(1+e^{-x_s^{\prime}\boldsymbol{\gamma}_i})^{-1}, \quad s=1,2,\ldots,n,\,i=1,2,\ldots,h
\end{equation*}
} %
\vspace{-20pt}

\noindent
Parameter vector $\lambda=(\lambda_1,\lambda_2,\ldots,\lambda_h)^{\prime}$ denotes weights for $h$ neurons. So $ \boldsymbol{F}(X_n\boldsymbol{\gamma}^{\prime})\lambda=$
{\small
\setlength{\abovedisplayskip}{3pt}
  \setlength{\belowdisplayskip}{\abovedisplayskip}
  \setlength{\abovedisplayshortskip}{0pt}
  \setlength{\belowdisplayshortskip}{3pt}\[
\begin{bmatrix}
F(x_1^{\prime}\boldsymbol{\gamma}_1)&
F(x_1^{\prime}\boldsymbol{\gamma}_2)&
\ldots&
F(x_1^{\prime}\boldsymbol{\gamma}_h)\\
F(x_2^{\prime}\boldsymbol{\gamma}_1)&
F(x_2^{\prime}\boldsymbol{\gamma}_2)&\ldots&
F(x_2^{\prime}\boldsymbol{\gamma}_h)\\
\vdots&\vdots&\ddots&\vdots\\
F(x_n^{\prime}\boldsymbol{\gamma}_1)&
F(x_n^{\prime}\boldsymbol{\gamma}_2)&
\ldots&
F(x_n^{\prime}\boldsymbol{\gamma}_h)
\end{bmatrix}
\begin{bmatrix}
\lambda_1\\
\lambda_2\\
\vdots\\
\lambda_h
\end{bmatrix}=
\begin{bmatrix}
\sum_{i=1}^{h}\lambda_i F(x_1^{\prime}\boldsymbol{\gamma}_i)\\
\sum_{i=1}^{h}\lambda_i F(x_2^{\prime}\boldsymbol{\gamma}_i)\\
\vdots\\
\sum_{i=1}^{h}\lambda_i F(x_n^{\prime}\boldsymbol{\gamma}_i)
\end{bmatrix}
\in \mathbb{R}^{n}
\]} %
\noindent
One important element in the model (\ref{model-m}) is the spatial weight matrix $W_n$.
The spatial weights depend on the definition of a neighborhood set for each observation. In our applications we begin by using a square symmetric $n\times n$ matrix with $(i,j)$ element equal to 1 if regions $i$ and $j$ are neighbors and $w_{ij}=0$ otherwise. The diagonal elements of the spatial neighbors matrix are set to zero. Then we row standardize the weight matrix, so the nonzero weights are scaled so that the weights in each row sum up to 1. In convention, people usually use the row standardized weight matrices because row standardization creates proportional weights in cases where features have an unequal number of neighbors; also this normalized matrix has nice properties in the range of eigenvalues (this will be mentioned later). As LeSage \cite{lesage1999spatial} suggests, there is a vast number of ways to define neighbors and to construct a weight matrix. In the following we discuss some commonly used methods in lattice cases and non-lattice cases. In a lattice case shown in the following Figure \ref{weights-1}, we have 9 locations and we label them as $1,2,\ldots,9$ at left bottom corners in each cell. Suppose $i$ is the target location and $j$ identifies a neighbor of $i$.
\begin{itemize}
\item Rook Contiguity (Fig 1 (a)): two regions are neighbors if they share (part of) a common edge (on any side) 
\item Bishop Continuity (Fig 1 (b)): two regions are spatial neighbors if they share a common vertex (or a point)
\item Queen Contiguity (Fig 1 (c)): this is the union of Rook and Bishop contiguity. Two regions are neighbors in this sense if they share any common edge or vertex
\end{itemize}
\begin{figure}[ht]
\begin{center}
\begin{subfigure}[b]{0.3\textwidth}
\begin{tikzpicture}
\draw[step=1.4cm,color=black] (0,0) grid (4.2,4.2);
\node at (0.2,0.2) {7};\node at (1.6,0.2) {8};\node at (3,0.2) {9};\node at (0.2,1.6) {4};\node at (1.6,1.6) {5};
\node at (3,1.6) {6};\node at (0.2,3) {1};\node at (1.6,3) {2};\node at (3,3) {3};
\node at (2.1,0.7) {$j$};
\node at (0.7,2.1) {$j$};
\node at (2.1,2.1) {$i$};
\node at (3.5,2.1) {$j$};
\node at (2.1,3.5) {$j$};
\end{tikzpicture}
\caption{}
\end{subfigure}
\begin{subfigure}[b]{0.3\textwidth}
\begin{tikzpicture}
\draw[step=1.4cm,color=black] (0,0) grid (4.2,4.2);
\node at (0.2,0.2) {7};\node at (1.6,0.2) {8};\node at (3,0.2) {9};\node at (0.2,1.6) {4};\node at (1.6,1.6) {5};
\node at (3,1.6) {6};\node at (0.2,3) {1};\node at (1.6,3) {2};\node at (3,3) {3};
\node at (0.7,0.7) {$j$};
\node at (3.5,0.7) {$j$};
\node at (2.1,2.1) {$i$};
\node at (0.7,3.5) {$j$};
\node at (3.5,3.5) {$j$};
\end{tikzpicture}
\caption{}
\end{subfigure}
\begin{subfigure}[b]{0.3\textwidth}
\begin{tikzpicture}
\draw[step=1.4cm,color=black] (0,0) grid (4.2,4.2);
\node at (0.2,0.2) {7};\node at (1.6,0.2) {8};\node at (3,0.2) {9};\node at (0.2,1.6) {4};\node at (1.6,1.6) {5};
\node at (3,1.6) {6};\node at (0.2,3) {1};\node at (1.6,3) {2};\node at (3,3) {3};
\node at (0.7,0.7) {$j$};
\node at (2.1,0.7) {$j$};
\node at (3.5,0.7) {$j$};
\node at (0.7,2.1) {$j$};
\node at (2.1,2.1) {$i$};
\node at (3.5,2.1) {$j$};
\node at (0.7,3.5) {$j$};
\node at (2.1,3.5) {$j$};
\node at (3.5,3.5) {$j$};
\end{tikzpicture}
\caption{}
\end{subfigure}
\end{center}
\caption{Examples of Rook (a), Bishop (b) and Queen Contiguity (c)}
\label{weights-1}
\end{figure}
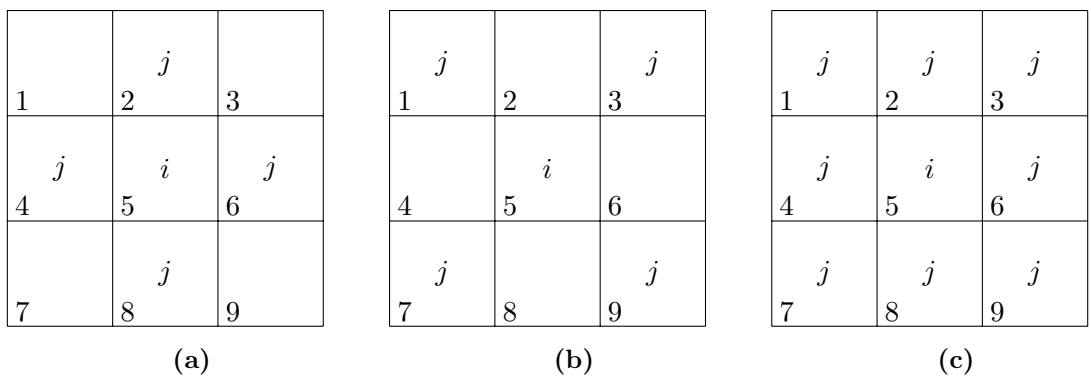
In practice, we may not always have a problem in a lattice. So an analog of an edge and a vertex is called ``snap distance"\cite{bivand2008applied} such that any border larger than this ``snap distance" will be regarded as an edge or otherwise a vertex. So the Queen contiguity may be interpreted as that two regions are neighbors as long as they are connected no matter how short the common border is. Under the Queen criterion, for example, based on the example illustrated in Figure (\ref{weights-1}(c)), a $9\times 9$ weight matrix for nine locations is shown below. 
{\tiny
\setlength{\abovedisplayskip}{3pt}
  \setlength{\belowdisplayskip}{\abovedisplayskip}
  \setlength{\abovedisplayshortskip}{0pt}
  \setlength{\belowdisplayshortskip}{0pt}
  \begin{align}\label{9by9-matrix}
  	\begin{pmatrix}
  	0& 1& 0& 1& 1& 0& 0& 0& 0\\
  	1& 0& 1& 1& 1& 1& 0& 0& 0\\
  	0& 1& 0& 0& 1& 1& 0& 0& 0\\
  	1& 1& 0& 0& 1& 0& 1& 1& 0\\
  	1& 1& 1& 1& 0& 1& 1& 1& 1\\
  	0& 1& 1& 0& 1& 0& 0& 1& 1\\
  	0& 0& 0& 1& 1& 0& 0& 1& 0\\
  	0& 0& 0& 1& 1& 1& 1& 0& 1\\
  	0& 0& 0& 0& 1& 1& 0& 1& 0    
  	\end{pmatrix}
  \end{align}
} %
However, in a non-lattice case when units, such as cities, are only points, this neighborhood definition does not work because all units/points do not share any common edge or vertex. So a distance based method is utilized to deal with such point case. Denote $d_{ij} \equiv d(i,j)$ as the distance between two units/points $i$ and $j$, then some commonly used ways to define neighborhoods are  
\begin{itemize}
\item Minimum Distance Neighbors:\\
A neighbor $j$ of unit $i$ satisfies that their distance $d_{ij} \in \Big(0, \displaystyle\max_{i=\{1,\ldots,n\}} \min_{j \neq i} d(i, j)\Big]$. This method controls that every unit has at least one neighbor but usually includes a large number of irrelevant connections.  
\item K-nearest Neighbors:\\
Neighbors of $i$ are restricted by the user-defined parameter $K$. A unit $j$ is a neighbor of $i$ if $j \in N_K(i)$, where $N_K(i)$ defines the K-nearest neighbors of $i$. This method also guarantees that there is no neighborless unit and has less noise then the Minimum Distance Neighbors. However, the user-choice parameter $K$ may not reflect the true level of connectedness or isolation between points.
\item Sphere of Influence Neighbors:\\
For each point $i \in S = \{1,\ldots,n\}$, $r_i= \min_{k\neq i}d(i,k)$ and denote $C_i$ as a circle of radius $r_i$ centered at $i$. Units $i$ and $j$ are neighborhoods whenever $C_i$ and $C_j$ intersect in exactly two points. This graph-based method improves the K-nearest Neighbors in a way that relatively long links are avoided and the number of connections per unit is variable. This method works well even with irregularly located areal entities and precludes the intervention of user-defined parameter $K$ in the previous method (See Figure \ref{weights-2}).  
\end{itemize}
\begin{figure}[ht]
\centering
\includegraphics[scale=0.3]{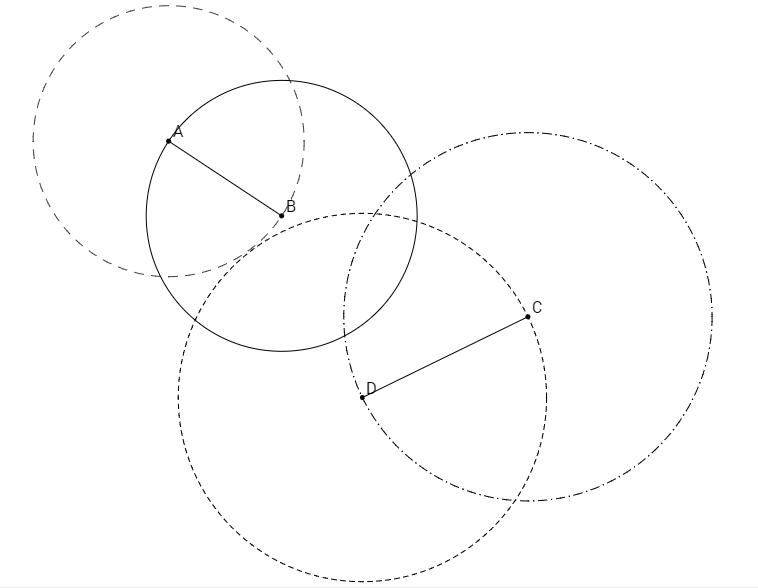}
\caption{\textbf{Sphere of Influence Graph:} A,B,C,D represent four units. Where the circles around each city overlap in at least two points, the cities can be considered neighbors. In the current example, A is a neighbor of only B, B is a neighbor to all, C is a neighbor of B and D, D is a neighbor of B and C but not A.}
\label{weights-2}
\end{figure}
According to Figure \ref{weights-2}, the weight matrix for A, B, C and D is:
{\tiny
\setlength{\abovedisplayskip}{3pt}
  \setlength{\belowdisplayskip}{\abovedisplayskip}
  \setlength{\abovedisplayshortskip}{0pt}
  \setlength{\belowdisplayshortskip}{3pt}\[
\left(
\begin{array}{cccc}
 0& 1& 0& 0\\
 1& 0& 1& 1\\
 0& 1& 0& 1\\
 0& 1& 1& 0 
\end{array}
\right)
\]} %

To write our model (\ref{model-m}) more explicitly, for each location $s$, $s=1,2,\ldots,n$
{\small\setlength{\abovedisplayskip}{3pt}
  \setlength{\belowdisplayskip}{\abovedisplayskip}
  \setlength{\abovedisplayshortskip}{0pt}
  \setlength{\belowdisplayshortskip}{3pt}
\begin{equation}\label{model}
y_s=x_s^{\prime}\boldsymbol{\beta}+\rho\sum_{i=1}^{n}w_{si}y_i +\sum_{i=1}^{h}\lambda_i F(x_s^{\prime}\boldsymbol{\gamma}_i)+\varepsilon_s
\end{equation} }%
\vspace{-20pt}

\noindent
The term $\sum_{i=1}^{h}\lambda_i F(x_s^{\prime}\boldsymbol{\gamma}_i)$, a linear combination of logistic functions with weights $\lambda_i, i = 1,2,\ldots, h$, forms a hidden layer of this neural network with $h$ neurons (Medeiros, Teräsvirta, Rech \cite{medeiros2006building}). This neural network helps discover nonlinear relationship between the response variable and its covariates.

\section{Likelihood Function}
Rewriting the equation in (\ref{model-m}), we have
{\small
\setlength{\abovedisplayskip}{3pt}
  \setlength{\belowdisplayskip}{\abovedisplayskip}
  \setlength{\abovedisplayshortskip}{0pt}
  \setlength{\belowdisplayshortskip}{3pt}
\begin{equation}\label{model-u}
(I_n-\rho W_n)Y_n- X_n\beta- \boldsymbol{F}(X_n\boldsymbol{\gamma}^{\prime})\lambda=\boldsymbol{\varepsilon}_n
\end{equation}} %
\vspace{-20pt}

\noindent
where $I_n$ is an $n\times n$ identity matrix. We denote $\boldsymbol{\theta}=(\beta_1,\ldots,\beta_q,\rho,\lambda_1,\ldots,\lambda_h,\boldsymbol{\gamma}_1^{\prime},\ldots,\boldsymbol{\gamma}_h^{\prime})^{\prime} \in \mathbb{R}^{(q+1)(h+1)}$ with true value $\boldsymbol{\theta}_0$.

For the analysis of identification and estimation of this spatial autoregressive model (\ref{model-m}), we adopt the following assumptions:
\begin{ass}\label{compact-space}
	The $(q+1)(h+1)$-dimensional parameter vector $\boldsymbol{\theta}=(\beta^{\prime},\rho,\lambda^{\prime},\boldsymbol{\gamma}_1^{\prime},\ldots,\boldsymbol{\gamma}_h^{\prime})^{\prime} \in \boldsymbol{\Theta}$, where $\boldsymbol{\Theta}$ is a subset of the $(q+1)(h+1)$- dimensional Euclidean space $\mathbb{R}^{(q+1)(h+1)}$. $\boldsymbol{\Theta}$ is a closed and bounded compact set and contains the true parameter value $\boldsymbol{\theta}_0$ as an interior point.
\end{ass}
\begin{ass}\label{rho-range}
	The spatial correlation coefficient $\rho$ satisfies
	$\rho \in (-1/\tau,1/\tau)$, where $\tau= \max\{|\tau_1|,|\tau_2|,\ldots,|\tau_n|\}$, $\tau_1,\ldots, \tau_n$ are eigenvalues of spatial weight matrix $W_n$. To avoid the non-stationarity issue when $\rho$ approaches to 1, we assume $\sup_{\rho\in\boldsymbol{\Theta}}|\rho|<1$.
\end{ass}
\begin{ass}\label{ub-weight}
	We assume $W_n$ is defined by queen contiguity and is uniformly bounded in row and column sums in absolute value as $n\rightarrow\infty$ so $(I_n-\rho W_n)^{-1}$ is also uniformly bounded in row and column sums as $n\rightarrow\infty$.
\end{ass}
\begin{ass}\label{x-boundness}
	$X_n$ is stationary, ergodic satisfying $\mathbb{E}\, |x_s|^2 <\infty, s=1,\ldots,n$ and $X_n$ is full column rank. 
\end{ass}
\begin{ass}\label{error-dist}
	The error terms $\varepsilon_s$, $s=1,2,\ldots,n$ are independent and identically distributed with density function $f(\cdot)$, zero mean and unit variance $\sigma^2=1$. The moment $E(|\varepsilon_s|^{2+r})$ exists for some $r>0$ and $E\left|\ln f(\varepsilon_s)\right|<\infty$.
\end{ass}
Assumption \ref{rho-range} defines the parameter space for $\rho$ as an interval around zero such that $I_n-\rho W_n$ is strictly diagonally dominant. By the Levy-Desplanques theorem \cite{taussky1949recurring}, it follows that $I_n-\rho W_n$ is nonsingular for any values $\rho$ in that interval. 

Note that the diagonal entries in $I_n-\rho W_n$ are all 1 (because $w_{ii}=0$). Using Gershgorin circle theorem \cite[p.~749-754]{gershgorin1931uber}, we can show that the largest eigenvalue of a row-standardized matrix $W_n$ is bounded by 1. Using the $9\times 9$ non-standardized weight matrix (\ref{9by9-matrix}) constructed under Queen's criterion in the section 2, the interval for $\rho$ is $(-0.207,0.207)$ whereas the row standardized weight matrix corresponds to $(-1,1)$.

It is natural to consider the neighborhood by connections and in many practical studies, since entries scaled to sum up to 1, each row of $W_n$ sums up to 1, which guarantees that all nonzero weights are in $(0,1]$. For simplicity, we define the weight matrix $W_n$ using the queen criterion and do row standardization. Assumption \ref{ub-weight} is originated by Kelejian and Prucha (1998 \cite{kelejian1998generalized}, 2001 \cite{kelejian1999generalized}) and is also used in Lee (2004 \cite{lee2004asymptotic}). Restricting $W_n$ to be uniformly bounded prevents the model prediction from exploding when $n$ goes to infinity. By Lemma A.4 in Lee \cite{lee2004asymptotic}, we can prove that $(I_n - \rho W_n)^{-1}$ is also uniformly bounded in row and column sums for $\rho \in (-1/\tau,1/\tau)$.

From Assumptions \ref{rho-range} and \ref{ub-weight} we can also decompose $W_n$ by its eigenvalue and eigenvector pairs $\tau_i, v_i$: $W_n = P\Lambda P^{-1}$, where $\Lambda$ is a diagonal matrix with eigenvalues $\tau_i$ on its diagonals and $P=[v_1,v_2,\ldots, v_n]$ (we assume $v_i$'s are normalized eigenvectors). So
{\small\setlength{\abovedisplayskip}{3pt}
	\setlength{\belowdisplayskip}{\abovedisplayskip}
	\setlength{\abovedisplayshortskip}{0pt}
	\setlength{\belowdisplayshortskip}{3pt}
	\begin{align}\label{matrix-decomposition}
	W_n=P\begin{pmatrix}
	\tau_1& 0 &\cdots&0\\
	0& \tau_2&\cdots&0\\
	0&0&\ddots&0\\
	0&0 &\cdots&\tau_n
	\end{pmatrix}P^{-1},
	(I_n-\rho W_n)^{-1}=P\begin{pmatrix}
	\frac{1}{1-\rho\tau_1}& 0 &\cdots&0\\
	0& \frac{1}{1-\rho\tau_2}&\cdots&0\\
	0&0&\ddots&0\\
	0&0 &\cdots&\frac{1}{1-\rho\tau_n}
	\end{pmatrix}P^{-1}
	\end{align}}%
This decomposition will later help us compute the likelihood function.

Assumption \ref{x-boundness} guarantees the stationarity of $\{x_s\}$ so we can apply ergodic theorem later in the proofs.

Assumption \ref{error-dist} imposes restrictions for the random error. We assume that errors $\{\varepsilon_s\}_{s=1}^n$ have an identical density function $f(\cdot)$. So to derive the likelihood function of $\boldsymbol{\theta}$, it is necessary to introduce the Jacobian coefficient which allows us to derive the joint distribution of $Y_n=\{y_s\}_{s=1}^n$ from that of $\{\varepsilon_s\}_{s=1}^n$, through equation (\ref{model-u}): 
{\small\setlength{\abovedisplayskip}{3pt}
  \setlength{\belowdisplayskip}{\abovedisplayskip}
  \setlength{\abovedisplayshortskip}{0pt}
  \setlength{\belowdisplayshortskip}{3pt}
\begin{equation}\label{jacobian}
J=det (\partial \boldsymbol{\varepsilon}_n/\partial Y_n)=|I_n - \rho W_n|
\end{equation}} %
\vspace{-20pt}

\noindent
Hence, based on the joint distribution for the vector of independent errors $\{\varepsilon_s\}_{s=1}^n$, and using (\ref{jacobian}) the log-likelihood function for $\boldsymbol{\theta}$ is given by (Anselin \cite[p.~63]{anselin2013spatial})
{\small\setlength{\abovedisplayskip}{3pt}
  \setlength{\belowdisplayskip}{\abovedisplayskip}
  \setlength{\abovedisplayshortskip}{0pt}
  \setlength{\belowdisplayshortskip}{3pt}
\begin{eqnarray}\label{likelihood}
\mathcal{L}_n(\boldsymbol{\theta})&=&\ln |I_n - \rho W_n|+\sum_{s=1}^{n}\ln f(\varepsilon_s(\boldsymbol{\theta}))\\\nonumber
\varepsilon_s(\boldsymbol{\theta})&=& y_s-x_s^{\prime}\beta-\rho\sum_{i=1}^{n}w_{si}y_i-\sum_{i=1}^{h}\lambda_i F(x_s^{\prime}\boldsymbol{\gamma}_i)
\end{eqnarray}} %
\vspace{-20pt}

\noindent
In practice, the density function $f$ could be chosen by looking at the distribution for observations and model residuals $\varepsilon_s(\boldsymbol{\theta})$. Common choices are normal distribution, t-distribution and Laplace distribution. We examined these three distributions (with unit variances under Assumption \ref{error-dist}) and the corresponding log-likelihood functions functions are given below.

When $\varepsilon_s\sim N(0,1)$,
{\small\setlength{\abovedisplayskip}{3pt}
  \setlength{\belowdisplayskip}{\abovedisplayskip}
  \setlength{\abovedisplayshortskip}{0pt}
  \setlength{\belowdisplayshortskip}{3pt}
\[f(\varepsilon_s)=\frac{1}{\sqrt{2\pi}}\exp (-\frac{\varepsilon_s^2}{2})\]\label{likelihood_normal}
\begin{equation*}
\mathcal{L}_n(\boldsymbol{\theta})=\ln |I_n - \rho W_n|-\frac{n}{2}\ln (2\pi)-\frac{1}{2}\sum_{s=1}^{n}\varepsilon_s^2(\boldsymbol{\theta})
\end{equation*}} %
\vspace{-20pt}

When $\varepsilon_s$ has the rescaled t distribution with degree of freedom $\nu$ ($\nu>2$, known) which is symmetric about zero and has variance 1:
{\small\setlength{\abovedisplayskip}{3pt}
  \setlength{\belowdisplayskip}{\abovedisplayskip}
  \setlength{\abovedisplayshortskip}{0pt}
  \setlength{\belowdisplayshortskip}{3pt} 
\[f(\varepsilon_s)=\sqrt{\frac{\nu}{\nu-2}}\frac{\Gamma[\frac{1}{2}(\nu+1)]}{\sqrt{\nu\pi}\,\Gamma(\frac{1}{2}\nu)}\cdot\left(1+\frac{\varepsilon_s^2}{\nu-2}\right)^{-\frac{1+\nu}{2}}\]
\begin{equation*}
\mathcal{L}_n(\boldsymbol{\theta})=\ln |I_n - \rho W_n|-\frac{n}{2}\ln (\nu-2)\pi+n\ln \frac{\Gamma[\frac{1}{2}(\nu+1)]}{\Gamma(\frac{1}{2}\nu)}-\frac{1+\nu}{2}\sum_{s=1}^{n}\ln \big(1+\frac{\varepsilon_s^2(\boldsymbol{\theta})}{\nu-2}\big)
\end{equation*}} %
\vspace{-20pt}

When $\varepsilon_s\sim$ Laplace distribution with mean $\mu =0$ and scale parameter $b=\sqrt{2}/2$,
{\small\setlength{\abovedisplayskip}{3pt}
  \setlength{\belowdisplayskip}{\abovedisplayskip}
  \setlength{\abovedisplayshortskip}{0pt}
  \setlength{\belowdisplayshortskip}{3pt}
\[f(\varepsilon_s)=\frac{1}{\sqrt{2}}\exp \Big(-\sqrt{2}|\varepsilon_s|\Big)\]
\begin{equation*}
\mathcal{L}_n(\boldsymbol{\theta})=\ln |I_n - \rho W_n|-\frac{n}{2}\ln 2-\sum_{s=1}^{n}\sqrt{2}|\varepsilon_s(\boldsymbol{\theta})|
\end{equation*}} %
\vspace{-20pt}

\noindent
In the following sections, we will discuss model identifiability and establish asymptotic properties for the maximum likelihood estimator $\hat{\boldsymbol{\theta}}=\arg\max\limits_{\boldsymbol{\theta}\in \boldsymbol{\Theta}} \mathcal{L}_n(\boldsymbol{\theta})$.

\section{Model Identification}
We now investigate the conditions under which our proposed model is identified. By Rothenberg \cite{rothenberg1971identification}, a parameter $\theta_0\in\boldsymbol{\Theta}$ is {\it globally identified} if there is no other $\theta$ in $\boldsymbol{\Theta}$ that is observationally equivalent to $\theta_0$ such that $f(y;\theta)=f(y;\theta_0)$; or the parameter $\theta_0$ is {\it locally identified} if there is no such $\theta$ in an open neighborhood of $\theta_0$ in $\boldsymbol{\Theta}$. 
The model (\ref{model}), in principle, is neither globally nor locally identified and the lack of identification of Neural Network models has been discussed in many papers (\citet{hwang1997prediction}; Medeiros \emph{et al.} \cite{medeiros2006building}). Here we extend the discussion to our proposed PSAR model. Three characteristics imply non-identification of our model:
(a) the interchangeable property: the value in the likelihood function may remain unchanged if we permute the hidden units. For a model with $h$ neurons, this will result in $h!$ different models that are indistinguishable from each other and have equal local maximums of the log-likelihood function;
(b) the ``symmetry" property: for a logistic function, $F(x)=1-F(-x)$ allows two equivalent parametrization for each of the hidden units;
(c) the reducible property: the presence of irrelevant neurons in model (\ref{model}) happens when $\lambda_i=0$ so parameters $\boldsymbol{\gamma}_i$ in this neuron would remain unidentified. Conversely, if $\boldsymbol{\gamma}_{i} =\boldsymbol{0}$, the output of that sigmoid function is a constant so $\lambda_i$ can take any value without affecting the value of likelihood functions. 

The problem of interchangeability (as mentioned in (a)) can be solved by imposing the following restriction, as in Medeiros \emph{et al.} \cite{medeiros2006building}:\\
{\bf Restriction 1.} {\it parameters $\lambda_1, \ldots, \lambda_h $  are restricted such that: $\lambda_1\geq\cdots\geq \lambda_h$.     
}\\
And to tackle (b) and (c), we can apply another restriction:\\
{\bf Restriction 2.} {\it The parameters $\lambda_i$ and $\gamma_{i1}$ should satisfy:\\
(1) $\lambda_i \neq 0$, $\forall i \in \{1,2,\ldots, h\}$; and\\
(2) $\gamma_{i1} > 0$, $\forall i \in \{1,2,\ldots, h\}$.      
}\\
To guarantee the non-singularity of model matrices and the uniqueness of parameters, we impose the following basic assumption:   
\begin{ass}\label{identifiability}
The true parameter vector $\boldsymbol{\theta}_0$ satisfies Restrictions 1-2.
\end{ass}
Referring to the section 4.3 by Medeiros \emph{et al.} \cite{medeiros2006building}, we can conclude the identifiability of the PSAR-ANN model
\begin{lem}\label{identify}
Under the Assumptions \ref{compact-space}-\ref{identifiability}, this partially specified spatial autoregressive model (\ref{model}) is globally identified.
\end{lem}

\section{Asymptotic Results}
\subsection{Preliminary}
Denote the true parameter vector as $\boldsymbol{\theta}_0$ and the solution which maximizes the log-likelihood function (\ref{likelihood}) as $\hat{\boldsymbol{\theta}}_n$ . Hence, $\hat{\boldsymbol{\theta}}_n$ should satisfy
{\small\setlength{\abovedisplayskip}{3pt}
  \setlength{\belowdisplayskip}{\abovedisplayskip}
  \setlength{\abovedisplayshortskip}{0pt}
  \setlength{\belowdisplayshortskip}{3pt}
\begin{align*}
\hat{\boldsymbol{\theta}}_n &\equiv \arg \max_{\boldsymbol{\theta}\in\boldsymbol{\Theta}}\mathcal{L}_n (\boldsymbol{\theta})\text{ ,}\\
\mathcal{L}_n(\boldsymbol{\theta})&=\ln |I_n - \rho W_n|+\sum_{s=1}^{n}\ln f\big(y_s-x_s^{\prime}\beta-\rho\sum_{i=1}^{n}w_{si}y_i-\sum_{i=1}^{h}\lambda_i F(x_s^{\prime}\boldsymbol{\gamma}_i)\big) 
\end{align*}} %
Suppose we have a $n_1 \times n_2$ lattice where we consider asymptotic properties of $\hat{\boldsymbol{\theta}}_n$ when $n=n_1n_2\rightarrow\infty$. Write the location $s$ as the coordinate $(s_x, s_y)$ in the $[1,n_1]\times[1,n_2]$ lattice space. The distance between two locations $s,j$ is defined as $d(s,j) = \max(|s_x-j_x|, |s_y-j_y|)$. So if observations at $s,j$ locations are neighbors (by queen criterion), their coordinates should satisfy $(s_x-j_x)^2+(s_y-j_y)^2\leq2$ or $d(s,j) =1$.

In a spatial context, we should notice that the functional form of $y_s$ is not identical for all the locations due to values of the weights $w_{si}$. For example, in a lattice, units at edges, vertexes or in the interior have different density functions due to different neighborhood structures (Figure \ref{neigh-str}). For an interior point (Figure \ref{neigh-str}(c)), its neighborhood set $\mathcal{N}_s$ contains eight neighbors where $w_{sj}=1/8$ if $d(s,j)=1$ otherwise $w_{sj}=0$, for $j=1,2,\ldots, n$. Similarly, an edge point (Figure \ref{neigh-str}(b)) has five neighboring units with $w_{sj}=1/5$ and the weight of a vertex neighborhood is $1/3$ because a vertex unit has only three neighbors. This is known as an edge effect in spatial problems.
\begin{figure}[ht]
\centering
\begin{subfigure}[b]{0.3\textwidth}
\begin{tikzpicture}
\draw[step=1.4cm,color=black] (0.1,0.1) grid (4.1,4.1);
\node at (0.7,3.5) {$s$};
\node at (2.1,3.5) {$j$};
\node at (0.7,2.1) {$j$};
\node at (2.1,2.1) {$j$};
\end{tikzpicture}
\caption{}
\end{subfigure}
\begin{subfigure}[b]{0.3\textwidth}
\begin{tikzpicture}
\draw[step=1.4cm,color=black] (0.1,0.1) grid (4.1,4.1);
\node at (0.7,0.7) {$j$};
\node at (2.1,0.7) {$j$};
\node at (2.1,2.1) {$j$};
\node at (0.7,2.1) {$s$};
\node at (2.1,3.5) {$j$};
\node at (0.7,3.5) {$j$};
\end{tikzpicture}
\caption{}
\end{subfigure}
\begin{subfigure}[b]{0.3\textwidth}
\begin{tikzpicture}
\draw[step=1.4cm,color=black] (0.1,0.1) grid (4.1,4.1);
\node at (0.7,0.7) {$j$};
\node at (2.1,0.7) {$j$};
\node at (3.5,0.7) {$j$};
\node at (0.7,2.1) {$j$};
\node at (2.1,2.1) {$s$};
\node at (3.5,2.1) {$j$};
\node at (0.7,3.5) {$j$};
\node at (2.1,3.5) {$j$};
\node at (3.5,3.5) {$j$};
\end{tikzpicture}
\caption{}
\end{subfigure}
\caption{Vertex (a), Edge (b) and Interior Points (c) Neighborhood Structures: $s$ is the target location and $j$ represents the neighborhood of $s$}
\label{neigh-str}
\end{figure}
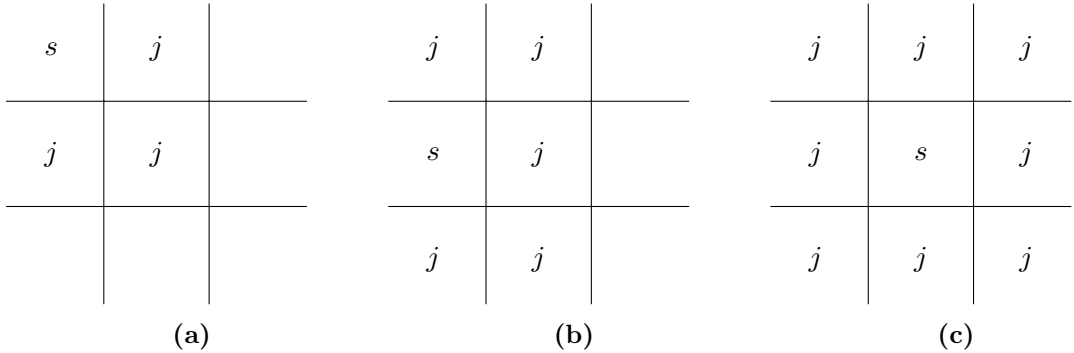
To deal with this, referring to Yao and Brockwell \cite{yao2006gaussian}, we construct an edge effect correction scheme based on the way that the sample size tends to infinity. 
In a space $[1,n_1]\times [1,n_2]$, we consider its interior area as $\mathcal{S}=\{(s_x,s_y): b_1\leq s_x \leq n_1-b_1, b_2\leq s_y \leq n_2-b_2\}$, where $b_1,b_2, n_1,n_2\rightarrow\infty$ satisfying that $b_1/n_1,b_2/n_2\rightarrow 0$ and other locations belong to the boundary areas $\mathcal{M}$.
Therefore the set $\mathcal{S}$ contains $n^{\ast}=(n_1-2b_1)(n_2-2b_2)$ interior locations while the set $\mathcal{M}$ contains $n-n^{\ast}$ boundary locations. Then $n^{\ast}/n\rightarrow1$ and $\mathcal{L}_n(\boldsymbol{\theta})$ can be split into a sum of two parts (interior $\mathcal{S}$ and boundary $\mathcal{M}$ parts):
{\small\setlength{\abovedisplayskip}{3pt}
  \setlength{\belowdisplayskip}{\abovedisplayskip}
  \setlength{\abovedisplayshortskip}{0pt}
  \setlength{\belowdisplayshortskip}{3pt}
\begin{align*}
\mathcal{L}_n(\boldsymbol{\theta})&=\sum_{s\in \mathcal{M}}l(\boldsymbol{\theta}|x_s, y_s)+\sum_{s\in \mathcal{S}}l(\boldsymbol{\theta}|x_s,y_s)\\
l(\boldsymbol{\theta}|x_s,y_s)&=n^{-1}\ln |I_n - \rho W_n|+\ln f\big(y_s-x_s^{\prime}\beta-\rho\sum_{i=1}^{n}w_{si}y_i-\sum_{i=1}^{h}\lambda_i F(x_s^{\prime}\boldsymbol{\gamma}_i)\big)
\end{align*}} %
\vspace{-20pt}

\noindent
Therefore, given that $\lim_{n_1,n_2\rightarrow\infty}\frac{|\mathcal{M}|}{n}=0$, $n^{-1}\sum_{s\in \mathcal{M}}l(\boldsymbol{\theta}|x_s, y_s)$ vanishes a.s. as $n$ tends to infinity for any $\boldsymbol{\theta}\in\boldsymbol{\Theta}$. Therefore,
{\small\setlength{\abovedisplayskip}{3pt}
  \setlength{\belowdisplayskip}{\abovedisplayskip}
  \setlength{\abovedisplayshortskip}{0pt}
  \setlength{\belowdisplayshortskip}{3pt}
\begin{align*}
\lim_{n_1,n_2\rightarrow\infty}n^{-1} \mathcal{L}_n(\boldsymbol{\theta})&=\lim_{n_1,n_2\rightarrow\infty} (n_1n_2)^{-1} \Big(\sum_{s\in \mathcal{M}}l(\boldsymbol{\theta}|x_s,y_s)+\sum_{s\in \mathcal{S}}l(\boldsymbol{\theta}|x_s, y_s)\Big)\\
&= \lim_{n_1,n_2\rightarrow\infty}(n_1n_2)^{-1}\sum_{s\in \mathcal{S}}l(\boldsymbol{\theta}|x_s,y_s)\quad a.s.
\end{align*}} %
\vspace{-20pt}

\noindent
In this equation, every location $s\in \mathcal{S}$ has eight neighboring units under the queen criterion with nonzero weights $w_{sj}=1/8$. Hence for an interior unit $s\in\mathcal{S}$, $\sum_{i=1}^{n}w_{si}y_i= \sum_{j=1}^n\frac{1}{8}y_jI_{\{d(s,j)=1\}}$. And the log likelihood function $\mathcal{L}_n(\boldsymbol{\theta})$ is approximately
{\small\setlength{\abovedisplayskip}{3pt}
  \setlength{\belowdisplayskip}{\abovedisplayskip}
  \setlength{\abovedisplayshortskip}{0pt}
  \setlength{\belowdisplayshortskip}{3pt}
\begin{equation}\label{mod-average}
n^{-1}\mathcal{L}_n(\boldsymbol{\theta})\approx n^{-1}\sum_{s\in \mathcal{S}}l(\boldsymbol{\theta}|x_s, y_s) \quad \text{for large $n$}
\end{equation}} %
So the maximum likelihood estimator $\hat{\boldsymbol{\theta}}_n$ approximately maximizes $n^{-1}\sum_{s\in \mathcal{S}}l(\boldsymbol{\theta}|x_s, y_s)$.
{\small\setlength{\abovedisplayskip}{3pt}
  \setlength{\belowdisplayskip}{\abovedisplayskip}
  \setlength{\abovedisplayshortskip}{0pt}
  \setlength{\belowdisplayshortskip}{3pt}
\begin{equation*}
\hat{\boldsymbol{\theta}}_n \approx \arg\max_{\boldsymbol{\theta}\in\boldsymbol{\Theta}}n^{-1}\sum_{s\in \mathcal{S}}l(\boldsymbol{\theta}|x_s, y_s)
\end{equation*}} %

\subsection{Consistency Results}

To establish the consistency of $\hat{\boldsymbol{\theta}}_n$, the heuristic insight is that because $\hat{\boldsymbol{\theta}}_n$ maximizes $n^{-1}\mathcal{L}_{n}(\boldsymbol{\theta})$, it approximately maximizes $n^{-1}\sum_{s\in \mathcal{S}}l(\boldsymbol{\theta}|x_s, y_s)$. By (\ref{mod-average}), $n^{-1}\mathcal{L}_{n}(\boldsymbol{\theta})$ can generally be shown tending to a real function $\mathcal{L}:\boldsymbol{\Theta}\rightarrow\mathbb{R}$ with maximizer $\boldsymbol{\theta}_0$ as $n\rightarrow\infty$ under mild conditions on the data generating process, then $\hat{\boldsymbol{\theta}}_n$ should tend to $\boldsymbol{\theta}_0$ almost surely. Before the formal proof of the consistency, we need the following assumptions on density function $f(\cdot)$ satisfied (similar assumptions are made in White \cite{white1996estimation}, Andrews, Davis and Breidt \cite{andrews2006maximum}, Lii and Rosenblatt \cite{lii1992approximate}).
\begin{ass}\label{continuous}
	For all $s \in \mathbb{R}$, $f(s) >0$ and $f(s)$ is twice continuously differentiable with respect to $s$.
\end{ass}
\begin{ass}\label{error-integral}The density should satisfy the following equations:\\
	$\bullet$ $\int sf^{\prime}(s)\,ds = sf(s)|^{\infty}_{-\infty}-\int f(s)\,ds = -1$\\
	$\bullet$ $\int f^{\prime\prime}(s)\,ds = f^{\prime}(s)|^{\infty}_{-\infty} = 0$\\
	$\bullet$ $\int s^2f^{\prime\prime}(s)\,ds = s^2f^{\prime}(s)|^{\infty}_{-\infty} - 2\int sf^{\prime}(s)\,ds = 2$
\end{ass}
\begin{ass}\label{error-dominance}
The density should follow the following dominance condition:\\
	$\left|\frac{f^{\prime}(s)}{f(s)}\right|$, $\left|\frac{f^{\prime}(s)}{f(s)}\right|^2$,  $\left|\frac{f^{\prime}(s)}{f(s)}\right|^4$, $\frac{f^{\prime\prime}(s)}{f(s)}$,  and $\frac{f^{\prime\prime}(s)f^{\prime2}(s)}{f^3(s)}$ are dominated by $a_1+a_2\left|s\right|^{c_1}$, where $a_1$, $a_2$, $c_1$ are non-negative constants and $\int_{-\infty}^{\infty}\left|s\right|^{c_1+2}f(s)\,ds < \infty$.
\end{ass}
Discussed in Breidt, Davis, Lii and Rosenblatt \cite{breid1991maximum} and Andrews, Davis and Breidt \cite[p.~1642-1645]{andrews2006maximum}), these assumptions on the density $f(\cdot)$ are satisfied in the t-distribution case when $\nu >2$ and the mixed Gaussian distribution. The assumption $\mathbb{E}\,|\ln f(s)|<\infty$ (see Assumption \ref{error-dist}) is also checked satisfied in the normal and t-distribution ($\nu >2$). The Laplace distribution does not strictly satisfy the Assumptions \ref{continuous}-\ref{error-dominance}, since it is not differentiable at 0 but it satisfies these boundedness conditions almost everywhere so we believe the consistency and asymptotic normality results remain valid for parameter estimates. This will be shown in the simulation section. 

\begin{lem}\label{uniq-max}
	Given Assumptions \ref{compact-space}-\ref{error-dominance},
{\small\setlength{\abovedisplayskip}{3pt}
	\setlength{\belowdisplayskip}{\abovedisplayskip}
	\setlength{\abovedisplayshortskip}{0pt}
	\setlength{\belowdisplayshortskip}{3pt}
		\begin{equation}\label{uniq-max-eqn}
		\boldsymbol{\theta}_0 =  \max_{\boldsymbol{\theta}\in\boldsymbol{\Theta}} \mathbb{E}\,\mathcal{L}_n(\boldsymbol{\theta}) \equiv \max_{\boldsymbol{\theta}\in\boldsymbol{\Theta}} \mathbb{E}\,\frac{\mathcal{L}_n(\boldsymbol{\theta})}{n}\quad \text{for all }n
		\end{equation}}%
\end{lem}
 \begin{proof}
 	$\mathcal{L}_n(\boldsymbol{\theta})$ is the log of the likelihood function $L_n(\boldsymbol{\theta})$,
 {\small\setlength{\abovedisplayskip}{3pt}
 	\setlength{\belowdisplayskip}{\abovedisplayskip}
 	\setlength{\abovedisplayshortskip}{0pt}
 	\setlength{\belowdisplayshortskip}{3pt}
 \begin{align*}
    \mathcal{L}_n(\boldsymbol{\theta}) = \ln |I_n-\rho W_n| + \sum_{s=1}^{n}\ln f(\varepsilon_s(\boldsymbol{\theta}))\\
 	\mathbb{E}\,\mathcal{L}_n(\boldsymbol{\theta})- \mathbb{E}\,\mathcal{L}_n(\boldsymbol{\theta}_0)=\mathbb{E}\ln \frac{L_n(\boldsymbol{\theta})}{L_n(\boldsymbol{\theta}_0)}
 \end{align*} }%
Denote $Z_n=(Y_n,X_n)$. By Jensen's inequality,
 {\small\setlength{\abovedisplayskip}{3pt}
	\setlength{\belowdisplayskip}{\abovedisplayskip}
	\setlength{\abovedisplayshortskip}{0pt}
	\setlength{\belowdisplayshortskip}{3pt}
	\begin{align*}
\mathbb{E}\ln \frac{L_n(\boldsymbol{\theta})}{L_n(\boldsymbol{\theta}_0)} \leq \ln \mathbb{E}\frac{L_n(\boldsymbol{\theta})}{L_n(\boldsymbol{\theta}_0)} = \ln \int_{-\infty}^{\infty} \frac{L_n(\boldsymbol{\theta})}{L_n(\boldsymbol{\theta}_0)} L_n(\boldsymbol{\theta}_0)\,dZ_n= 0 
\end{align*}}%
So $\mathbb{E}\mathcal{L}_n(\boldsymbol{\theta})\leq \mathbb{E}\mathcal{L}_n(\boldsymbol{\theta}_0)$. By Lemma \ref{identify}, the PSAR model is globally identified and therefore, $\mathbb{E}\mathcal{L}_n(\boldsymbol{\theta})$ is uniquely maximized at $\boldsymbol{\theta}_0$ for all $n$. Since the parameter vector $\boldsymbol{\theta}_0$ does not depend on sample size $n$, it is equivalent to say that, $\boldsymbol{\theta}_0 = \max_{\boldsymbol{\theta}\in\boldsymbol{\Theta}} \frac{1}{n}\mathbb{E}\mathcal{L}_n(\boldsymbol{\theta})$. 
 \end{proof}
In the following, to simplify the expression, denote $g(x_s,\boldsymbol{\theta}) = x_s^{\prime}\beta+ \boldsymbol{F}(x_s^{\prime}\boldsymbol{\gamma})\lambda$. 
Define the Hadamard product $\circ$ as,
	{\small\setlength{\abovedisplayskip}{3pt}
	\setlength{\belowdisplayskip}{\abovedisplayskip}
	\setlength{\abovedisplayshortskip}{0pt}
	\setlength{\belowdisplayshortskip}{3pt}
	\begin{gather*}
		a\circ B =\begin{bmatrix}
		a_1b_{11} & a_1b_{21}&\cdots&a_1b_{n1}\\
		a_2b_{12} & a_2b_{22}&\cdots&a_2b_{n2}\\
		\vdots&\vdots&\ddots&\vdots\\
		a_nb_{1n} & a_nb_{2n}&\cdots&a_nb_{nn}
		\end{bmatrix},
		a\circ b_1 =\begin{bmatrix}
		a_1b_{11}\\
		a_2b_{12}\\
		\vdots\\
		a_nb_{1n}
		\end{bmatrix}
	\end{gather*}}%
where $a,b_1,\ldots,b_n\in\mathbb{R}^n$, a matrix $B=(b_1, \ldots,b_n)\in \mathbb{R}^{n \times n}$. And let
{\small\setlength{\abovedisplayskip}{3pt}
	\setlength{\belowdisplayskip}{\abovedisplayskip}
	\setlength{\abovedisplayshortskip}{0pt}
	\setlength{\belowdisplayshortskip}{3pt}
	\begin{align*}
	k_0 &= \int \left|\frac{f^{\prime}(s)}{f(s)}\right|f(s)\,ds\\
	k_1 &= \int \left|\frac{f^{\prime^2}(s)}{f^2(s)}-\frac{f^{\prime\prime}(s)}{f(s)}\right|f(s)\,ds\\
	k_2 & =\int \left|\frac{sf^{\prime^2}(s)}{f(s)}-\frac{sf^{\prime\prime}(s)}{f(s)}\right|f(s)\,ds\\
	k_3 &= \int \left|\frac{s^2f^{\prime^2}(s)}{f(s)}-\frac{s^2f^{\prime\prime}(s)}{f(s)}\right|f(s)\,ds
\end{align*}} %
\begin{lem}\label{uni-cvg}
Given Assumptions \ref{compact-space}-\ref{error-dominance},
{\small\setlength{\abovedisplayskip}{3pt}
  \setlength{\belowdisplayskip}{\abovedisplayskip}
  \setlength{\abovedisplayshortskip}{0pt}
  \setlength{\belowdisplayshortskip}{3pt}
\begin{gather}\label{uni-cvg-eqn}
\sup\limits_{\boldsymbol{\theta}\in\boldsymbol{\Theta}}\left|\frac{1}{n}\sum_{s=1}^{n}\ln f(\varepsilon_s(\boldsymbol{\theta}))-\mathbb{E} \frac{1}{n}\sum_{s=1}^{n}\ln f(\varepsilon_s(\boldsymbol{\theta}))\right|\xrightarrow{p} 0\text{ as }n\rightarrow\infty
\end{gather}} %
\end{lem}
\begin{proof}
As illustrated in equation (\ref{mod-average}), in a lattice with size $n_1\times n_2$,
{\small\setlength{\abovedisplayskip}{3pt}
  \setlength{\belowdisplayskip}{\abovedisplayskip}
  \setlength{\abovedisplayshortskip}{0pt}
  \setlength{\belowdisplayshortskip}{3pt}
\begin{gather*}
\sup\limits_{\boldsymbol{\theta}\in\boldsymbol{\Theta}}\left|\frac{1}{n}\sum_{s=1}^{n}\ln f(\varepsilon_s(\boldsymbol{\theta}))-\frac{1}{n}\sum_{s\in\mathcal{S}}\ln f(\varepsilon_s(\boldsymbol{\theta}))\right|\xrightarrow{a.s.}0\text{ as }n_1,n_2\rightarrow\infty
\end{gather*}} %
Therefore, to prove (\ref{uni-cvg-eqn}) is equivalent to show that
{\small\setlength{\abovedisplayskip}{3pt}
  \setlength{\belowdisplayskip}{\abovedisplayskip}
  \setlength{\abovedisplayshortskip}{0pt}
  \setlength{\belowdisplayshortskip}{3pt}
\begin{gather}\label{uni-cvg-eqn-2}
\sup\limits_{\boldsymbol{\theta}\in\boldsymbol{\Theta}}\left|\frac{1}{n}\sum_{s\in\mathcal{S}}\ln f(\varepsilon_s(\boldsymbol{\theta}))-\mathbb{E} \frac{1}{n}\sum_{s\in\mathcal{S}}\ln f(\varepsilon_s(\boldsymbol{\theta}))\right|\xrightarrow{p} 0\text{ as }n\rightarrow\infty
\end{gather}} %
where $\mathcal{S}$ denotes the interior units mentioned before. Since the interior units have the same neighboring structure, the space process for them is stationary when $n_1,n_2$ go to infinity.
We first show $\left|\frac{1}{n}\sum_{s\in\mathcal{S}}\ln f(\varepsilon_s(\boldsymbol{\theta}))-\mathbb{E} \frac{1}{n}\sum_{s\in\mathcal{S}}\ln f(\varepsilon_s(\boldsymbol{\theta}))\right|\xrightarrow{p} 0$ for fixed $\boldsymbol{\theta}$ (Similar proof in \citet[Theorem 3.1,4.1]{lee2004asymptotic}). 

To prove this, we want to show that $\mathbb{E}|\ln f(\varepsilon_{s}(\boldsymbol{\theta}))| < \infty, s\in\mathcal{S}$. Expanding $\ln f(\varepsilon_{s}(\boldsymbol{\theta}))$ around $\boldsymbol{\theta}_0$ with respect to $\boldsymbol{\theta}$,
{\small\setlength{\abovedisplayskip}{3pt}
	\setlength{\belowdisplayskip}{\abovedisplayskip}
	\setlength{\abovedisplayshortskip}{0pt}
	\setlength{\belowdisplayshortskip}{3pt}
	\begin{align*}
	\ln f(\varepsilon_{s}(\boldsymbol{\theta})) &= \ln f(\varepsilon_{s}(\boldsymbol{\theta}_0))+ \left|\frac{f^{\prime}(\varepsilon_{s}(\tilde{\boldsymbol{\theta}}_n))}{f(\varepsilon_{s}(\tilde{\boldsymbol{\theta}}_n))}\frac{\partial \varepsilon_{s}(\tilde{\boldsymbol{\theta}}_n)}{\partial \boldsymbol{\theta}^{\prime}}\right|(\boldsymbol{\theta}-\boldsymbol{\theta}_0)\\
	\mathbb{E}|\ln f(\varepsilon_{s}(\boldsymbol{\theta}))| &\leq \mathbb{E}|\ln f(\varepsilon_{s}(\boldsymbol{\theta}_0))|+\mathbb{E}\left|\frac{f^{\prime}(\varepsilon_{s}(\tilde{\boldsymbol{\theta}}_n))}{f(\varepsilon_{s}(\tilde{\boldsymbol{\theta}}_n))}\frac{\partial \varepsilon_{s}(\tilde{\boldsymbol{\theta}}_n)}{\partial \boldsymbol{\theta}^{\prime}}\right||\boldsymbol{\theta}-\boldsymbol{\theta}_0|
	\end{align*}}%
where $\tilde{\boldsymbol{\theta}}_n$ is between $\boldsymbol{\theta}$ and $\boldsymbol{\theta}_0$. Under the true parameter values $\varepsilon_s(\boldsymbol{\theta}_0)$ (denoted as $\varepsilon_s$ or $\boldsymbol{\varepsilon}_n$ as its vector form in the following) is independent and identically distributed. From Assumption \ref{error-dist}, $\mathbb{E}\left|\ln f(\varepsilon_s)\right|<\infty$.
For $\mathbb{E}\left|\frac{f^{\prime}(\varepsilon_{s}(\tilde{\boldsymbol{\theta}}))}{f(\varepsilon_{s}(\tilde{\boldsymbol{\theta}}))}\frac{\partial \varepsilon_{s}(\tilde{\boldsymbol{\theta}})}{\partial \boldsymbol{\theta}^{\prime}}\right|$, $\left|\frac{\partial \varepsilon_{s}(\tilde{\boldsymbol{\theta}})}{\partial \boldsymbol{\theta}}\right|$ can be expressed as
{\small\setlength{\abovedisplayskip}{3pt}
	\setlength{\belowdisplayskip}{\abovedisplayskip}
	\setlength{\abovedisplayshortskip}{0pt}
	\setlength{\belowdisplayshortskip}{3pt}
	\begin{align}\label{error-theta}\nonumber
	\left|\frac{\partial \varepsilon_{s}(\tilde{\boldsymbol{\theta}})}{\partial \beta}\right| &= \left|x_{s}\right|\\
	\left|\frac{\partial \varepsilon_{s}(\tilde{\boldsymbol{\theta}})}{\partial \lambda}\right| &= \left|\boldsymbol{F}(x_{s}^{\prime}\tilde{\boldsymbol{\gamma}})^{\prime}\right| \leq \boldsymbol{1}_{h}\\\nonumber
	\left|\frac{\partial \varepsilon_{s}(\tilde{\boldsymbol{\theta}})}{\partial \boldsymbol{\gamma}_i}\right| &= \left|\tilde{\lambda}_i\frac{\partial F(x_{s}^{\prime}\tilde{\boldsymbol{\gamma}}_i)}{\partial x_{s}^{\prime}\boldsymbol{\gamma}_i}x_{s}\right|=\left|\tilde{\lambda}_i F(x_{s}^{\prime}\tilde{\boldsymbol{\gamma}}_i)(1-F(x_{s}^{\prime}\tilde{\boldsymbol{\gamma}}_i))x_{s}\right|\\\nonumber
	&\leq \max_{\lambda_{i}\in\boldsymbol{\Theta}}|\lambda_{i}|\frac{|x_s|}{4}\\\nonumber
	\left|\frac{\partial \varepsilon_{s}(\tilde{\boldsymbol{\theta}})}{\partial \rho}\right| &= \left|\sum_{i=1}^{n}w_{si}y_{i}\right| =\left|[M_n(\boldsymbol{g}(X_n,\tilde{\boldsymbol{\theta}}_n)+\boldsymbol{\varepsilon}_n(\tilde{\boldsymbol{\theta}}_n))]_s\right|=\left|\sum_{k=1}^{n}m_{sk}(g(x_k,\tilde{\boldsymbol{\theta}}_n)+\varepsilon_k(\tilde{\boldsymbol{\theta}}_n))\right|
	\end{align}}%
where $m_{ij}$ is $(i,j)$ element of $M_n = W_n(I_n-\rho W_n)^{-1}$. $M_n$ is bounded uniformly in column and row sums (see Assumption \ref{ub-weight}) so $\sum_{j=1}^{n}m_{ij}, \sum_{i=1}^{n}m_{ij}$ are bounded by a constant $b$ for $i,j=1,\ldots,n$.
The logistic function $F(x)$ is bounded by $1$ and its derivative $F^{\prime}(x)$ is also bounded by $1$. $|\frac{f^{\prime}(s)}{f(s)}|$ is dominated by $a_1+a_2\left|s\right|^{c_1}$, $\int_{-\infty}^{\infty}|s|^{c_1}f(s)<\infty$ which implies that $\mathbb{E}\left|\frac{f^{\prime}(\varepsilon_{s}(\tilde{\boldsymbol{\theta}}_n))}{f(\varepsilon_{s}(\tilde{\boldsymbol{\theta}}_n))}\right|<\infty$. With Cauchy–Schwarz inequality \cite{steele2004cauchy} and the finite second moment of $X_n$, we can have,
{\small\setlength{\abovedisplayskip}{3pt}
	\setlength{\belowdisplayskip}{\abovedisplayskip}
	\setlength{\abovedisplayshortskip}{0pt}
	\setlength{\belowdisplayshortskip}{3pt}
	\begin{align}\label{log-derivative-bound}\nonumber
	\mathbb{E}\left|\frac{f^{\prime}(\varepsilon_{s}(\tilde{\boldsymbol{\theta}}_n))}{f(\varepsilon_{s}(\tilde{\boldsymbol{\theta}}_n))}\frac{\partial \varepsilon_{s}(\tilde{\boldsymbol{\theta}}_n)}{\partial \beta}\right| &= \mathbb{E}\left|\frac{f^{\prime}(\varepsilon_{s}(\tilde{\boldsymbol{\theta}}_n))}{f(\varepsilon_{s}(\tilde{\boldsymbol{\theta}}_n))}x_{s}\right|<\left(\mathbb{E}\left|\frac{f^{\prime}(\varepsilon_{s}(\tilde{\boldsymbol{\theta}}_n))}{f(\varepsilon_{s}(\tilde{\boldsymbol{\theta}}_n))}\right|^2\mathbb{E}\left|x_{s}\right|^2\right)^{1/2}< \infty\\
	\mathbb{E}\left|\frac{f^{\prime}(\varepsilon_{s}(\tilde{\boldsymbol{\theta}}_n))}{f(\varepsilon_{s}(\tilde{\boldsymbol{\theta}}_n))}\frac{\partial \varepsilon_{s}(\tilde{\boldsymbol{\theta}}_n)}{\partial \lambda}\right| &= \mathbb{E}\left|\frac{f^{\prime}(\varepsilon_{s}(\tilde{\boldsymbol{\theta}}_n))}{f(\varepsilon_{s}(\tilde{\boldsymbol{\theta}}_n))}\boldsymbol{F}(x_{s}^{\prime}\tilde{\boldsymbol{\gamma}})^{\prime}\right|\leq \mathbb{E}\left|\frac{f^{\prime}(\varepsilon_{s}(\tilde{\boldsymbol{\theta}}_n))}{f(\varepsilon_{s}(\tilde{\boldsymbol{\theta}}_n))}\boldsymbol{1}_h\right|<\infty\\\nonumber
	\mathbb{E}\left|\frac{f^{\prime}(\varepsilon_{s}(\tilde{\boldsymbol{\theta}}_n))}{f(\varepsilon_{s}(\tilde{\boldsymbol{\theta}}_n))}\frac{\partial \varepsilon_{s}(\tilde{\boldsymbol{\theta}}_n)}{\partial \boldsymbol{\gamma}_i}\right| &\leq \mathbb{E}\left|\frac{f^{\prime}(\varepsilon_{s}(\tilde{\boldsymbol{\theta}}_n))}{f(\varepsilon_{s}(\tilde{\boldsymbol{\theta}}_n))}\tilde{\lambda}_ix_{s}\right|< \infty\\\nonumber
	\mathbb{E}\left|\frac{f^{\prime}(\varepsilon_{s}(\tilde{\boldsymbol{\theta}}_n))}{f(\varepsilon_{s}(\tilde{\boldsymbol{\theta}}_n))}\frac{\partial \varepsilon_{s}(\tilde{\boldsymbol{\theta}}_n)}{\partial \rho}\right| &= \mathbb{E}\left|\frac{f^{\prime}(\varepsilon_{s}(\tilde{\boldsymbol{\theta}}_n))}{f(\varepsilon_{s}(\tilde{\boldsymbol{\theta}}_n))}\sum_{k=1}^{n}m_{sk}(g(x_k,\tilde{\boldsymbol{\theta}}_n)+\varepsilon_k(\tilde{\boldsymbol{\theta}}_n))\right|\\\nonumber
	&< b\cdot \mathbb{E}\left|\frac{\varepsilon_{s}(\tilde{\boldsymbol{\theta}}_n)f^{\prime}(\varepsilon_{s}(\tilde{\boldsymbol{\theta}}_n))}{f(\varepsilon_{s}(\tilde{\boldsymbol{\theta}}_n))}\right|+k_0\mathbb{E}\left|\sum_{k=1}^{n}m_{sk}g(x_k,\tilde{\boldsymbol{\theta}}_n)\right|
	\end{align}}%
Since $\mathbb{E}\,|x_s|^2<\infty$ for all $s$, $\mathbb{E}|g(x_k,\tilde{\boldsymbol{\theta}}_n)|$ is finite for $\tilde{\boldsymbol{\theta}}_n\in\boldsymbol{\Theta}$. By assumption \ref{error-dominance}, $\left|\frac{sf^{\prime}(s)}{f(s)}\right| <\infty$, so $	\mathbb{E}\left|\frac{f^{\prime}(\varepsilon_{s}(\tilde{\boldsymbol{\theta}}_n))}{f(\varepsilon_{s}(\tilde{\boldsymbol{\theta}}_n))}\frac{\partial \varepsilon_{s}(\tilde{\boldsymbol{\theta}}_n)}{\partial \rho}\right|<\infty$.
Therefore $\mathbb{E}|\ln f(\varepsilon_s(\boldsymbol{\theta}_0))|<\infty$ and we can conclude that
$\mathbb{E}|\ln f(\varepsilon_s(\boldsymbol{\theta}))| < \infty$, so that, by ergodic theorem,
{\small\setlength{\abovedisplayskip}{3pt}
	\setlength{\belowdisplayskip}{\abovedisplayskip}
	\setlength{\abovedisplayshortskip}{0pt}
	\setlength{\belowdisplayshortskip}{3pt}
	\begin{gather*}
		\left|\frac{1}{n}\sum_{s\in\mathcal{S}}\ln f(\varepsilon_s(\boldsymbol{\theta}))-\mathbb{E}\frac{1}{n}\sum_{s\in\mathcal{S}}\ln f(\varepsilon_s(\boldsymbol{\theta}))\right | \rightarrow_{p} 0,\quad n\rightarrow\infty
    \end{gather*}}%

To complete the proof of uniform convergence,  we also need to show the equicontinuity of $\frac{1}{n}\sum_{s\in\mathcal{S}}\ln f(\varepsilon_s(\boldsymbol{\theta}))$, i.e., for all $\boldsymbol{\theta}_1, \boldsymbol{\theta}_2\in \boldsymbol{\Theta}$,
{\small\setlength{\abovedisplayskip}{3pt}
	\setlength{\belowdisplayskip}{\abovedisplayskip}
	\setlength{\abovedisplayshortskip}{0pt}
	\setlength{\belowdisplayshortskip}{3pt}
	\begin{align}\label{equi}
	\frac{1}{n}\left|\sum_{s\in\mathcal{S}}\ln f(\varepsilon_s(\boldsymbol{\theta}_1))-\sum_{s\in\mathcal{S}}\ln f(\varepsilon_s(\boldsymbol{\theta}_2))\right| \leq ||\boldsymbol{\theta}_1-\boldsymbol{\theta}_2||O_p(1)
   \end{align}}%
Applying the mean value theorem to the left side in (\ref{equi}):
{\small\setlength{\abovedisplayskip}{3pt}
	\setlength{\belowdisplayskip}{\abovedisplayskip}
	\setlength{\abovedisplayshortskip}{0pt}
	\setlength{\belowdisplayshortskip}{3pt}
	\begin{align*}
	\frac{1}{n}\left|\sum_{s\in\mathcal{S}}\ln f(\varepsilon_s(\boldsymbol{\theta}_1))-\sum_{s\in\mathcal{S}}\ln f(\varepsilon_s(\boldsymbol{\theta}_2))\right| &\leq \frac{1}{n}\left|\sum_{s\in\mathcal{S}}\frac{\partial \ln f(\varepsilon_s(\tilde{\boldsymbol{\theta}}_n))}{\partial \boldsymbol{\theta}^{\prime}}\right|||\boldsymbol{\theta}_1-\boldsymbol{\theta}_2||\\
	&=\frac{1}{n}\left|\sum_{s\in\mathcal{S}}\frac{ f^{\prime}(\varepsilon_s(\tilde{\boldsymbol{\theta}}_n))}{ f(\varepsilon_s(\tilde{\boldsymbol{\theta}}_n))}\frac{\partial \varepsilon_s(\tilde{\boldsymbol{\theta}}_n)}{\partial \boldsymbol{\theta}^{\prime}}\right| ||\boldsymbol{\theta}_1-\boldsymbol{\theta}_2||
	\end{align*}}%
where $\tilde{\boldsymbol{\theta}}_n$ is some value between $\boldsymbol{\theta}_1$ and $\boldsymbol{\theta}_2$. By the ergodic theorem, $\frac{1}{n}\left|\sum_{s\in\mathcal{S}}\frac{ f^{\prime}(\varepsilon_s(\tilde{\boldsymbol{\theta}}_n))}{ f(\varepsilon_s(\tilde{\boldsymbol{\theta}}_n))}\frac{\partial \varepsilon_s(\tilde{\boldsymbol{\theta}}_n)}{\partial \boldsymbol{\theta}}\right|\xrightarrow{a.s.} \mathbb{E}\left|\frac{f^{\prime}(\varepsilon_s(\tilde{\boldsymbol{\theta}}))}{f(\varepsilon_s(\tilde{\boldsymbol{\theta}}_n))}\frac{\partial\varepsilon_s(\tilde{\boldsymbol{\theta}}_n)}{\partial \boldsymbol{\theta}}\right|
$. Since $\boldsymbol{\theta}$ is in a compact set $\boldsymbol{\Theta}$, we show in (\ref{res-expansion}) that, for all $s$, $\varepsilon_s(\boldsymbol{\theta})$ is bounded by some function of $Y_n,X_n$ not depending on $\boldsymbol{\theta}$.
{\small\setlength{\abovedisplayskip}{3pt}
	\setlength{\belowdisplayskip}{\abovedisplayskip}
	\setlength{\abovedisplayshortskip}{0pt}
	\setlength{\belowdisplayshortskip}{3pt}
	\begin{align}\label{res-expansion}\nonumber
    |\boldsymbol{\varepsilon}_n(\boldsymbol{\theta})|&=\left|Y_n-\rho W_nY_n-X_n\beta-\boldsymbol{F}(X_n\boldsymbol{\gamma}^{\prime})\lambda\right|\\
    &\leq \left|(I_n-\rho W_n)Y_n\right|+\left|X_n\beta\right| +\left|\boldsymbol{F}(X_n\boldsymbol{\gamma}^{\prime})\lambda\right|\\\nonumber
    &\leq (I_n+\max_{\rho\in\boldsymbol{\Theta}}|\rho|W_n)|Y_n|+|X_n|\max_{\beta\in\boldsymbol{\Theta}}|\beta|+\max_{\lambda\in\boldsymbol{\Theta}}||\lambda||\boldsymbol{1}_n
   \end{align}}%
Similarly, referring to (\ref{error-theta}), it is easy to show that $\left|\frac{\partial\varepsilon_s(\boldsymbol{\theta})}{\partial \boldsymbol{\theta}}\right|$ is also bounded by some function about $Y_n$ and $X_n$.
Therefore, due to the dominance of $\left|\frac{f^{\prime}(s)}{f(s)}\right|$ (see Assumption \ref{error-dominance}) and stationarity of $X_n,Y_n$, for $\tilde{\boldsymbol{\theta}}_n$ between $\boldsymbol{\theta}_1$ and $\boldsymbol{\theta}_2$, there exists a constant $M$ such that
{\small\setlength{\abovedisplayskip}{3pt}
	\setlength{\belowdisplayskip}{\abovedisplayskip}
	\setlength{\abovedisplayshortskip}{0pt}
	\setlength{\belowdisplayshortskip}{3pt}
	\begin{align}\label{loglikelihood-theta-expectation}
	\frac{1}{n}\left|\sum_{s\in\mathcal{S}}\frac{ f^{\prime}(\varepsilon_s(\tilde{\boldsymbol{\theta}}_n))}{ f(\varepsilon_s(\tilde{\boldsymbol{\theta}}_n))}\frac{\partial \varepsilon_s(\tilde{\boldsymbol{\theta}}_n)}{\partial \boldsymbol{\theta}^{\prime}}\right|\leq M\quad \text{for }n\rightarrow\infty
\end{align}}%
Hence, for $\boldsymbol{\theta}_1,\boldsymbol{\theta}_2 \in \boldsymbol{\Theta}$
{\small\setlength{\abovedisplayskip}{3pt}
	\setlength{\belowdisplayskip}{\abovedisplayskip}
	\setlength{\abovedisplayshortskip}{0pt}
	\setlength{\belowdisplayshortskip}{3pt}
	\begin{align*}
	\frac{1}{n}\left|\sum_{s\in\mathcal{S}}\ln f(\varepsilon_s(\boldsymbol{\theta}_1))-\sum_{s\in\mathcal{S}}\ln f(\varepsilon_s(\boldsymbol{\theta}_2))\right| = ||\boldsymbol{\theta}_1-\boldsymbol{\theta}_2|| O_p(1)
	\end{align*}}%
So $\frac{1}{n}\left|\sum_{s\in\mathcal{S}}\ln f(\varepsilon_s(\boldsymbol{\theta}))\right|$ is equicontinuous for $\boldsymbol{\theta}\in\boldsymbol{\Theta}$.
With the pointwise convergence and equicontinuity, we can conclude the uniform convergence in (\ref{uni-cvg-eqn-2}) and furthermore (\ref{uni-cvg-eqn}) follows. 
\end{proof}
We now give a formal statement of consistency of the maximum likelihood estimator $\hat{\boldsymbol{\theta}}_n$.
\begin{thm}
Given Assumptions \ref{compact-space}-\ref{error-dominance}, $\hat{\boldsymbol{\theta}}_n-\boldsymbol{\theta}_0\xrightarrow{p}0$ as $n\rightarrow\infty$.
\end{thm}
\begin{proof}
Similar to the proof in Lung-fei Lee \cite{lee2004asymptotic}, we need to show the stochastic equicontinuity of $\frac{1}{n}\ln |I_n-\rho W_n|$ to have the uniform convergence of the log likelihood function $\mathcal{L}_n(\boldsymbol{\theta})$. Applying the mean value theorem, 
{\small\setlength{\abovedisplayskip}{3pt}
	\setlength{\belowdisplayskip}{\abovedisplayskip}
	\setlength{\abovedisplayshortskip}{0pt}
	\setlength{\belowdisplayshortskip}{3pt}
	\begin{align*}
	\left|\frac{1}{n}(\ln |I_n-\rho_1W_n|-\ln |I_n-\rho_2W_n|)\right| = \left|(\rho_1-\rho_2)\frac{1}{n}tr(W_n(I_n-\tilde{\rho}_nW_n)^{-1})\right|
	\end{align*}} %
where $\tilde{\rho}_n$ is between $\rho_1$ and $\rho_2$. Since $W_n$ is a row standardized matrix, the row sum equals to 1. By Assumption \ref{rho-range} and \ref{ub-weight}, $\sup_{\rho\in\boldsymbol{\Theta}}|\rho|<1$, $W_n$ is bounded in both row and column sums uniformly and using (\ref{matrix-decomposition}), 
{\small\setlength{\abovedisplayskip}{3pt}
	\setlength{\belowdisplayskip}{\abovedisplayskip}
	\setlength{\abovedisplayshortskip}{0pt}
	\setlength{\belowdisplayshortskip}{3pt}
	\begin{align*}
   \left|\frac{1}{n}tr(W_n(I_n-\tilde{\rho}_nW_n)^{-1})\right| =\left|\frac{1}{n}\sum_{i=1}^{n}\frac{\tau_i}{1-\tilde{\rho}_n\tau_i}\right|\leq C_{1}
	\end{align*}} %
where $C_{1}$ is a constant not depending on $n$. So $	\left|\frac{1}{n}(\ln |I_n-\rho_1W_n|-\ln |I_n-\rho_2W_n|)\right| \leq |\rho_1-\rho_2|C_1$ and with Lemma \ref{uni-cvg} we can conclude the uniform convergence that
{\small\setlength{\abovedisplayskip}{3pt}
	\setlength{\belowdisplayskip}{\abovedisplayskip}
	\setlength{\abovedisplayshortskip}{0pt}
	\setlength{\belowdisplayshortskip}{3pt}
	\begin{align}\label{uni-cvg-loglikelihood}
	\sup_{\boldsymbol{\theta}\in\boldsymbol{\Theta}}\left|\frac{1}{n}\mathcal{L}_n(\boldsymbol{\theta}|Y_n, X_n)-\mathbb{E} \frac{1}{n}\mathcal{L}_n(\boldsymbol{\theta}|Y_n, X_n)\right| \xrightarrow{p} 0.
	\end{align}} %
With Assumptions \ref{compact-space}-\ref{error-dominance}, the parameter space $\boldsymbol{\Theta}$ is compact; $\frac{1}{n}\mathcal{L}_n(\boldsymbol{\theta}|Y_n, X_n)$ is continuous in $\boldsymbol{\theta}\in \boldsymbol{\Theta}$ and is a measurable function of $Y_n$, $X_n$ for all $\boldsymbol{\theta}\in \boldsymbol{\Theta}$. $\mathbb{E}\frac{1}{n}\mathcal{L}_n(\boldsymbol{\theta}|Y_n, X_n)$ is continuous on $\boldsymbol{\Theta}$ and by Lemma \ref{uniq-max}, $\mathbb{E}\frac{1}{n}\mathcal{L}_n(\boldsymbol{\theta}|Y_n, X_n)$ has a unique maximum at $\boldsymbol{\theta}_0$. Referring to Theorem 3.5 in White\cite{white1994parametric}, with the uniform convergence in (\ref{uni-cvg-loglikelihood}), we can conclude that $\hat{\boldsymbol{\theta}}_n-\boldsymbol{\theta}_0\xrightarrow{p}0$ as $n\rightarrow\infty$.
\end{proof}

\subsection{Asymptotic Distribution}
\begin{ass}\label{Amatrix-limit}
	The limit $A(\boldsymbol{\theta}_0)=-\lim_{n\rightarrow\infty}\mathbb{E}\frac{1}{n}\frac{\partial^2 \mathcal{L}_n(\boldsymbol{\theta}_0)}{\partial \boldsymbol{\theta}\partial \boldsymbol{\theta}^{\prime}}$ is nonsingular.
\end{ass}
\begin{ass}\label{Bmatrix-limit}
The limit $B(\boldsymbol{\theta}_0)=\lim_{n\rightarrow\infty}\mathbb{E}\frac{1}{n}\frac{\partial \mathcal{L}_n(\boldsymbol{\theta}_0)}{\partial \boldsymbol{\theta}}\frac{\partial \mathcal{L}_n(\boldsymbol{\theta}_0)}{\partial \boldsymbol{\theta}^{\prime}}$ is nonsingular.
\end{ass}
These assumptions are to guarantee the existence of the covariance matrix of the limiting distribution of parameters in a PSAR-ANN model. 
We now give the asymptotic distribution of the maximum likelihood estimator $\hat{\boldsymbol{\theta}}_n$. 
\begin{thm}\label{asythm}
Under Assumptions \ref{compact-space}-\ref{Bmatrix-limit},
{\small\setlength{\abovedisplayskip}{3pt}
  \setlength{\belowdisplayskip}{\abovedisplayskip}
  \setlength{\abovedisplayshortskip}{0pt}
  \setlength{\belowdisplayshortskip}{3pt}
\begin{equation}\label{asymptotic}
\sqrt{n}(\hat{\boldsymbol{\theta}}_n-\boldsymbol{\theta}_0) \xrightarrow{d} N (\boldsymbol{0}, \boldsymbol{\Omega}_0)
\end{equation}} %
where $\boldsymbol{\Omega}_0= A(\boldsymbol{\theta}_0)^{-1}B(\boldsymbol{\theta}_0)A(\boldsymbol{\theta}_0)^{-1} = A(\boldsymbol{\theta}_0)^{-1}$.
\end{thm}
\begin{proof}
Since $\hat{\boldsymbol{\theta}}_n$ maximizes $\mathcal{L}_n(\boldsymbol{\theta})$, $\frac{\partial \mathcal{L}_n(\hat{\boldsymbol{\theta}}_n)}{\partial \boldsymbol{\theta}} = 0$. By the mean value theorem, expand $\frac{\partial \mathcal{L}_n(\hat{\boldsymbol{\theta}}_n)}{\partial \boldsymbol{\theta}}$ around $\boldsymbol{\theta}_0$ with respect to $\boldsymbol{\theta}$,
{\small\setlength{\abovedisplayskip}{3pt}
	\setlength{\belowdisplayskip}{\abovedisplayskip}
	\setlength{\abovedisplayshortskip}{0pt}
	\setlength{\belowdisplayshortskip}{3pt}
	\begin{gather*}
	\frac{\partial \mathcal{L}_n(\hat{\boldsymbol{\theta}}_n)}{\partial \boldsymbol{\theta}}  = \frac{\partial \mathcal{L}_n(\boldsymbol{\theta}_0)}{\partial \boldsymbol{\theta}}+\frac{\partial^2 \mathcal{L}_n(\tilde{\boldsymbol{\theta}}_n)}{\partial \boldsymbol{\theta}\partial \boldsymbol{\theta}^{\prime}}(\hat{\boldsymbol{\theta}}_n-\boldsymbol{\theta}_0)\\
	0 = \frac{\partial \mathcal{L}_n(\boldsymbol{\theta}_0)}{\partial \boldsymbol{\theta}}+\frac{\partial^2 \mathcal{L}_n(\tilde{\boldsymbol{\theta}}_n)}{\partial \boldsymbol{\theta}\partial \boldsymbol{\theta}^{\prime}}(\hat{\boldsymbol{\theta}}_n-\boldsymbol{\theta}_0)
		\end{gather*}} %
where $\tilde{\boldsymbol{\theta}}_n$ is between $\hat{\boldsymbol{\theta}}_n$ and $\boldsymbol{\theta}_0$. Therefore, we can have the following equation:
{\small\setlength{\abovedisplayskip}{3pt}
	\setlength{\belowdisplayskip}{\abovedisplayskip}
	\setlength{\abovedisplayshortskip}{0pt}
	\setlength{\belowdisplayshortskip}{3pt}
	\begin{gather}\label{taylor}	
	\sqrt{n}(\hat{\boldsymbol{\theta}}_n-\boldsymbol{\theta}_0) = \left[-\frac{1}{n}\frac{\partial^2 \mathcal{L}_n(\tilde{\boldsymbol{\theta}}_n)}{\partial \boldsymbol{\theta}\partial \boldsymbol{\theta}^{\prime}}\right]^{-1}\frac{1}{\sqrt{n}}\frac{\partial\mathcal{L}_n(\boldsymbol{\theta}_0)}{\partial \boldsymbol{\theta}}
	\end{gather}} %
We first show the limiting distribution of $\frac{1}{\sqrt{n}}\frac{\partial\mathcal{L}_n(\boldsymbol{\theta}_0)}{\partial \boldsymbol{\theta}}$. Under $\boldsymbol{\theta}_0$, $\boldsymbol{\varepsilon}_n(\boldsymbol{\theta}_0)=\boldsymbol{\varepsilon}_n$,
{\small\setlength{\abovedisplayskip}{3pt}
	\setlength{\belowdisplayskip}{\abovedisplayskip}
	\setlength{\abovedisplayshortskip}{0pt}
	\setlength{\belowdisplayshortskip}{3pt}
	\begin{align}\label{random_error}
	\boldsymbol{\varepsilon}_n(\boldsymbol{\theta}_0)=(I_n-\rho_0W_n)Y_n-X_n\beta_0-F(X_n\boldsymbol{\gamma}^{\prime}_0)\lambda_0=\boldsymbol{\varepsilon}_n
	\end{align}} %
Denote $\frac{\boldsymbol{f}^{\prime}(\boldsymbol{\varepsilon}_n(\boldsymbol{\theta}))}{\boldsymbol{f}(\boldsymbol{\varepsilon}_{n}(\boldsymbol{\theta}))}$ as $V_n(\boldsymbol{\theta})\in \mathbb{R}^n$ and $ \frac{\boldsymbol{f}^{\prime}(\boldsymbol{\varepsilon}_{n})}{\boldsymbol{f}(\boldsymbol{\varepsilon}_{n})}$ as $V_n\in \mathbb{R}^n$, then the first order derivatives are
	{\small\setlength{\abovedisplayskip}{3pt}
	\setlength{\belowdisplayskip}{\abovedisplayskip}
	\setlength{\abovedisplayshortskip}{0pt}
	\setlength{\belowdisplayshortskip}{3pt}
	\begin{gather}\label{first-order-derivative}
	\frac{1}{\sqrt{n}}\frac{\partial \mathcal{L}_{n}(\boldsymbol{\theta})}{\partial \boldsymbol{\theta}}=
	\begin{pmatrix*}[l]
	-\frac{1}{\sqrt{n}}\left((W_nY_n)^{\prime}V_n(\boldsymbol{\theta})+tr(W_n(I_n-\rho W_n)^{-1})\right)\\
	-\frac{1}{\sqrt{n}}X_n^{\prime}V_n(\boldsymbol{\theta})\\
	-\frac{1}{\sqrt{n}}(\boldsymbol{F}(X_n\boldsymbol{\gamma}^{\prime}))^{\prime} V_n(\boldsymbol{\theta}))\\
	-\frac{\lambda_1}{\sqrt{n}}X_n^{\prime}(\boldsymbol{F}(X_n\boldsymbol{\gamma}_1)\circ V_n(\boldsymbol{\theta}))\\
	\hspace{9em}\vdots\\
	-\frac{\lambda_h}{\sqrt{n}}X_n^{\prime}(\boldsymbol{F}(X_n\boldsymbol{\gamma}_h)\circ V_n(\boldsymbol{\theta}))\\
	\end{pmatrix*}
	\end{gather}}%
By Lemma \ref{uniq-max}, the true parameter values maximize $\frac{1}{n}\mathbb{E}(\mathcal{L}_n(\boldsymbol{\theta}))$, so $\frac{1}{n}\frac{\partial \mathbb{E}(\mathcal{L}_n(\boldsymbol{\theta}_0))}{\partial \boldsymbol{\theta}}=\boldsymbol{0}$. In (\ref{log-derivative-bound}) and (\ref{res-expansion}), we showed that $\mathbb{E}\left|\frac{\partial \ln f(\varepsilon_s(\boldsymbol{\theta}))}{\partial \boldsymbol{\theta}}\right|$ is dominated by some function not related to $\boldsymbol{\theta}$ and (\ref{loglikelihood-theta-expectation}) indicates that $\mathbb{E}\left|\frac{\partial \ln f(\varepsilon_s(\boldsymbol{\theta}))}{\partial \boldsymbol{\theta}}\right|$ is bounded for interior units in $\mathcal{S}$. Hence, $\mathbb{E}\frac{\partial \ln f(\varepsilon_s(\boldsymbol{\theta}))}{\partial \boldsymbol{\theta}} = \frac{\partial }{\partial \boldsymbol{\theta}}\mathbb{E}\ln f(\varepsilon_s(\boldsymbol{\theta}))
$, it follows that, with $\frac{1}{n}\mathcal{L}_n(\boldsymbol{\theta})=\frac{1}{n}\ln|I_n-\rho_0W_n|+\frac{1}{n}\sum_{s=1}^{n}\ln f(\varepsilon_s(\boldsymbol{\theta}))$, we can have,
{\small\setlength{\abovedisplayskip}{3pt}
	\setlength{\belowdisplayskip}{\abovedisplayskip}
	\setlength{\abovedisplayshortskip}{0pt}
	\setlength{\belowdisplayshortskip}{3pt}
	\begin{gather*}	
	\frac{1}{n}\frac{\partial \mathbb{E}\mathcal{L}_n(\boldsymbol{\theta}_0)}{\partial \boldsymbol{\theta}} = \frac{1}{n}\mathbb{E}\frac{\partial \mathcal{L}_(\boldsymbol{\theta}_0)}{\partial \boldsymbol{\theta}}=\boldsymbol{0}
	\end{gather*}} %
Therefore, with Assumption \ref{Bmatrix-limit}
{\small\setlength{\abovedisplayskip}{3pt}
	\setlength{\belowdisplayskip}{\abovedisplayskip}
	\setlength{\abovedisplayshortskip}{0pt}
	\setlength{\belowdisplayshortskip}{3pt}
	\begin{gather*}
	Var\left(\frac{1}{\sqrt{n}}\frac{\partial\mathcal{L}_n(\boldsymbol{\theta}_0)}{\partial \boldsymbol{\theta}}\right) = -\mathbb{E}\frac{1}{n}\frac{\partial^2 \mathcal{L}_n(\boldsymbol{\theta}_0)}{\partial\boldsymbol{\theta}\partial\boldsymbol{\theta}^{\prime}} =\mathbb{E} \frac{1}{n} \frac{\partial\mathcal{L}_n(\boldsymbol{\theta}_0)}{\partial \boldsymbol{\theta}}\frac{\partial\mathcal{L}_n(\boldsymbol{\theta}_0)}{\partial \boldsymbol{\theta}^{\prime}}\rightarrow B(\boldsymbol{\theta}_0)
	\end{gather*}} %
And under this $A(\boldsymbol{\theta}_0) = B(\boldsymbol{\theta}_0)$ when $n\rightarrow\infty$. From (\ref{first-order-derivative}), we can see that $\frac{\partial \mathcal{L}_n(\boldsymbol{\theta}_0)}{\partial \boldsymbol{\theta}}$ is a sum of $n$ independent and identically distributed random variables. By the central limit theorem,
with the existence of high order moments of random errors in Assumption \ref{error-dist}, we can conclude the limiting distribution of $\frac{1}{\sqrt{n}}\frac{\partial\mathcal{L}_n(\boldsymbol{\theta}_0)}{\partial \boldsymbol{\theta}}$ is $N(\boldsymbol{0}, B(\boldsymbol{\theta}_0))$.

Next, we want to show that $\frac{1}{n}\frac{\partial^2 \mathcal{L}_n(\tilde{\boldsymbol{\theta}}_n)}{\partial \boldsymbol{\theta}\partial \boldsymbol{\theta}^{\prime}}-\frac{1}{n}\frac{\partial^2 \mathcal{L}_n(\boldsymbol{\theta}_0)}{\partial \boldsymbol{\theta}\partial \boldsymbol{\theta}^{\prime}}\xrightarrow{p}0$. Following the results in (\ref{first-order-derivative}), define $U_n(\boldsymbol{\theta})=\frac{\boldsymbol{f}^{\prime\prime}(\boldsymbol{\varepsilon}_{n}(\boldsymbol{\theta}))}{\boldsymbol{f}(\boldsymbol{\varepsilon}_{n}(\boldsymbol{\theta}))}-\frac{\boldsymbol{f}^{\prime2}(\boldsymbol{\varepsilon}_{n}(\boldsymbol{\theta}))}{\boldsymbol{f}^2(\boldsymbol{\varepsilon}_{n}(\boldsymbol{\theta}))}\in \mathbb{R}^n$ and $U_n(\boldsymbol{\theta}_0)=U_n$ so the second order derivatives are given below $-\frac{1}{n}\frac{\partial^2\mathcal{L}_{n}(\boldsymbol{\theta})}{\partial\boldsymbol{\theta}\partial\boldsymbol{\theta}^{\prime}}=$
 {\small\setlength{\abovedisplayskip}{3pt}
	\setlength{\belowdisplayskip}{\abovedisplayskip}
	\setlength{\abovedisplayshortskip}{0pt}
	\setlength{\belowdisplayshortskip}{3pt}		
	\begin{gather}\label{second-order-derivative}
	\frac{1}{n}
	\begin{pmatrix*}[l]
	G_{0}(\boldsymbol{\theta})&(W_nY_n)^{\prime}G_{1}(\boldsymbol{\theta})&(W_nY_n)^{\prime}G_{2}(\boldsymbol{\theta})&(W_nY_n)^{\prime}H_{1}(\boldsymbol{\theta})&\cdots&(W_nY_n)^{\prime}H_{h}(\boldsymbol{\theta})\\
	G_{1}^{\prime}(\boldsymbol{\theta})W_nY_n&X_n^{\prime}G_{1}(\boldsymbol{\theta})&X_n^{\prime}G_{2}(\boldsymbol{\theta})&X_n^{\prime}H_{1}(\boldsymbol{\theta})&\cdots&X_n^{\prime}H_{h}(\boldsymbol{\theta})\\
	G_{2}^{\prime}(\boldsymbol{\theta})W_nY_n&G_{2}^{\prime}(\boldsymbol{\theta})X_n&\boldsymbol{F}(X_n\boldsymbol{\gamma}^{\prime})^{\prime}G_{2}(\boldsymbol{\theta})&\boldsymbol{F}(X_n\boldsymbol{\gamma}^{\prime})^{\prime}H_{1}(\boldsymbol{\theta})&\cdots&\boldsymbol{F}(X_n\boldsymbol{\gamma}^{\prime})^{\prime}H_{h}(\boldsymbol{\theta})\\
	&&&+K_{1}(\boldsymbol{\theta})&\cdots&+K_{h}(\boldsymbol{\theta})\\
	H_{1}^{\prime}(\boldsymbol{\theta})W_nY_n&H_{1}^{\prime}(\boldsymbol{\theta})X_n&H_{1}^{\prime}(\boldsymbol{\theta})\boldsymbol{F}(X_n\boldsymbol{\gamma}^{\prime})\\
	&&+K_{1}(\boldsymbol{\theta})^{\prime}&&\\
	\vdots&\vdots&\vdots&&J(\boldsymbol{\theta})&\\
	H_{h}^{\prime}(\boldsymbol{\theta})W_nY_n&H_{h}^{\prime}(\boldsymbol{\theta})X_n&H_{h}^{\prime}(\boldsymbol{\theta})\boldsymbol{F}(X_n\boldsymbol{\gamma}^{\prime})\\
	&&+K_{h}(\boldsymbol{\theta})^{\prime}&&\\
	\end{pmatrix*}
	\end{gather}}%
	 {\small\setlength{\abovedisplayskip}{3pt}
		\setlength{\belowdisplayskip}{\abovedisplayskip}
		\setlength{\abovedisplayshortskip}{0pt}
		\setlength{\belowdisplayshortskip}{3pt}		
		\begin{align*}
	J_{ij}(\boldsymbol{\theta}) &=\left\{\begin{array}{ll}
		\lambda_{i}X_n^{\prime}(\boldsymbol{F}^{\prime\prime}(X_n\boldsymbol{\gamma}_{i})\circ V_n(\boldsymbol{\theta})\circ  X_n)+\lambda_{i}X_n^{\prime}(\boldsymbol{F}^{\prime}(X_n\boldsymbol{\gamma}_{i})\circ H_{i})&i=j\\
	\lambda_{i}(\boldsymbol{F}^{\prime}(X_n\boldsymbol{\gamma}_{i})\circ H_{j})^{\prime}X_n&i>j \quad i,j ={1,2,\ldots,h}\\
	\lambda_{i}X_n^{\prime}(\boldsymbol{F}^{\prime}(X_n\boldsymbol{\gamma}_{i})\circ H_{j})&i< j
	\end{array}\right.
	\\\nonumber
	G_{0}(\boldsymbol{\theta}) &=\left(-W_nY_n\circ W_nY_n)^{\prime}U_n(\boldsymbol{\theta})+tr((W_n(I_n-\rho W_n)^{-1})^2\right)\\\nonumber
	G_{1}(\boldsymbol{\theta})&=-U_n(\boldsymbol{\theta})\circ X_n\\\nonumber
	G_{2}(\boldsymbol{\theta})&=-U_n(\boldsymbol{\theta})\circ \boldsymbol{F}(X_n\boldsymbol{\gamma}^{\prime})\\\nonumber
	H_{i}(\boldsymbol{\theta})&=-U_n(\boldsymbol{\theta})\circ (\lambda_i\boldsymbol{F}^{\prime}(X_n\boldsymbol{\gamma}_i)\circ X_n)\quad i=1,\ldots,h\\
	K_{i}(\boldsymbol{\theta}) &= [V_n(\boldsymbol{\theta})\circ \boldsymbol{F}^{\prime}(X_n\boldsymbol{\gamma}^{\prime})]^{\prime}X_n \circ e_{i}\quad i=1,\ldots,h\quad k = 1,\ldots, h\\
	e_{i,k} & = \left\{\begin{array}{ll}
		1&k=i\\
		0& k\neq i
		\end{array}\right.
	\end{align*}}%
	
Since $\tilde{\boldsymbol{\theta}}_n$ is between $\hat{\boldsymbol{\theta}}_n$ and $\boldsymbol{\theta}_0$, $\hat{\boldsymbol{\theta}}_n\xrightarrow{p}\boldsymbol{\theta}_0$ so $\tilde{\boldsymbol{\theta}}_n$ also converges to $\boldsymbol{\theta}_0$ in probability as $n\rightarrow\infty$.
By Assumption \ref{error-dominance}, $\left|\frac{f^{\prime}(s)}{f(s)}\right|,\left|\frac{f^{\prime\prime}(s)}{f(s)}\right|$ and $\left|\frac{f^{\prime 2}(s)}{f^{2}(s)}\right|$ are continuous and are bounded by $a_1+a_2\left|s\right|^{c_1}$ so $V_n(\boldsymbol{\theta}), U_n(\boldsymbol{\theta})$ are continuous. With $\rho\in(-\frac{1}{\tau},\frac{1}{\tau})$, $tr((W_n(I_n-\rho W_n)^{-1})^2)= \sum_{i=1}^{n}\frac{\tau_i^2}{(1-\rho \tau_i)^2}$ is also a continuous function of $\rho$.

Therefore elements in $\frac{1}{n}\frac{\partial^2\mathcal{L}_n(\boldsymbol{\theta})}{\partial \boldsymbol{\theta}\partial\boldsymbol{\theta}^{\prime}}$ are continuous functions for $\boldsymbol{\theta}$ in $\boldsymbol{\Theta}$.
By the continuity,
{\small\setlength{\abovedisplayskip}{3pt}
	\setlength{\belowdisplayskip}{\abovedisplayskip}
	\setlength{\abovedisplayshortskip}{0pt}
	\setlength{\belowdisplayshortskip}{3pt}
	\begin{gather}\label{cvg-function-p}
	\frac{1}{n}\frac{\partial^2\mathcal{L}_n(\tilde{\boldsymbol{\theta}}_n)}{\partial \boldsymbol{\theta}\partial\boldsymbol{\theta}^{\prime}}-\frac{1}{n}\frac{\partial^2\mathcal{L}_n(\boldsymbol{\theta}_0)}{\partial \boldsymbol{\theta}\partial\boldsymbol{\theta}^{\prime}}\xrightarrow{p}0,\quad\text{as } \tilde{\boldsymbol{\theta}}_n\xrightarrow{p}\boldsymbol{\theta}_0
	\end{gather}} %

Finally, show that $\left|\frac{1}{n}\frac{\partial^2\mathcal{L}_n(\boldsymbol{\theta}_0)}{\partial \boldsymbol{\theta}\partial \boldsymbol{\theta}^{\prime}}-\mathbb{E} \frac{1}{n}\frac{\partial^2\mathcal{L}_n(\boldsymbol{\theta}_0)}{\partial \boldsymbol{\theta}\partial\boldsymbol{\theta}^{\prime}}\right|\xrightarrow{p}0$. 

Since $Y_n,X_n$ are stationary, we can first show that for each $s$,
{\small\setlength{\abovedisplayskip}{3pt}
	\setlength{\belowdisplayskip}{\abovedisplayskip}
	\setlength{\abovedisplayshortskip}{0pt}
	\setlength{\belowdisplayshortskip}{3pt}
	\begin{gather}\label{second-derivative-expectation}
 \mathbb{E}\left|\frac{\partial^2}{\partial\boldsymbol{\theta}\partial\boldsymbol{\theta}^{\prime}}\left(\frac{1}{n}\sum_{i=1}^{n}\ln(1-\rho_0\tau_i)+\ln f(\varepsilon_s(\boldsymbol{\theta}_0))\right)\right| <\infty
\end{gather}} %
We first discuss the expected value of second derivative with respect to $\rho$ in (\ref{second-derivative-expectation}).
By triangular inequality, $ \mathbb{E}\left|\frac{\partial^2}{\partial\rho\partial\rho}\left(\frac{1}{n}\sum_{i=1}^{n}\ln(1-\rho_0\tau_i)+\ln f(\varepsilon_s(\boldsymbol{\theta}_0))\right)\right|<\mathbb{E}\left|\frac{1}{n}\sum_{i=1}^{n}\frac{\partial^2\ln(1-\rho_0\tau_i)}{\partial\rho\partial\rho}\right|+\mathbb{E}\left|\frac{\partial^2 \ln f(\varepsilon_s(\boldsymbol{\theta}_0))}{\partial\rho\partial \rho}\right|$. Because $\sum_{i=1}^{n}\frac{\partial^2\ln(1-\rho_0\tau_i)}{\partial\rho\partial\rho} = tr(M_n^2)$ (defined in (\ref{log-derivative-bound})), this can be further simplified to
{\small\setlength{\abovedisplayskip}{3pt}
	\setlength{\belowdisplayskip}{\abovedisplayskip}
	\setlength{\abovedisplayshortskip}{0pt}
	\setlength{\belowdisplayshortskip}{3pt}
	\begin{gather}\label{rho-rho-abs}
	\frac{1}{n}tr(M_n^2)+\mathbb{E}\left|\left(\frac{f^{\prime^2}(\varepsilon_s)}{f^2(\varepsilon_s)}-\frac{f^{\prime\prime}(\varepsilon_s)}{f(\varepsilon_s)}\right)\left(\sum_{k=1}^{n}w_{sk}y_k\right)^2\right|
\end{gather}} %
Because $M_n$ is uniformly bounded in column and row sums, $\frac{1}{n}tr(M_n^2)<\infty$, $\sum_{k=1}^{n}m_{sk}< b$ so $\sum_{j=1}^{n}\sum_{k=1}^{n}m_{sj}m_{sk} <(\sum_{k=1}^{n}m_{sk})^2 < b^2 $.
We need to show $\mathbb{E}\left|\left(\frac{f^{\prime^2}(\varepsilon_s)}{f^2(\varepsilon_s)}-\frac{f^{\prime\prime}(\varepsilon_s)}{f(\varepsilon_s)}\right)\left(\sum_{k=1}^{n}w_{sk}y_k\right)^2\right|<\infty$.

Because $Y_n = (I_n-\rho_0W_n)^{-1}(\boldsymbol{g}(X_n,\boldsymbol{\theta}_0)+\boldsymbol{\varepsilon}_n)$, $W_nY_n = M_n(\boldsymbol{g}(X_n,\boldsymbol{\theta}_0)+\boldsymbol{\varepsilon}_n)$, $\sum_{k=1}^{n}w_{sk}y_k=\sum_{k=1}^{n}m_{sk}(g(x_k,\boldsymbol{\theta}_0)+\varepsilon_k)$. 
It follows that $\mathbb{E}\left|\left(\frac{f^{\prime^2}(\varepsilon_s)}{f^2(\varepsilon_s)}-\frac{f^{\prime\prime}(\varepsilon_s)}{f(\varepsilon_s)}\right)\left(\sum_{k=1}^{n}w_{sk}y_k\right)^2\right|=$
{\small\setlength{\abovedisplayskip}{3pt}
	\setlength{\belowdisplayskip}{\abovedisplayskip}
	\setlength{\abovedisplayshortskip}{0pt}
	\setlength{\belowdisplayshortskip}{3pt}
	\begin{align*}
  &\mathbb{E}\left|\left(\frac{f^{\prime^2}(\varepsilon_s)}{f^2(\varepsilon_s)}-\frac{f^{\prime\prime}(\varepsilon_s)}{f(\varepsilon_s)}\right)\left(\sum_{k=1}^{n}m_{sk}(g(x_{k},\boldsymbol{\theta}_0)+\varepsilon_k)\right)^2\right|\\
<&\mathbb{E}\left|\left(\frac{f^{\prime^2}(\varepsilon_s)}{f^2(\varepsilon_s)}-\frac{f^{\prime\prime}(\varepsilon_s)}{f(\varepsilon_s)}\right)\sum_{k=1}^{n}m_{sk}^2[g(x_{k},\boldsymbol{\theta}_0)+\varepsilon_k]^2\right|\\
	&+\mathbb{E}\left|\left(\frac{f^{\prime^2}(\varepsilon_s)}{f^2(\varepsilon_s)}-\frac{f^{\prime\prime}(\varepsilon_s)}{f(\varepsilon_s)}\right)\sum_{j=1,j\neq k}^{n}\sum_{k=1}^{n}m_{sk}m_{sj}[g(x_{k},\boldsymbol{\theta}_0)+\varepsilon_k][g(x_{j},\boldsymbol{\theta}_0)+\varepsilon_j]\right|
\end{align*}} %
By assumption, $\mathbb{E}\, \varepsilon_k\varepsilon_j=0$ if $k\neq j$, $\mathbb{E}\left|\frac{\varepsilon_sf^{\prime^2}(\varepsilon_s)}{f^2(\varepsilon_s)}-\frac{\varepsilon_sf^{\prime\prime}(\varepsilon_s)}{f(\varepsilon_s)}\right|<\infty$, $\mathbb{E}\left|\frac{\varepsilon_s^2f^{\prime^2}(\varepsilon_s)}{f^2(\varepsilon_s)}-\frac{\varepsilon_s^2f^{\prime\prime}(\varepsilon_s)}{f(\varepsilon_s)}\right|<\infty$. Through mathematical computation, we can prove that $\mathbb{E}\left|\frac{\partial^2 \ln f(\varepsilon_s(\boldsymbol{\theta}_0))}{\partial\rho\partial\rho}\right|$ is finite, i.e.,
{\small\setlength{\abovedisplayskip}{3pt}
	\setlength{\belowdisplayskip}{\abovedisplayskip}
	\setlength{\abovedisplayshortskip}{0pt}
	\setlength{\belowdisplayshortskip}{3pt}
	\begin{gather*} \mathbb{E}\left|\frac{\partial^2}{\partial\rho\partial\rho}\left(\frac{1}{n}\sum_{i=1}^{n}\ln(1-\rho_0\tau_i)+\ln f(\varepsilon_s(\boldsymbol{\theta}_0))\right)\right| <\infty
\end{gather*}}%
Because $\frac{1}{n}\sum_{i=1}^{n}\ln(1-\rho_0\tau_i)$ in (\ref{second-derivative-expectation}) only relates to $\rho$, this term goes away when taken second derivative with respect to other parameters. Hence, other elements in (\ref{second-derivative-expectation}) equal to those in $\mathbb{E}\left|\frac{\partial^2 \ln f(\varepsilon_s(\boldsymbol{\theta}_0))}{\partial\boldsymbol{\theta}\partial \boldsymbol{\theta}^{\prime}}\right|$ and we can show that  those expectations are also finite.
{\small\setlength{\abovedisplayskip}{3pt}
	\setlength{\belowdisplayskip}{\abovedisplayskip}
	\setlength{\abovedisplayshortskip}{0pt}
	\setlength{\belowdisplayshortskip}{3pt}
	\begin{align}\label{rho-beta-abs}
	\mathbb{E}\left|\frac{\partial^2 \ln f(\varepsilon_s(\boldsymbol{\theta}_0))}{\partial\rho\partial \beta^{\prime}}\right| &\leq |x_s^{\prime}|\cdot \left(k_2|m_{ss}|+k_1|b-m_{ss}|\cdot\mathbb{E}|\varepsilon_s|+k_1\left|\sum_{k=1}^{n}m_{sk}g(x_k,\boldsymbol{\theta}_0)\right|\right)\\
	\mathbb{E}\left|\frac{\partial^2 \ln f(\varepsilon_s(\boldsymbol{\theta}_0))}{\partial\rho\partial \lambda^{\prime}}\right| &\leq\boldsymbol{1}_h^{\prime}\cdot \left(k_2|m_{ss}|+k_1|b-m_{ss}|\cdot\mathbb{E}|\varepsilon_s|+k_1\left|\sum_{k=1}^{n}m_{sk}g(x_k,\boldsymbol{\theta}_0)\right|\right)\\
	\mathbb{E}\left|\frac{\partial^2 \ln f(\varepsilon_s(\boldsymbol{\theta}_0))}{\partial\rho\partial \gamma_{i}^{\prime}}\right| &\leq
 \frac{|\lambda_{i0}x_s^{\prime}|}{4}\cdot \left(k_2|m_{ss}|+k_1|b-m_{ss}|\cdot\mathbb{E}|\varepsilon_s|+k_1\left|\sum_{k=1}^{n}m_{sk}g(x_k,\boldsymbol{\theta}_0)\right|\right)\\
	\mathbb{E}\left|\frac{\partial^2 \ln f(\varepsilon_s(\boldsymbol{\theta}_0))}{\partial\beta\partial \beta^{\prime}}\right| &= k_1|x_sx_s^{\prime}|\\
	\mathbb{E}\left|\frac{\partial^2 \ln f(\varepsilon_s(\boldsymbol{\theta}_0))}{\partial\beta\partial \lambda^{\prime}}\right| &= k_1|x_s\boldsymbol{F}(x_s^{\prime}\boldsymbol{\gamma}_0)|\\
	\mathbb{E}\left|\frac{\partial^2 \ln f(\varepsilon_s(\boldsymbol{\theta}_0))}{\partial\beta\partial \boldsymbol{\gamma}_{i}^{\prime}}\right| &\leq \frac{k_1}{4}|\lambda_{i0}x_sx_s^{\prime}|\\
	\mathbb{E}\left|\frac{\partial^2 \ln f(\varepsilon_s(\boldsymbol{\theta}_0))}{\partial\lambda\partial \lambda^{\prime}}\right| & = k_1\left|\boldsymbol{F}(x_s^{\prime}\boldsymbol{\gamma}_0)^{\prime}\boldsymbol{F}(x_s^{\prime}\boldsymbol{\gamma}_0)\right| \leq k_1\cdot \boldsymbol{1}_{h\times h}\\
	\mathbb{E}\left|\frac{\partial^2 \ln f(\varepsilon_s(\boldsymbol{\theta}_0))}{\partial\lambda\partial \boldsymbol{\gamma}_{i}^{\prime}}\right| & = \frac{k_1}{4}|\lambda_{i0}F^{\prime}(x_s^{\prime}\boldsymbol{\gamma}_{i0})|\cdot|\boldsymbol{F}(x_s^{\prime}\boldsymbol{\gamma}_{0})^{\prime}x_s^{\prime}|\leq \frac{k_1|\lambda_{i0}|}{4}\cdot|\boldsymbol{F}(x_s^{\prime}\boldsymbol{\gamma}_{0})^{\prime}x_s^{\prime}|\\
	\mathbb{E}\left|\frac{\partial^2 \ln f(\varepsilon_s(\boldsymbol{\theta}_0))}{\partial\boldsymbol{\gamma}_{i}\partial \boldsymbol{\gamma}_{j}^{\prime}}\right| & \leq \frac{k_1|\lambda_{i0}\lambda_{j0}|}{16}\cdot|x_sx_s^{\prime}|,\quad i\neq j\\\label{gamma-gamma-abs}
	\mathbb{E}\left|\frac{\partial^2 \ln f(\varepsilon_s(\boldsymbol{\theta}_0))}{\partial\boldsymbol{\gamma}_{i}\partial \boldsymbol{\gamma}_{i}^{\prime}}\right| & \leq \frac{k_1\lambda^2_{i0}}{16}\cdot|x_sx_s^{\prime}| +\frac{\sqrt{3}k_0|\lambda_{i0}|}{18}|x_sx_s^{\prime}|
\end{align}} %
With assumptions \ref{compact-space}-\ref{error-dominance}, (\ref{rho-rho-abs})-(\ref{gamma-gamma-abs}) are finite. Then we can apply the ergodic theorem \cite{birkhoff1931proof} and conclude that
{\small\setlength{\abovedisplayskip}{3pt}
	\setlength{\belowdisplayskip}{\abovedisplayskip}
	\setlength{\abovedisplayshortskip}{0pt}
	\setlength{\belowdisplayshortskip}{3pt}
	\begin{align*}
\left|\frac{1}{n}\frac{\partial^2\mathcal{L}_n(\boldsymbol{\theta}_0)}{\partial \boldsymbol{\theta}\partial\boldsymbol{\theta}^{\prime}}-\mathbb{E} \frac{1}{n}\frac{\partial^2\mathcal{L}_n(\boldsymbol{\theta}_0)}{\partial \boldsymbol{\theta}\partial\boldsymbol{\theta}^{\prime}}\right|\xrightarrow{p}0
	\end{align*}} %
We have proved that $\left|\frac{1}{n}\frac{\partial^2 \mathcal{L}_n(\tilde{\boldsymbol{\theta}}_n)}{\partial \boldsymbol{\theta}\partial\boldsymbol{\theta}^{\prime}}-\frac{1}{n}\frac{\partial^2\mathcal{L}_n(\boldsymbol{\theta}_0)}{\partial \boldsymbol{\theta}\partial\boldsymbol{\theta}^{\prime}}\right|\xrightarrow{p}\boldsymbol{0}$ so it is trivial that
 {\small\setlength{\abovedisplayskip}{3pt}
 	\setlength{\belowdisplayskip}{\abovedisplayskip}
 	\setlength{\abovedisplayshortskip}{0pt}
 	\setlength{\belowdisplayshortskip}{3pt}
 	\begin{gather}\label{cvg-3}
 	\left|\frac{1}{n}\frac{\partial^2\mathcal{L}_n(\tilde{\boldsymbol{\theta}}_n)}{\partial \boldsymbol{\theta}\partial\boldsymbol{\theta}^{\prime}}-\mathbb{E} \frac{1}{n}\frac{\partial^2\mathcal{L}_n(\boldsymbol{\theta}_0)}{\partial \boldsymbol{\theta}\partial\boldsymbol{\theta}^{\prime}}\right|\xrightarrow{p}0
 	\end{gather}}%
Recall the equation (\ref{taylor}), we have proved that $\frac{1}{\sqrt{n}}\frac{\partial \mathcal{L}_n(\boldsymbol{\theta}_0)}{\partial \boldsymbol{\theta}}$ has the limiting distribution $N(\boldsymbol{0}, B(\boldsymbol{\theta}_0))$. With (\ref{cvg-3}), for $\tilde{\boldsymbol{\theta}}_n$ between $\hat{\boldsymbol{\theta}}_n$ and $\boldsymbol{\theta}_0$,  $-\frac{1}{n}\frac{\partial^2\mathcal{L}_n(\tilde{\boldsymbol{\theta}}_n)}{\partial \boldsymbol{\theta}\partial\boldsymbol{\theta}^{\prime}}\xrightarrow{p}A(\boldsymbol{\theta}_0)$ so we can conclude that $\sqrt{n} (\hat{\boldsymbol{\theta}}_n-\boldsymbol{\theta}_0) \xrightarrow{d} N (\boldsymbol{0}, \boldsymbol{\Omega}_0)$, where $\boldsymbol{\Omega}_0 = A^{-1}(\boldsymbol{\theta}_0)B(\boldsymbol{\theta}_0)A^{-1}(\boldsymbol{\theta}_0)$.
\end{proof}

\section{Numerical Results}
\subsection{Simulation Study}
In this section, we conduct simulation experiments to examine the estimators' behavior for finite samples. For estimation purposes it is often useful to reparametrize the logistic function $F(x_s^{\prime}\boldsymbol{\gamma}_i)$ as
{\small\setlength{\abovedisplayskip}{3pt}
	\setlength{\belowdisplayskip}{\abovedisplayskip}
	\setlength{\abovedisplayshortskip}{0pt}
	\setlength{\belowdisplayshortskip}{3pt}
	\begin{equation}\label{logistic_repara}
	F\left(||\boldsymbol{\gamma}_i||\cdot x_s^{\prime}\frac{\boldsymbol{\gamma}_i}{||\boldsymbol{\gamma}_i||}\right) = \left(1+e^{-||\boldsymbol{\gamma}_i||\cdot x_s^{\prime}\frac{\boldsymbol{\gamma}_i}{||\boldsymbol{\gamma}_i||}}\right)^{-1},\, i=1,\ldots,h
	\end{equation}} %
where $||\boldsymbol{\gamma}_i||, i =1,\ldots, h$ is the $L_2$-norm of  $\boldsymbol{\gamma}_i$.
We use a univariate exogenous variable and let $X_n=(x_1,\ldots,x_n)^{\prime}$.
 For illustration, we only include the nonlinear component of $X_n$. Usually we would like to normalize predictors before fitting a neural network model to avoid the computation overflow \cite{medeiros2006building} so we add a centralizing constant $\gamma_0$ in this simulation. The model becomes
{\small\setlength{\abovedisplayskip}{3pt}
  \setlength{\belowdisplayskip}{\abovedisplayskip}
  \setlength{\abovedisplayshortskip}{0pt}
  \setlength{\belowdisplayshortskip}{3pt}
\begin{equation}\label{simulation-model}
y_s = \rho\sum_{i=1}^{n}w_{si}y_i + \lambda F(\gamma_1(x_s-\gamma_0)) + \varepsilon_s
\end{equation}} %
\vspace{-20pt}

\noindent
For identification reasons mentioned in Restriction 1 and 2, we impose $\gamma_1>0$. 

We sample $n=2500, 4900$ random errors respectively from three distributions (standard normal, rescaled t-distribution and Laplace distribution) with variance 1 and $X$ is a univariate exogenous variable, values of which sampled from a normal distribution $N(0.5,3^2)$. We set the true parameters to be $\rho_0 = 0.6$, $\lambda_0 = 5$, and weights in the neural net $\gamma_{00} = 0.5$, $\gamma_{10} = 1$. 
The log-likelihood function $\mathcal{L}_n(\boldsymbol{\theta})$ is given in (\ref{likelihood-simulation}) and we use L-BFGS-B method\cite{byrd1995limited,zhu1997algorithm} (recommended for bound constrained optimization) to find the parameter estimates $\hat{\boldsymbol{\theta}}$ which maximize (\ref{likelihood-simulation}).
{\small\setlength{\abovedisplayskip}{3pt}
	\setlength{\belowdisplayskip}{\abovedisplayskip}
	\setlength{\abovedisplayshortskip}{0pt}
	\setlength{\belowdisplayshortskip}{3pt}
	\begin{align}\label{likelihood-simulation}
	\mathcal{L}_n(\boldsymbol{\theta}) &= \ln |I_n-\rho W_n|+\sum_{s=1}^{n}\ln f(\varepsilon_s(\boldsymbol{\theta}))\\
	\varepsilon_s(\boldsymbol{\theta}) &= y_s-\rho\sum_{s=1}^{n}w_{si}y_i-x_s\beta-\lambda F(\gamma_1(x_s-\gamma_0))
	\end{align}} %
For the model under consideration, we estimated the covariance of the asymptotic normal distribution equation (\ref{asymptotic}). 
Since matrices $A(\boldsymbol{\theta}_0)$ and $B(\boldsymbol{\theta}_0)$ involve expected values with respect to the true parameter $\boldsymbol{\theta}_0$, given merely observations, in practice they can be estimated as follows:
{\small\setlength{\abovedisplayskip}{3pt}
  \setlength{\belowdisplayskip}{\abovedisplayskip}
  \setlength{\abovedisplayshortskip}{0pt}
  \setlength{\belowdisplayshortskip}{3pt}
\begin{align*}
\hat{A}(\hat{\boldsymbol{\theta}}) &= \frac{1}{n}\sum_{s=1}^{n}-\frac{\partial^2 l(\hat{\boldsymbol{\theta}}|x_s,y_s)}{\partial \boldsymbol{\theta}\partial \boldsymbol{\theta}^{\prime}}\\
\hat{B}(\hat{\boldsymbol{\theta}}) &= \frac{1}{n}\sum_{s=1}^{n}\frac{\partial l(\hat{\boldsymbol{\theta}}|x_s,y_s)}{\partial \boldsymbol{\theta}}\frac{\partial l(\hat{\boldsymbol{\theta}}|x_s,y_s)}{\partial \boldsymbol{\theta}^{\prime}}
\end{align*}} %
where
{\small\setlength{\abovedisplayskip}{3pt}
  \setlength{\belowdisplayskip}{\abovedisplayskip}
  \setlength{\abovedisplayshortskip}{0pt}
  \setlength{\belowdisplayshortskip}{3pt}
\begin{align*}
l(\boldsymbol{\theta}|x_s,y_s) = \frac{1}{n}\ln |I_n-\rho W_n|+ \ln f(\varepsilon_s(\boldsymbol{\theta}))
\end{align*}} %
Using (\ref{first-order-derivative}) and (\ref{second-order-derivative}), we can calculate $\hat{A}(\boldsymbol{\theta}_0), \hat{B}(\boldsymbol{\theta}_0)$ to assess the asymptotic properties of parameter estimates.
Note that the derivative of the log-likelihood with respect to $\rho$ cannot be calculated directly because it requires taking derivative with respect to a log-determinant of $I_n-\rho W_n$. For small sample sizes, we can compute the determinant directly and get the corresponding derivatives; but for large sample sizes, for example a dataset with $3000$ observations, $W_n$ is a $3000\times 3000$ weight matrix which makes it impossible to calculate the derivative directly. Since $W_n$ is a square matrix, we can apply the spectral decomposition such that $W_n$ can be expressed in terms of its $n$ eigenvalue-eigenvector pairs in (\ref{matrix-decomposition}).
So we can apply the following approach to calculate the derivative of $\ln |I_n-\rho W_n|$, which greatly reduce the burden of computations (Viton \cite{viton2010notes}).
{\small\setlength{\abovedisplayskip}{3pt}
  \setlength{\belowdisplayskip}{\abovedisplayskip}
  \setlength{\abovedisplayshortskip}{0pt}
  \setlength{\belowdisplayshortskip}{3pt}\begin{align*}
\ln |I_n-\rho W_n| = \ln \left(\prod_{s = 1}^{n}(1-\rho \tau_i)\right)
\end{align*}} %
Further the derivatives of the log-likelihood function with respect to $\rho$ is
{\small\setlength{\abovedisplayskip}{3pt}
  \setlength{\belowdisplayskip}{\abovedisplayskip}
  \setlength{\abovedisplayshortskip}{0pt}
  \setlength{\belowdisplayshortskip}{3pt}  
\begin{align*}
\frac{\partial l_s(\boldsymbol{\theta}|x_s,y_s)}{\partial \rho} &= \frac{1}{n}\sum_{i = 1}^{n}\frac{-\tau_i}{(1-\rho \tau_i)}+\{y_s-\rho\sum_{i=1}^{n}w_{si}y_i-\lambda F(\gamma_1(x_s-\gamma_0)\}\cdot\left(\sum_{i=1}^{n}w_{si}y_i\right)\\
\frac{\partial^2 l_s(\boldsymbol{\theta}|x_s,y_s)}{\partial \rho\partial \rho} & = -\frac{1}{n}\sum_{i = 1}^{n}\left[\frac{\tau_i^2}{(1-\rho \tau_i)^2}+\left(\sum_{i=1}^{n}w_{si}y_i\right)^2\right]
\end{align*}} %
Finally we can estimate the covariance matrix by equation (\ref{asy_var}).
{\small\setlength{\abovedisplayskip}{3pt}
  \setlength{\belowdisplayskip}{\abovedisplayskip}
  \setlength{\abovedisplayshortskip}{0pt}
  \setlength{\belowdisplayshortskip}{3pt}
\begin{align}\label{asy_var}
\hat{\boldsymbol{\Omega}} &= \hat{A}^{-1}(\boldsymbol{\theta}_0)\hat{B}(\boldsymbol{\theta}_0)\hat{A}^{-1}(\boldsymbol{\theta}_0)
\end{align}} %
In our simulation study, we computed $\hat{\boldsymbol{\theta}}$ the for 200 replicates for each $n=2500, 4900$.  
The estimate $\hat{\boldsymbol{\Omega}}$ of the asymptotic covariance matrix is computed based on a sample with 10000 simulated observations. Table \ref{simulation-1} compares the empirical mean and standard errors (in parentheses) of parameter estimators with the true value and their asymptotic standard deviations (in squared brackets) respectively. Comparing the simulation results when $\boldsymbol{\varepsilon}$ follows a standard normal distribution with simulation results when $\boldsymbol{\varepsilon}$ follows a $t(4)$ distribution, means of the estimates over 200 replicates are closer to the true values and their empirical standard deviations are smaller when $\boldsymbol{\varepsilon}$ follows the heavy tailed distribution. For all these experiments with different error distributions, the empirical standard deviations of $\hat{\boldsymbol{\theta}}$ are close to the asymptotic standard deviations which implies that the estimators' finite sample behavior roughly matches their asymptotic distributions. Note that when $\boldsymbol{\varepsilon}$ is sampled from a Laplace distribution, this covariance matrix cannot be computed because its second order derivative is not differentiable at $0$. But the simulated $\hat{\boldsymbol{\theta}}$'s still appear consistent properties. Normal plots for parameter estimates are shown in Figure \ref{qqplot} and give a strong indication of normality.  

\setlength{\extrarowheight}{5pt}
\begin{table}[h!]
\centering
\begin{tabular}{lcccc}
\toprule
\multirow{2}{*}{$\varepsilon$} &
      \multicolumn{4}{c}{$n=2500$}
       \\\cmidrule(lr){2-5}
 & $\hat{\rho} $ & $\hat{\lambda} $ & $\hat{\gamma}_0 $ & $\hat{\gamma}_1 $ \\\midrule
 \multirow{3}{*}{$N(0,1)$}
   & 0.6178 & 4.8504 &  0.5410 & 1.0576 \\ [-5pt]
   & (0.0075)& (0.0812)& (0.0425)&(0.0431)\\ [-5pt]
   & [0.0046] & [0.0639] & [0.0417] & [0.0354]\\
\multirow{3}{*}{$t(4)$}& 0.6132 & 4.8952 & 0.5326 & 1.0411 \\[-5pt]
& (0.0060) & (0.0623) & (0.0364) & (0.0320) \\ [-5pt]
& [0.0044] & [0.0562] & [0.0353] & [0.0310]\\
$Laplace$ & 0.6107 & 4.9132 & 0.5283 & 1.0358 \\[-5pt]
$(0, \frac{\sqrt{2}}{2})$& (0.0053) & (0.0562) & (0.0295) & (0.0291) \\\midrule
\multirow{2}{*}{$\varepsilon$} &
      \multicolumn{4}{c}{$n=4900$} \\\cmidrule(lr){2-5}

 & $\hat{\rho} $ & $\hat{\lambda} $ & $\hat{\gamma}_0 $ & $\hat{\gamma}_1 $ \\\midrule
 \multirow{3}{*}{$N(0,1)$}
   & 0.6175 & 4.8617 & 0.5435 & 1.0517\\ [-5pt]
   &(0.0056)&(0.0572)&(0.0297)&(0.0303)\\ [-5pt]
   &[0.0033] & [0.0456] & [0.0298] & [0.0252] \\
\multirow{3}{*}{$t(4)$}& 0.6130 & 4.8957& 0.5312& 1.0380\\[-5pt]
& (0.0051)& (0.0526)& (0.0274)& (0.0246)\\ [-5pt]
& [0.0031] & [0.0426] & [0.0260] & [0.0235]\\
$Laplace$ & 0.6096& 4.9242& 0.5217& 1.0268\\[-5pt]
$(0, \frac{\sqrt{2}}{2})$& (0.0047)& (0.0487)& (0.0239)& (0.0233)\\\bottomrule       
\end{tabular}
\caption{Empirical mean and standard errors (in parentheses) of parameter estimates when $\varepsilon$ is sampled from a standard normal, standardized student t distribution and a Laplace distribution. The asymptotic standard errors are displayed for reference in square brackets.}
\label{simulation-1}
\end{table}
\begin{figure}[h!]
\begin{center}
  \includegraphics[width=10cm]{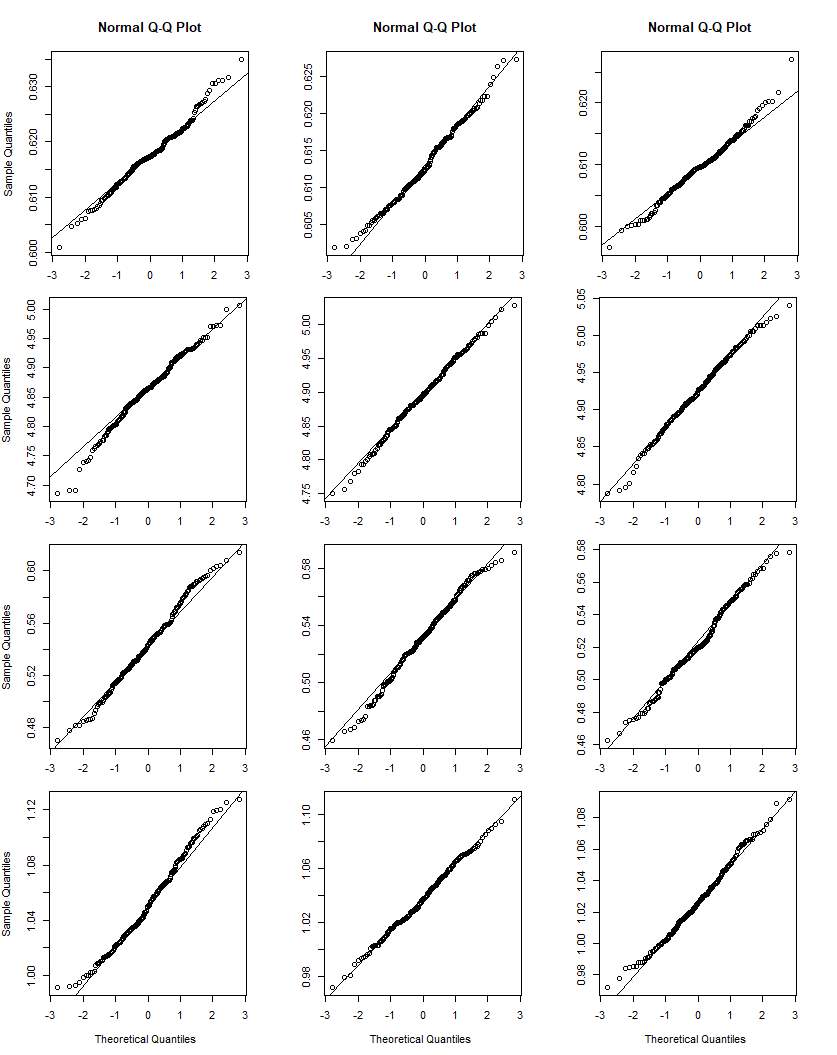}
  \caption{ Normal plots for parameter estimates $\rho$ (1st row), $\lambda$ (2nd row), $\gamma_0$ (3rd row) and $\gamma_1$ (4th row)  when $\varepsilon_s$ follows a standard normal distribution (first column), standardized t distribution (middle column) and Laplace distribution (last column) $n = 70\times 70$ }
  \label{qqplot}
  \end{center}
\end{figure}

\subsection{Real Data Example}
Spatial models have a lot of applications in understanding spatial interactions in  cross-sectional data.
Among them, the study of electoral behavior has attracted considerable attention by political scientists. Poole and Rosenthal \cite{poole1984us} found that the spatial variation plays an important role in presidential electoral dynamics. And mentioned by Braha and de Aguiar (2017 \cite{votingcontagion}), most studies in the U.S. consider vote choices as the result of attitudinal factors such as evaluations of the candidates and government performances as well as social factors such as race, social class, and region.
Inspired by their research, we would like to understand this electoral dynamics using our proposed partially specified spatial autoregressive model and to help identify how  social factors influence people's voting preferences.

Here,  we focus on the proportion of votes cast for U.S. presidential candidates at the county level in 2004. Counties are grouped by state, and let $Y$ be the corresponding fraction of votes (vote-share) in a county for the Democratic candidate (John Kerry) in 2004.
Predictors $X$ are chosen from economic and social factors covering the living standard, economy development and racial distribution.
\begin{figure}[h!]
	\centering
	\includegraphics[width=10cm]{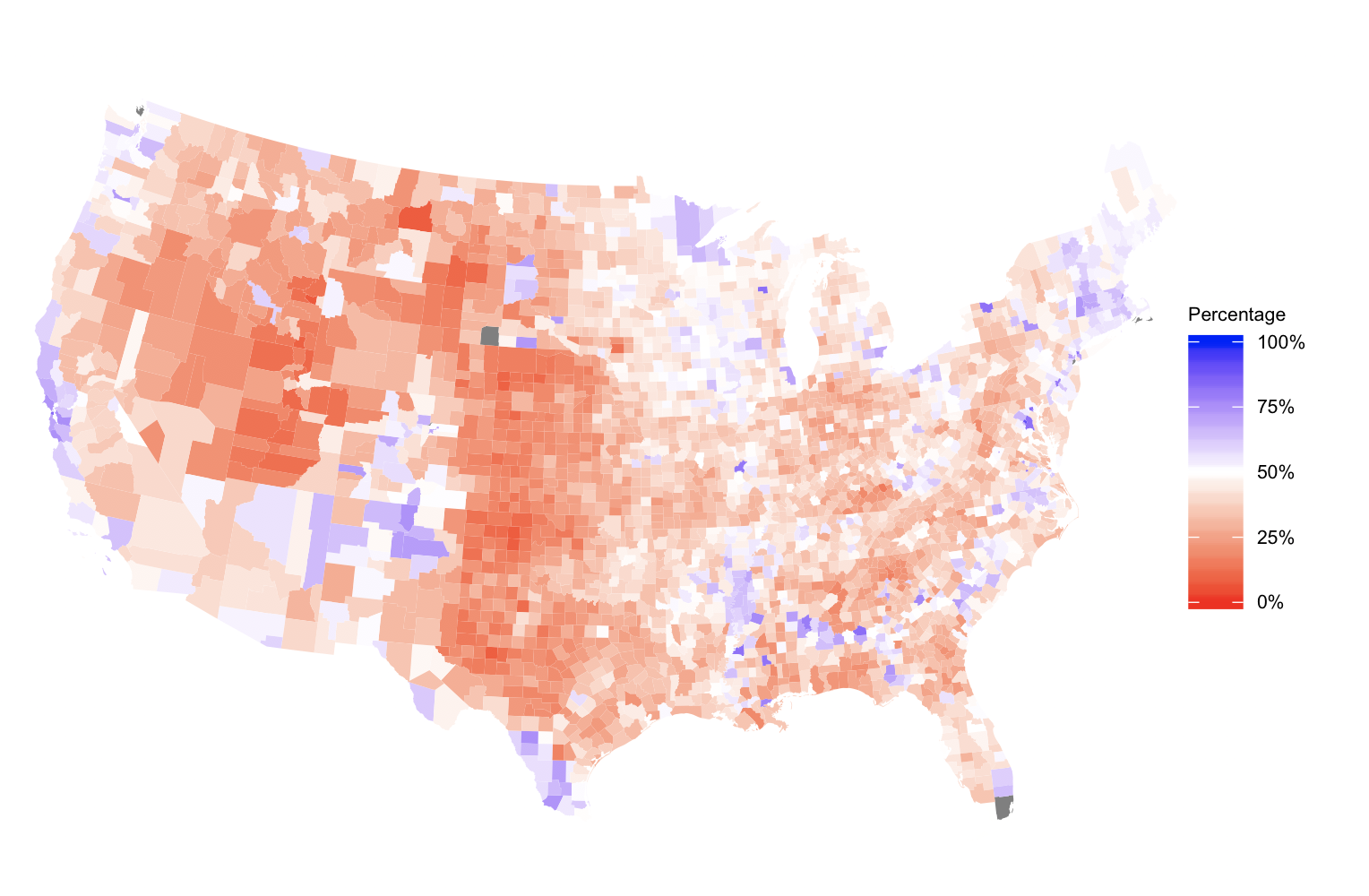}
	\caption{Fractions of Vote-shares per County for Democratic presidential candidate in 2004}
	\label{democracy}
\end{figure}
Figure \ref{democracy} shows the observed values of $Y_n$ for 2004. This heat map exhibits strong correlation between observations in neighboring counties which is supported by Moran's Test on $Y$ (test statistic = $52.4$, P-value $ <2.2\times10^{-16}$). This indicates that $Y$, the fraction of vote-share for Democratic candidate, is not independently distributed across the space. So we consider fitting a spatial model to the data.

In our analysis, we exclude the four U.S. counties with no neighbors (San Juan, Dukes, Nantucket, Richmond) to avoid the non-singularity of our spatial weight matrix $W_n$ in the modeling, so the total number of observations is $n = 3107$.

First we fit a linear regression model to see if it is sufficient to explain the voting dynamic using explanatory variables $X=(X_1,\ldots, X_5)$. From the preliminary analysis fitting $Y$ on all the available variables, we chose the five most significant ones for modeling out of more than 20 different variables. The chosen predictors are percent residents under 18 years $X_{1}$ (\texttt{UNDER18}), percent white residents $X_{2}$ (\texttt{WHITE}), percent residents below poverty line $X_{3}$ (\texttt{pctpoor}),  per capita income $X_{4}$ (\texttt{pcincome}) and USDA urban/rural code $X_{5}$ (\texttt{urbrural},  \texttt{0 = most rural}, \texttt{9 = most urban}). The corresponding least-squares line is as follows:
{\small\setlength{\abovedisplayskip}{3pt}
	\setlength{\belowdisplayskip}{\abovedisplayskip}
	\setlength{\abovedisplayshortskip}{0pt}
	\setlength{\belowdisplayshortskip}{3pt}
	\begin{align}\label{linear-model}
	\hat{Y}= 80.4 - 0.932X_{1} - 0.250X_{2}+0.324X_{3}+2.76\times 10 ^{-5}X_{4}-1.24X_{5}
	\end{align}} 
These six parameter estimates are all significant at $\alpha = 0.05$ and by looking at signs, it is easy to tell how these covariates relate to the voting behaviors. However, one major drawback of this linear model is that the fitted residuals are still correlated across the space (null hypothesis of independence rejected in Moran's Test, test statistic = 54.1, P-value $ <2.2\times10^{-16}$ ; see Figure \ref{res_lm}) so a multiple linear regression fails to adequately describe the spatial dependence in $Y$.
\begin{figure}[h!]
	\centering
	\includegraphics[width=10cm]{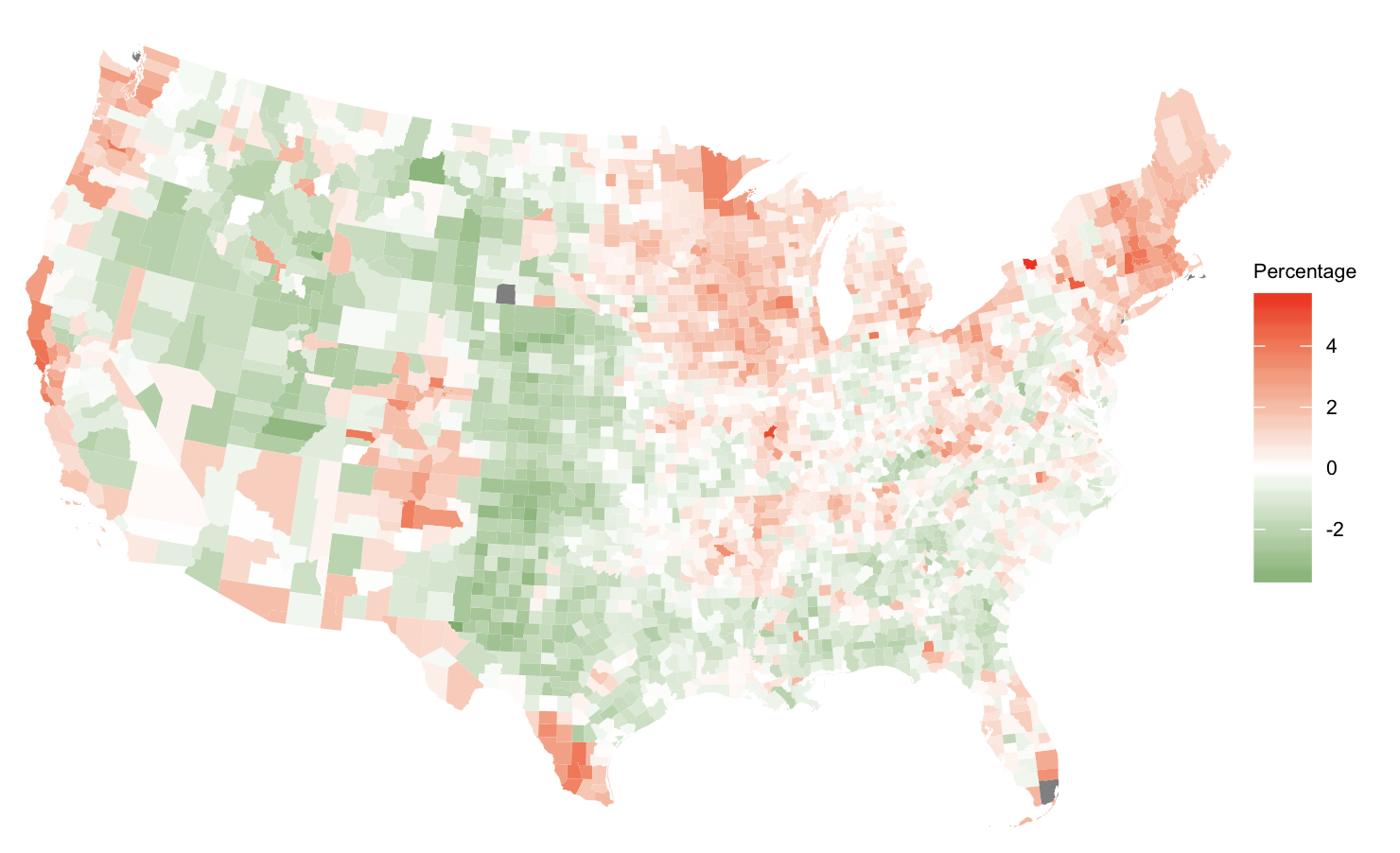}
	\caption{Residuals after fitting a linear regression model}
	\label{res_lm}
\end{figure}
\begin{figure}[h!]
	\centering
		\includegraphics[width=10cm]{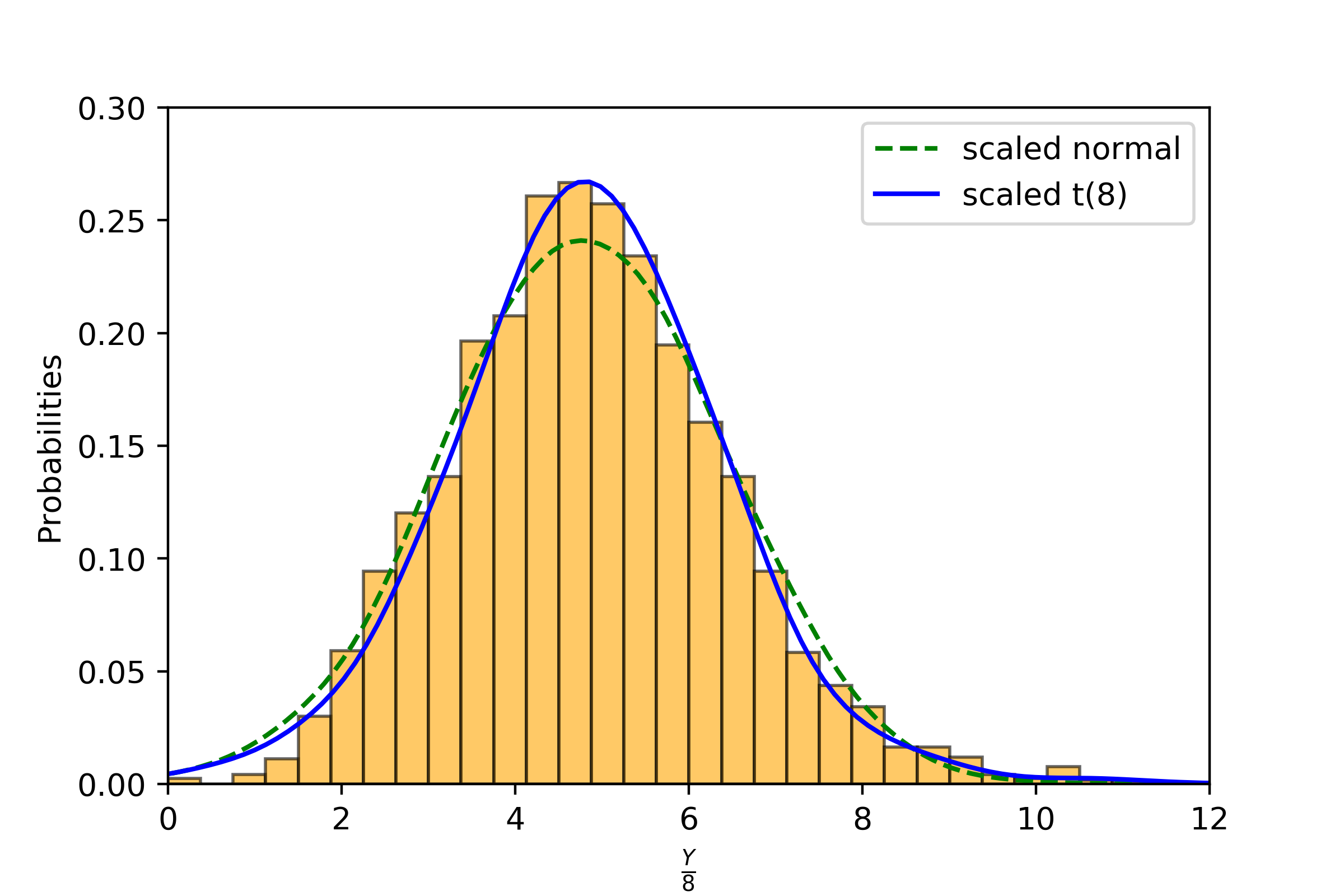}
		\caption{Histogram of (scaled by $\frac{1}{8}$) fraction of vote-shares per county for Democratic presidential candidate in 2004 overlaid
			with a scaled $t$ and a normal density curves. The mean is 4.87 and the standard deviation is 1.57.}
		\label{hist_y}
\end{figure}
\begin{figure}[h!]
	\centering
	\includegraphics[width=12cm]{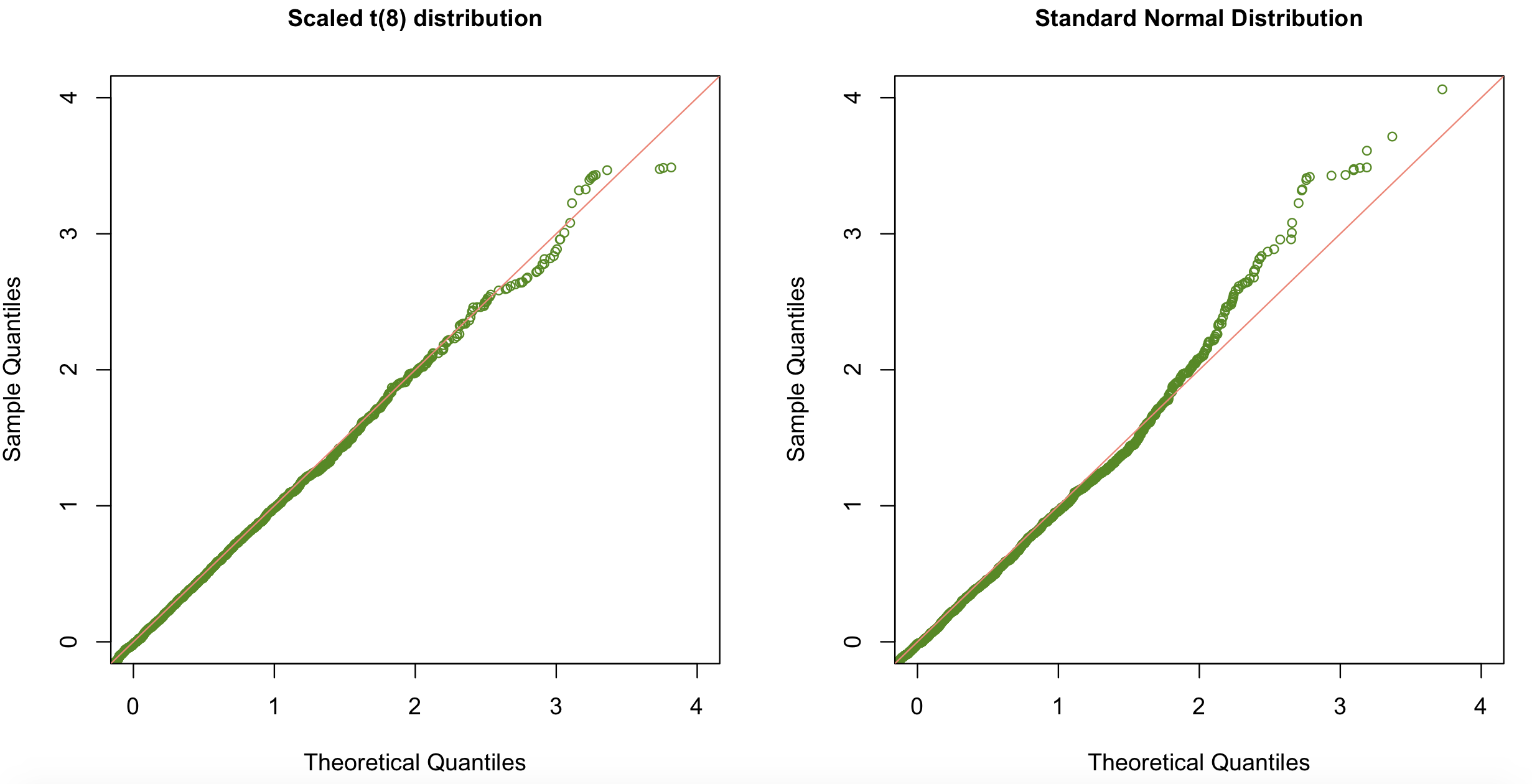}
	\caption{Q-Q plots of $Y$ versus scaled $t(8)$ and standard normal distributions: Y-axis is the sample quantiles of $Y/8$ and X-axis is the theoretical quantiles of a $\mathit{t}$-distribution (left) and a normal distribution (right).}
	\label{qqplot-yhist}
\end{figure}
Another concern is that a Gaussian estimation procedure was used; it is not most efficient when there appears to be heavy tailed errors. Figure \ref{hist_y} shows the histogram of $Y/8$ (men 4.87, standard deviation 1.57) which looks closer to a $\mathit{t}$-distribution than a normal distribution (scaled to have the same mean and standard deviation as those of $Y/8$). Figure \ref{qqplot-yhist} also demonstrates the tail distribution of $Y/8$, where the vertical axis is the sample quantiles of $Y/8$ and horizontal axis is the theoretical quantiles of scaled $\mathit{t}(8)$ and normal distribution. Clearly the observation $Y/8$ is heavy tailed. To address these, we would like to fit a spatial autoregressive model to those data and assume that the random error follows a scaled $t(8)$ distribution (scaled $\mathit{t}(8)$ has a closer density curve to $Y/8$ shown in Figure \ref{hist_y} and \ref{qqplot-yhist}). We maximized the corresponding log-likelihood function to obtain parameter estimates. For simplicity, we then only selected the three most significant variables as predictors $X$ based on the linear regression results; they are \texttt{UNDER18}, \texttt{WHITE} and \texttt{pctpoor}. The weight matrix $W_n$ is generated through a shapefile \cite{shapefile} (a geospatial vector storage format for storing geometric location and associated attribute information) using the queen criterion. 
Scatter plots of $X_1$, $X_2$ and $X_3$ versus $Y$ are shown in Figure \ref{preditor-scatterplot}.
We can clearly observe the nonlinear trend between $X_2, X_3$ and $Y$. In the linear model $X_1$ is the most significant variable but despite the linear trend, the scatter plot of \texttt{UNDER18} versus $Y$ has lots of noises around the center range from 20 to 30 percent. This may be caused by some spatial correlation in $X_1$ itself so we try despatializing $X_1$ by fitting an ordinary spatial autoregressive model $X_1 = \rho_xW_{3107}X_1+\varepsilon$. The spatial correlation of $X_1$ is estimated as $0.6$ so we define the despatialized variable $\tilde{X_1} = (I_n-0.6W_{3107})X_1$  (the scatter plot of $X_1$ in Figure \ref{preditor-scatterplot} does not show specific pattern even though this variable is significant from our preliminary analysis. So we consider de-spatializing $X_1$ and $\sum_{s=1}^{n}\hat{\varepsilon}_s^2$ of the model fitted with despatializated $X_1$ is smaller than that of the model fitted with original $X_1$). In addition, to avoid the computation overflow when maximizing the corresponding log-likelihood function, we normalized these predictors to have zero means and unit variances and also rescaled $Y$ by $\frac{1}{8}$. We conduct the following analysis using these transformed variables $Y^{\ast}$, $X^{\ast} = (X^{\ast}_1, X^{\ast}_2, X^{\ast}_3)$.
{\small\setlength{\abovedisplayskip}{3pt}
	\setlength{\belowdisplayskip}{\abovedisplayskip}
	\setlength{\abovedisplayshortskip}{0pt}
	\setlength{\belowdisplayshortskip}{3pt}
	\begin{align*}
	Y^{\ast} &= Y/8\\
	X_1^{\ast} &= \frac{\tilde{X_1}-Average(\tilde{X_1})}{Std(\tilde{X_1})}\\
	X_2^{\ast} &= \frac{X_2-Average(X_2)}{Std(X_2)}\\
	X_3^{\ast} &= \frac{X_3-Average(X_3)}{Std(X_3)}
	\end{align*}} 
\begin{figure}[h!]
	\centering
		\includegraphics[width=.3\linewidth]{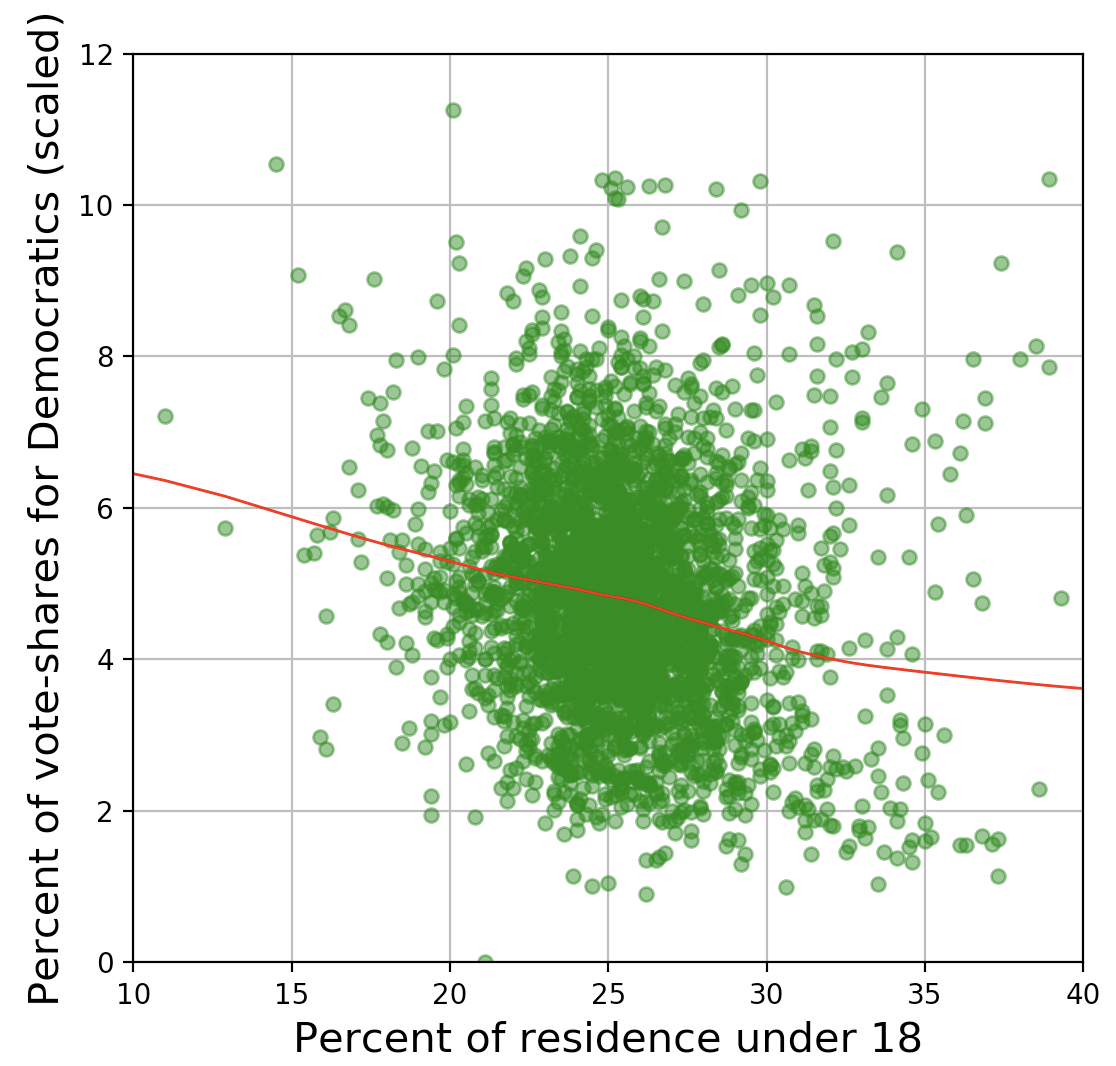}
		\includegraphics[width=.3\linewidth]{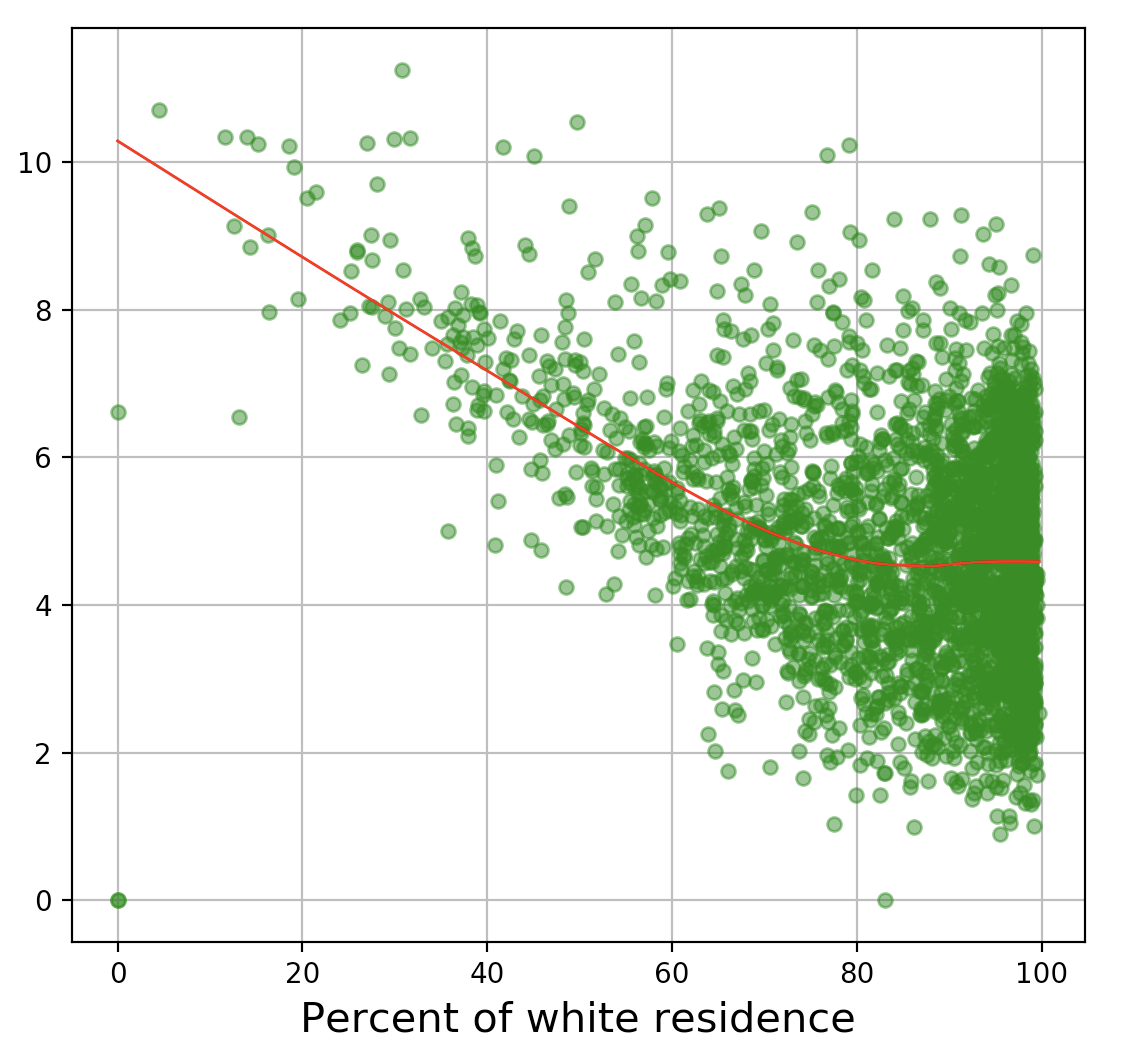}
		\includegraphics[width=.3\linewidth]{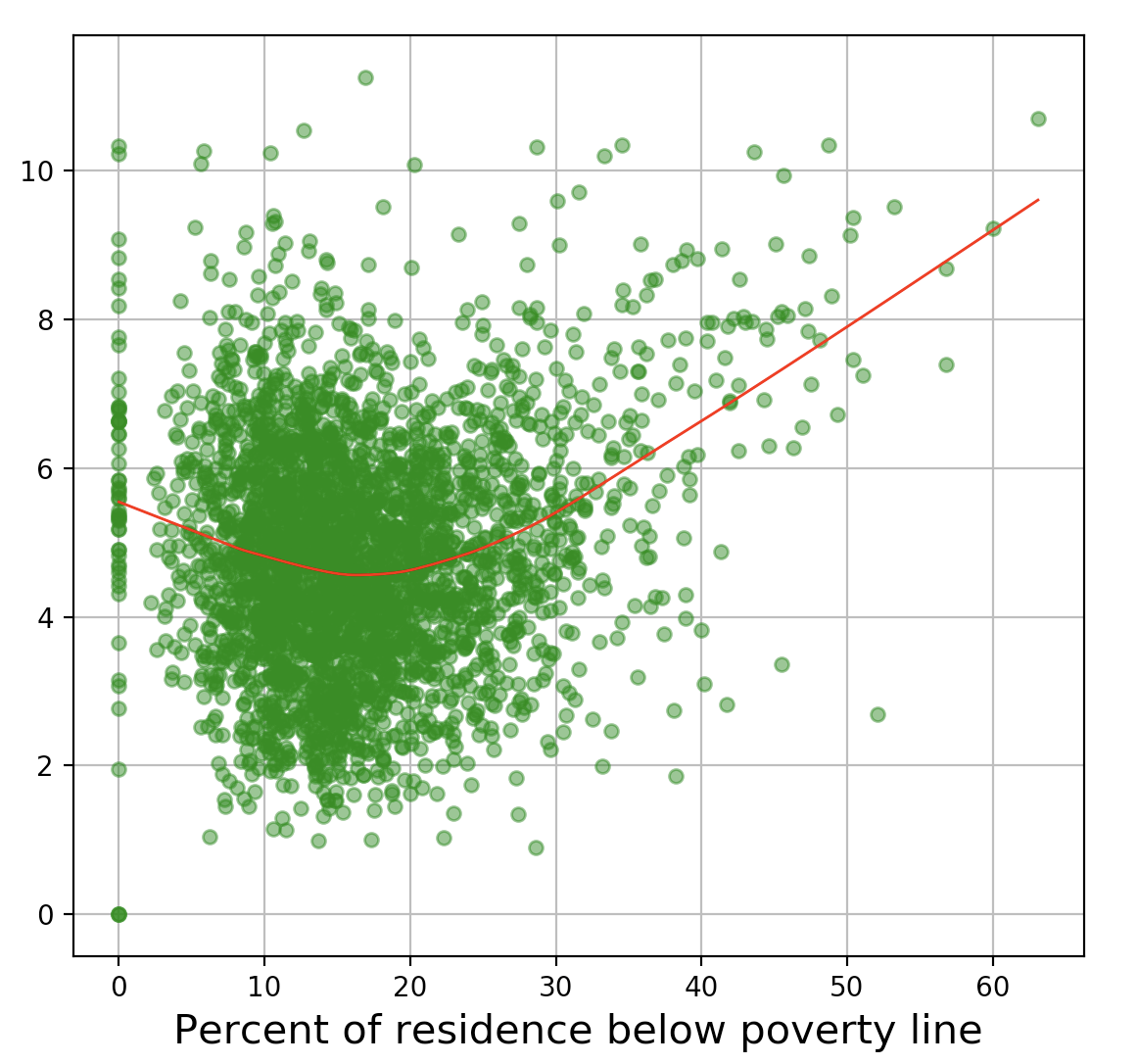}
		\caption{Scatter plots: percentage residence under 18 (left), percent white residents (middle) and percent below poverty line (right) in U.S. counties. The red lines are the Lowess smoothing curve between $Y^{\ast}$ and predictors.}
		\label{preditor-scatterplot}
\end{figure}
The first spatial model we tried is the ordinary spatial autoregressive model with $X^{\ast}_1$, $X^{\ast}_2$, $X^{\ast}_3$ and we assume the error follows a scaled $t$ distribution with $df = 8$ ($Y^{\ast} = \rho W_{3107} Y^{\ast} + X^{\ast}\beta +\varepsilon$; the fitted residual variance is about 1 when we assume the $t$ distribution with $df=8$ referring to Figure \ref{hist_y}). The model fit via maximizing log-likelihood function is shown below:
 {\small\setlength{\abovedisplayskip}{3pt}
 	\setlength{\belowdisplayskip}{\abovedisplayskip}
 	\setlength{\abovedisplayshortskip}{0pt}
 	\setlength{\belowdisplayshortskip}{3pt}
 	\begin{align}\label{psar-model}
  Y^{\ast}= 0.744W_{3107} Y^{\ast}+ 1.222
 	-0.284X^{\ast}_{1} - 0.451X^{\ast}_{2}+0.03X^{\ast}_{3} +\varepsilon
 	\end{align}} 
 The spatial correlation parameter $\hat{\rho} = 0.744$ indicates pretty high spatial dependence in $Y^{\ast}$ and the spatial dependence in the residuals is insignificant (Moran's test statistics = 1.38, P-value = 0.167). However, in Figure \ref{preditor-scatterplot}, there appears to be a nonlinear relationship between $Y$ and $X_2, X_3$. To address this, we would like to fit our proposed PSAR-ANN model to the same dataset and still we assume a $t(8)$ distributed error. The log-likelihood function in this case should be
 {\small\setlength{\abovedisplayskip}{3pt}
 	\setlength{\belowdisplayskip}{\abovedisplayskip}
 	\setlength{\abovedisplayshortskip}{0pt}
 	\setlength{\belowdisplayshortskip}{3pt}
 	\begin{align}\label{likelihood-example}
 	\mathcal{L}_{3107}(\boldsymbol{\theta}) &= \ln |I_{3017}-\rho W_{3017}|-4.5\sum_{s=1}^{3107}\ln(1+\frac{\varepsilon_s(\boldsymbol{\theta})^2}{6})\\
 	\boldsymbol{\varepsilon}(\boldsymbol{\theta}) & = (I_{3017}-\rho W_{3017})Y^{\ast}-X^{\ast}\beta-F(X^{\ast}\boldsymbol{\gamma}^{\prime})\lambda
 	\end{align}} 
Since the PSAR-ANN model has both linear and nonlinear components, we optimize these two parts iteratively to find the maximum likelihood estimators. Many optimization algorithm are sensitive to the choice of starting-values and people usually train neural network models starting at very small initial values. So especially each time instead of using the previous parameter estimates for neural network component, we always reinitialize the starting values for the neural network component and use L-BFGS-B algorithm \cite{byrd1995limited} to search the optimum. The optimization steps are outlined below:
 \begin{itemize}
 	\item Step 0: Based on some pre-knowledge about the parameters, set starting values $(\rho^{0}, \beta^{0}, \lambda^{0}, \boldsymbol{\gamma}^{0})$ and predetermine bounds for parameters in the optimization.
 	\item Step 1: In the $i$th iteration for the linear component optimization, fixing $\lambda^{i-1}$, $\boldsymbol{\gamma}^{i-1}$, use $(\rho^{i-1}, \beta^{i-1}, \lambda^{i-1}, \boldsymbol{\gamma}^{i-1})$ as starting values and apply L-BFGS-B algorithm\cite{byrd1995limited,zhu1997algorithm} to find $\rho^{(i)}$ and $\beta^{(i)}$ which maximize $\mathcal{L}_{3107}(\boldsymbol{\theta})$ in (\ref{likelihood-example}) given $\lambda^{i-1}$, $\boldsymbol{\gamma}^{i-1}$.
 	\item Step 2: In the $i$th iteration for the nonlinear component optimization, fixing $\rho^{(i)}, \beta^{(i)}$ from Step 1, randomly initialize $\lambda, \boldsymbol{\gamma}$ starting values  from a small interval $(0,0.05)$ (to avoid the computation overflow when calculating exponentials) and again use L-BFGS-B algorithm \cite{byrd1995limited} to find $\lambda^{(i)}, \boldsymbol{\gamma}^{(i)}$ which maximize $\mathcal{L}_{3107}(\boldsymbol{\theta})$ in (\ref{likelihood-example}) given $\rho^{(i)}, \beta^{(i)}$.
    \item Step 3: Repeat Step 1, 2 until the difference of the corresponding log-likelihood function values in Step 1 and 2 is smaller than some threshold value (for example $10^{-2}$).
 \end{itemize}
The following is the estimated PSAR-ANN model
{\small\setlength{\abovedisplayskip}{3pt}
	\setlength{\belowdisplayskip}{\abovedisplayskip}
	\setlength{\abovedisplayshortskip}{0pt}
	\setlength{\belowdisplayshortskip}{3pt}
	\begin{align}\label{psar-ann-model}\nonumber
	Y^{\ast} =&0.721W_{3107}Y^{\ast} + 1.693
	-0.185X_1^{\ast}- 0.658X_2^{\ast}+0.181X_3^{\ast}\\
	&-0.937F(1.509X_1^{\ast} -2.544X_2^{\ast}+2.268X_3^{\ast})+\hat{\boldsymbol{\varepsilon}}
	\end{align}} 
In model (\ref{psar-ann-model}), the correlation estimate is roughly the same as the model (\ref{psar-model}) indicating that people in neighboring counties tend to have similar voting preferences. The Moran's test statistic for the residuals is 1.78 with P-value =  0.0745. We compare the SAR model with our proposed model and find that even though the new model has more parameters, it has lower AIC (AIC $= 2(\#\text{parameters})-2\ln L_n(\hat{\boldsymbol{\theta}})$) compared to the original spatial autoregressive model (See table \ref{aic}). Through likelihood ratio test ($\mathcal{H}_0$: SAR model is adequate, $\mathcal{H}_1$: PSAR-ANN model is adequate), the test statistic $-2\ln L_{SAR}+2\ln L_{PSAR-ANN} = 157.45$ with $df = 4$, P-value $< 0.05$, so we rejected $\mathcal{H}_0$ and conclude that the PSAR-ANN model is a better fit than SAR model. Figure \ref{res_nn} shows the residuals (of PSAR-ANN model) heat map and its histogram. Through the residual histogram, assuming the error density as a standardized $t(8)$ ($df$ is chosen by the shape of the residual histogram) seems to be more appropriate than a standard normal distribution.
\begin{table}
	\centering
	\begin{tabular}{ccc}
		\toprule
		& SAR & PSAR-ANN  \\\midrule
		$\#$ Parameters& $5$ & $9$\\
		\multirow{2}{*}{Moran's Test}
		& $0.167$& $0.0745$ \\ [-2pt]
		& $(1.3808)$ & $(1.7836)$\\
		$-\ln L$& $1958.08$ & $1879.35$\\
		$AIC$&$3926.16$& $3776.17$ \\\bottomrule     
	\end{tabular}
	\caption{Comparison of SAR and PSAR-ANN model by $\#$ parameters, Moran's test P-value (test statistics), $-\ln L$ and $AIC$}
\label{aic}
\end{table}
\begin{figure}[h!]
	\centering
	\includegraphics[width=0.48\linewidth]{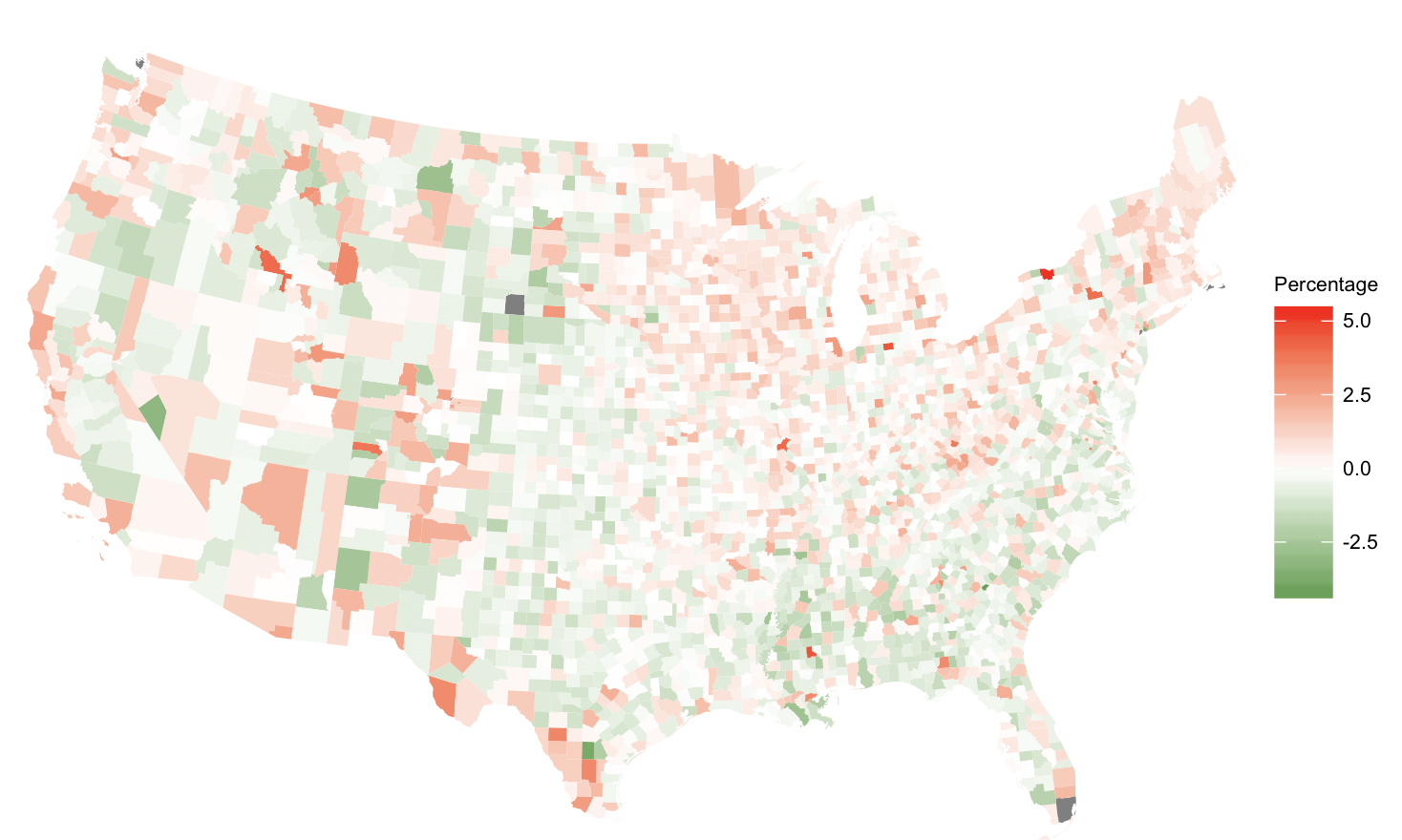}
	\includegraphics[width=0.45\linewidth]{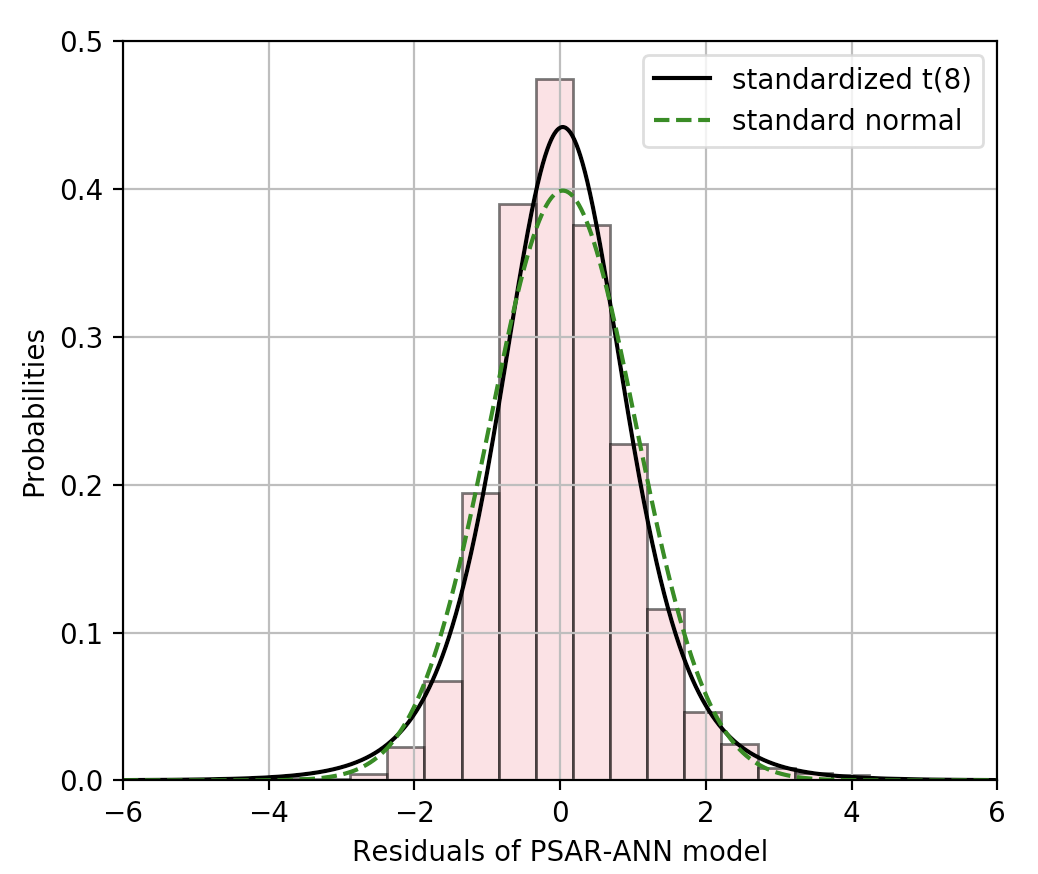}
	\caption{Heat map (left)  and histogram (right) of residuals of the PSAR-ANN model}
	\label{res_nn}
\end{figure}

The covariance matrix for parameter estimates in model (\ref{psar-ann-model}) is calculated and the 95\% confidence intervals for the model parameters are shown in table \ref{CI}. From the table, all the parameters are significant at $95\%$ significance level. Looking at the signs of parameter estimates, we can learn that a county with more young residence under 18 and white residence is more likely to support Republicans while people struggling to make ends meet are prone to support the Democrat. This opposite effect can also be observed in Figure \ref{preditor-scatterplot}. The neural network component in our model helps to capture  the nonlinear relationship between $X$ and $Y$. Parameter $\lambda$ is significant so the nonlinear component is appreciable in modeling and $\boldsymbol{\gamma}$'s are all significant at $95\%$ confidence level. Figure \ref{nn-scatter-plot} shows scatter plots of $X_1^{\ast},X_2^{\ast}, X_3^{\ast}$, where points are colored by the value of the fitted neural network component $-0.937F(1.509X_1^{\ast}-2.544X_2^{\ast}+2.268X_3^{\ast})$. Observations with green color are counties tending to have more voters for the Democratic candidate while the red points represent counties tending to have more voters for the Republican candidate. From the distribution of these colored points, it appears that counties with more people below poverty line and less white residence tend to have more Democratic voters. On the other hand, voters in counties with more children and higher percent white residence tend to be less likely to vote Democratic. These findings also correspond to the trend we can find in the linear component but they are presented in a non-linear way.
\begin{figure}[h!]
	\centering
	\includegraphics[width=0.9\linewidth]{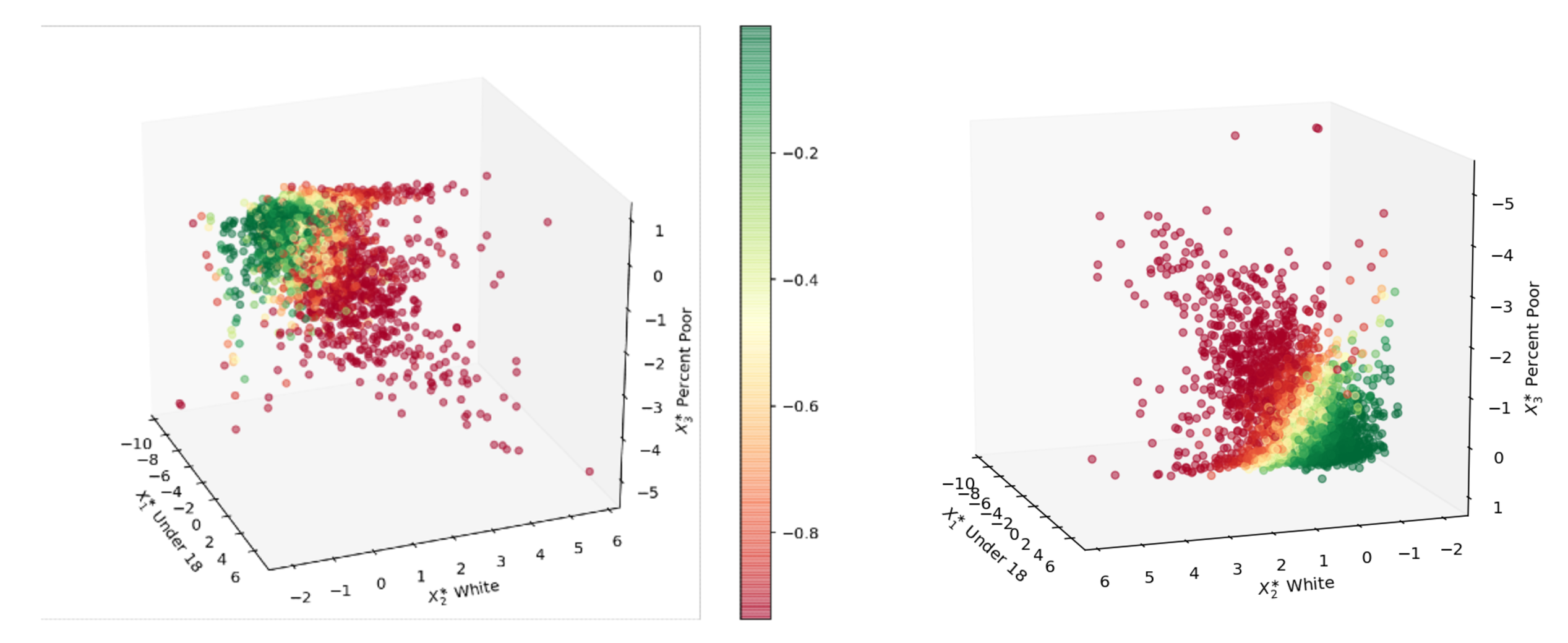}
	\caption{Scatter plot of $X_1^{\ast},X_2^{\ast}, X_3^{\ast}$ colored by the output of fitted neural network component $-0.937F(1.509X_1^{\ast}-2.544X_2^{\ast}+2.268X_3^{\ast})$.}
	\label{nn-scatter-plot}
\end{figure}

To conclude, our proposed model PSAR-ANN appears to successfully capture some spatial election dynamics. It allows for non-Gaussian random errors and is flexible in learning nonlinear relationships between the response and exogenous variables.
\setlength{\extrarowheight}{5pt}
\begin{table}[h!]
	\centering
	\begin{tabular}{cccc}
		\toprule
		Parameter& Estimate&Std.&$95\%$ C.I. \\ \midrule
		$\rho$ & 0.721 & 0.0102&  $(0.7010,0.7410)$ \\ [-5pt]
		$\beta_0$& 1.693& 0.0573& $(1.5807,1.8053)$\\ [-5pt]
		$\beta_1$& -0.185 & 0.0219 & $(-0.2279,-0.1421)$\\[-5pt]
		$\beta_2$& -0.658 & 0.0288 & $(-0.7144,-0.6016)$\\ [-5pt]
		$\beta_3$& 0.181 & 0.0243 & $(0.1334,0.2286)$\\
		$\lambda$& -0.937& 0.0581& $(-1.0464,-0.8276)$\\ [-5pt]
	    $\gamma_1$& 1.509& 0.0239& $(1.4622,1.5558)$\\ [-5pt]
	    $\gamma_2$& -2.544& 0.0137& $(-2.5709,-2.5171)$\\ [-5pt]
	    $\gamma_3$& 2.268& 0.0157& $(2.2372,2.2988)$\\
	  \bottomrule       
	\end{tabular}
\caption{Parameter estimates of PSAR-ANN model with $95\%$ confidence intervals}
\label{CI}
\end{table}

\small{
\singlespacing
\bibliography{references}

\begin{thebibliography}{37}
\providecommand{\natexlab}[1]{#1}
\providecommand{\url}[1]{\texttt{#1}}
\expandafter\ifx\csname urlstyle\endcsname\relax
  \providecommand{\doi}[1]{doi: #1}\else
  \providecommand{\doi}{doi: \begingroup \urlstyle{rm}\Url}\fi

\bibitem[Ai and Chen(2003)]{ai2003efficient}
Chunrong Ai and Xiaohong Chen.
\newblock Efficient estimation of models with conditional moment restrictions
  containing unknown functions.
\newblock \emph{Econometrica}, 71\penalty0 (6):\penalty0 1795--1843, 2003.

\bibitem[Andrews et~al.(2006)Andrews, Davis, and Breidt]{andrews2006maximum}
Beth Andrews, Richard~A Davis, and F~Jay Breidt.
\newblock Maximum likelihood estimation for all-pass time series models.
\newblock \emph{Journal of Multivariate Analysis}, 97\penalty0 (7):\penalty0
  1638--1659, 2006.

\bibitem[Anselin(2013)]{anselin2013spatial}
Luc Anselin.
\newblock \emph{Spatial econometrics: methods and models}, volume~4.
\newblock Springer Science \& Business Media, 2013.

\bibitem[Birkhoff(1931)]{birkhoff1931proof}
George~D Birkhoff.
\newblock Proof of the ergodic theorem.
\newblock \emph{Proceedings of the National Academy of Sciences}, 17\penalty0
  (12):\penalty0 656--660, 1931.

\bibitem[Bivand et~al.(2008)Bivand, Pebesma, Gomez-Rubio, and
  Pebesma]{bivand2008applied}
Roger~S Bivand, Edzer~J Pebesma, Virgilio Gomez-Rubio, and Edzer~Jan Pebesma.
\newblock \emph{Applied spatial data analysis with R}, volume 747248717.
\newblock Springer, 2008.

\bibitem[Braha and de~Aguiar(2017)]{votingcontagion}
Dan Braha and Marcus A.~M. de~Aguiar.
\newblock Voting contagion: Modeling and analysis of a century of u.s.
  presidential elections.
\newblock \emph{PLOS ONE}, 12\penalty0 (5):\penalty0 1--30, 05 2017.
\newblock \doi{10.1371/journal.pone.0177970}.
\newblock URL \url{https://doi.org/10.1371/journal.pone.0177970}.

\bibitem[Breid et~al.(1991)Breid, Davis, Lh, and Rosenblatt]{breid1991maximum}
F~Jay Breid, Richard~A Davis, Keh-Shin Lh, and Murray Rosenblatt.
\newblock Maximum likelihood estimation for noncausal autoregressive processes.
\newblock \emph{Journal of Multivariate Analysis}, 36\penalty0 (2):\penalty0
  175--198, 1991.

\bibitem[Byrd et~al.(1995)Byrd, Lu, Nocedal, and Zhu]{byrd1995limited}
Richard~H Byrd, Peihuang Lu, Jorge Nocedal, and Ciyou Zhu.
\newblock A limited memory algorithm for bound constrained optimization.
\newblock \emph{SIAM Journal on Scientific Computing}, 16\penalty0
  (5):\penalty0 1190--1208, 1995.

\bibitem[Cressie(2015)]{cressie2015statistics}
Noel Cressie.
\newblock \emph{Statistics for spatial data}.
\newblock John Wiley \& Sons, 2015.

\bibitem[Gallant and White(1992)]{gallant1992learning}
A~Ronald Gallant and Halbert White.
\newblock On learning the derivatives of an unknown mapping with multilayer
  feedforward networks.
\newblock \emph{Neural Networks}, 5\penalty0 (1):\penalty0 129--138, 1992.

\bibitem[Gershgorin(1931)]{gershgorin1931uber}
Semyon~Aranovich Gershgorin.
\newblock Uber die abgrenzung der eigenwerte einer matrix.
\newblock \emph{Bulletin de l'Acad\'emie des Sciences de l'URSS. Classe des
  sciences math\'ematiques et na}, \penalty0 (6):\penalty0 749--754, 1931.

\bibitem[Getis(1995)]{getis1995cliff}
Arthur Getis.
\newblock Cliff, ad and ord, jk 1973: Spatial autocorrelation. london: Pion.
\newblock \emph{Progress in Human Geography}, 19\penalty0 (2):\penalty0
  245--249, 1995.

\bibitem[Hornik et~al.(1990)Hornik, Stinchcombe, and
  White]{hornik1990universal}
Kurt Hornik, Maxwell Stinchcombe, and Halbert White.
\newblock Universal approximation of an unknown mapping and its derivatives
  using multilayer feedforward networks.
\newblock \emph{Neural networks}, 3\penalty0 (5):\penalty0 551--560, 1990.

\bibitem[Hwang and Ding(1997)]{hwang1997prediction}
JT~Gene Hwang and A~Adam Ding.
\newblock Prediction intervals for artificial neural networks.
\newblock \emph{Journal of the American Statistical Association}, 92\penalty0
  (438):\penalty0 748--757, 1997.

\bibitem[Kelejian and Prucha(1998)]{kelejian1998generalized}
Harry~H Kelejian and Ingmar~R Prucha.
\newblock A generalized spatial two-stage least squares procedure for
  estimating a spatial autoregressive model with autoregressive disturbances.
\newblock \emph{The Journal of Real Estate Finance and Economics}, 17\penalty0
  (1):\penalty0 99--121, 1998.

\bibitem[Kelejian and Prucha(1999)]{kelejian1999generalized}
Harry~H Kelejian and Ingmar~R Prucha.
\newblock A generalized moments estimator for the autoregressive parameter in a
  spatial model.
\newblock \emph{International economic review}, 40\penalty0 (2):\penalty0
  509--533, 1999.

\bibitem[Kelejian and Prucha(2010)]{kelejian2010specification}
Harry~H Kelejian and Ingmar~R Prucha.
\newblock Specification and estimation of spatial autoregressive models with
  autoregressive and heteroskedastic disturbances.
\newblock \emph{Journal of Econometrics}, 157\penalty0 (1):\penalty0 53--67,
  2010.

\bibitem[Lee(2004)]{lee2004asymptotic}
Lung-Fei Lee.
\newblock Asymptotic distributions of quasi-maximum likelihood estimators for
  spatial autoregressive models.
\newblock \emph{Econometrica}, 72\penalty0 (6):\penalty0 1899--1925, 2004.

\bibitem[LeSage et~al.(2009)]{lesage2009introduction}
J~Pace LeSage et~al.
\newblock Introduction to spatial econometrics.
\newblock Technical report, CRC Press, 2009.

\bibitem[LeSage et~al.(1999)]{lesage1999spatial}
James~P LeSage et~al.
\newblock Spatial econometrics, 1999.

\bibitem[Lii and Rosenblatt(1992)]{lii1992approximate}
Keh-Shin Lii and Murray Rosenblatt.
\newblock An approximate maximum likelihood estimation for non-gaussian
  non-minimum phase moving average processes.
\newblock \emph{Journal of Multivariate Analysis}, 43\penalty0 (2):\penalty0
  272--299, 1992.

\bibitem[Medeiros et~al.(2006)Medeiros, Ter{\"a}svirta, and
  Rech]{medeiros2006building}
Marcelo~C Medeiros, Timo Ter{\"a}svirta, and Gianluigi Rech.
\newblock Building neural network models for time series: a statistical
  approach.
\newblock \emph{Journal of Forecasting}, 25\penalty0 (1):\penalty0 49--75,
  2006.

\bibitem[Ord(1975)]{ord1975estimation}
Keith Ord.
\newblock Estimation methods for models of spatial interaction.
\newblock \emph{Journal of the American Statistical Association}, 70\penalty0
  (349):\penalty0 120--126, 1975.

\bibitem[Paelinck(1978)]{paelinck1978spatial}
J~Paelinck.
\newblock Spatial econometrics.
\newblock \emph{Economics Letters}, 1\penalty0 (1):\penalty0 59--63, 1978.

\bibitem[Poole and Rosenthal(1984)]{poole1984us}
Keith~T Poole and Howard Rosenthal.
\newblock U.s. presidential elections 1968-80: A spatial analysis.
\newblock \emph{American Journal of Political Science}, pages 282--312, 1984.

\bibitem[Rothenberg(1971)]{rothenberg1971identification}
Thomas~J Rothenberg.
\newblock Identification in parametric models.
\newblock \emph{Econometrica: Journal of the Econometric Society}, pages
  577--591, 1971.

\bibitem[Smirnov and Anselin(2001)]{smirnov2001fast}
Oleg Smirnov and Luc Anselin.
\newblock Fast maximum likelihood estimation of very large spatial
  autoregressive models: a characteristic polynomial approach.
\newblock \emph{Computational Statistics \& Data Analysis}, 35\penalty0
  (3):\penalty0 301--319, 2001.

\bibitem[Stabler(2013)]{shapefile}
Ben Stabler.
\newblock \emph{shapefiles: Read and Write ESRI Shapefiles}, 2013.
\newblock URL \url{https://CRAN.R-project.org/package=shapefiles}.
\newblock R package version 0.7.

\bibitem[Steele(2004)]{steele2004cauchy}
J~Michael Steele.
\newblock \emph{The Cauchy-Schwarz master class: an introduction to the art of
  mathematical inequalities}.
\newblock Cambridge University Press, 2004.

\bibitem[Su and Jin(2010)]{su2010profile}
Liangjun Su and Sainan Jin.
\newblock Profile quasi-maximum likelihood estimation of partially linear
  spatial autoregressive models.
\newblock \emph{Journal of Econometrics}, 157\penalty0 (1):\penalty0 18--33,
  2010.

\bibitem[Taussky(1949)]{taussky1949recurring}
Olga Taussky.
\newblock A recurring theorem on determinants.
\newblock \emph{The American Mathematical Monthly}, 56\penalty0
  (10P1):\penalty0 672--676, 1949.

\bibitem[Viton(2010)]{viton2010notes}
Philip~A Viton.
\newblock Notes on spatial econometric models.
\newblock \emph{City and regional planning}, 870\penalty0 (03):\penalty0 9--10,
  2010.

\bibitem[White(1994)]{white1994parametric}
Halbert White.
\newblock Parametric statistical estimation with artificial neural networks: A
  condensed discussion.
\newblock \emph{From statistics to neural networks: theory and pattern
  recognition applications}, 136:\penalty0 127, 1994.

\bibitem[White(1996)]{white1996estimation}
Halbert White.
\newblock \emph{Estimation, inference and specification analysis}.
\newblock Number~22. Cambridge university press, 1996.

\bibitem[Yao et~al.(2006)Yao, Brockwell, et~al.]{yao2006gaussian}
Qiwei Yao, Peter~J Brockwell, et~al.
\newblock Gaussian maximum likelihood estimation for arma models ii: spatial
  processes.
\newblock \emph{Bernoulli}, 12\penalty0 (3):\penalty0 403--429, 2006.

\bibitem[Zhang and Yang(2015)]{zhang2015statistical}
Yuan-qing Zhang and Guang-ren Yang.
\newblock Statistical inference of partially specified spatial autoregressive
  model.
\newblock \emph{Acta Mathematicae Applicatae Sinica, English Series},
  31\penalty0 (1):\penalty0 1--16, 2015.

\bibitem[Zhu et~al.(1997)Zhu, Byrd, Lu, and Nocedal]{zhu1997algorithm}
Ciyou Zhu, Richard~H Byrd, Peihuang Lu, and Jorge Nocedal.
\newblock Algorithm 778: L-bfgs-b: Fortran subroutines for large-scale
  bound-constrained optimization.
\newblock \emph{ACM Transactions on Mathematical Software (TOMS)}, 23\penalty0
  (4):\penalty0 550--560, 1997.

\end{thebibliography}
}
\end{document}